\documentclass{lmcs}
\pdfoutput=1
\usepackage[utf8]{inputenc}

\usepackage{lastpage}
\lmcsdoi{20}{2}{5}
\lmcsheading{}{\pageref{LastPage}}{}{}%
{Dec.~05,~2022}{Apr.~30,~2024}{}


\usepackage{hyperref}
\usepackage{amssymb,amsmath,mathtools,stmaryrd,mathrsfs}
\usepackage{extarrows}
\usepackage{xspace}

\usepackage{LMCS_macros-encoding_CQP_qCCS}

\newlist{compactitem}{itemize}{2}   
\setlist[compactitem,1]{label=\textbullet,labelindent=*,leftmargin=*,nosep}
\setlist[compactitem,2]{label=$-$,labelindent=*,leftmargin=*,nosep}

\newlist{compactenum}{enumerate}{2}
\setlist[compactenum,1]{font=\normalfont,label=(\arabic*),labelindent=*,leftmargin=*,start=1,nosep}
\setlist[compactenum,2]{font=\normalfont,label=(\alph*),labelindent=*,leftmargin=*,start=1,nosep}

\newlist{compactdesc}{description}{1}
\setlist[compactdesc,1]{labelindent=0mm,labelwidth=*,labelsep=2mm,leftmargin=1cm,nosep}

\begin{document}


\title{Encodability Criteria for Quantum Based Systems}
\thanks{We thank the anonymous reviewers for their constructive feedback and help to improve this paper.}

\author[A.~Schmitt]{Anna Schmitt\lmcsorcid{0000-0001-6675-2879}}[a]
\author[K.~Peters]{Kirstin Peters\lmcsorcid{0000-0002-4281-0074}}[b]
\author[Y.~Deng]{Yuxin Deng\lmcsorcid{0000-0003-0753-418X}}[c]

\address{TU Darmstadt, Germany}
\email{Anna.Schmitt@tu-darmstadt.de}

\address{Augsburg University, Germany}
\email{kirstin.peters@uni-a.de}

\address{East China Normal University, Shanghai}
\email{yxdeng@sei.ecnu.edu.cn}


\begin{abstract}
	Quantum based systems are a relatively new research area for that different modelling languages including process calculi are currently under development. Encodings are often used to compare process calculi. Quality criteria are used then to rule out trivial or meaningless encodings. In this new context of quantum based systems, it is necessary to analyse the applicability of these quality criteria and to potentially extend or adapt them.
	As a first step, we test the suitability of classical criteria for encodings between quantum based languages and discuss new criteria.
	
	Concretely, we present an encoding, from a language inspired by \CQP into a language inspired by \qCCS.
	We show that this encoding satisfies compositionality, name invariance (for channel and qubit names), operational correspondence, divergence reflection, success sensitiveness, and that it preserves the size of quantum registers.
	Then we show that there is no encoding from \qCCS into \CQP that is compositional, operationally corresponding, and success sensitive.
\end{abstract}

\maketitle


\section{Introduction}
\label{sec:introduction}

The technological progress turns quantum based systems from theoretical models to hopefully soon practicable realisations. This progress inspired research on quantum algorithms and protocols.
They allow for a significant increase in efficiency in many cases and provide new approaches to secure systems.
These algorithms and protocols in turn call for verification methods that can deal with the new quantum based setting.

Among the various tools for such verifications, also several process calculi for quantum based systems are developed \cite{JorrandLalire04,GayNagarajan05,gay06,YingFengDuanJi09}.
To compare the expressive power and suitability for different application areas, encodings have been widely used for classical, \ie not quantum based, systems.
To rule out trivial or meaningless encodings, they are required to satisfy quality criteria.
In this new context of quantum based systems, we have to analyse the applicability of these quality criteria and potentially extend or adapt them.

Therefore, we start by considering a well-known framework of quality criteria introduced by Gorla in \cite{gorla10} for the classical setting.
As a case study we want to compare \emph{Communicating Quantum Processes} (\CQP) introduced in \cite{GayNagarajan05} and the \emph{Algebra of Quantum Processes} (\qCCS) introduced in \cite{FengDuanJiYing07,YingFengDuanJi09}.
These two process calculi are particularly interesting, because they model quantum registers and the behaviour of quantum based systems in fundamentally different ways.
\CQP considers closed systems, where qubits are manipulated by unitary transformations and the behaviour is expressed by a probabilistic transition system.
In contrast, \qCCS focuses on open systems and super-operators.
Moreover, the transition system of \qCCS as presented at \cite{YingFengDuanJi09} is non-probabilistic.
(Unitary transformations and super-operators are discussed in the next section.)

Unfortunately, the languages also differ in classical aspects: \CQP has $ \pi $-calculus-like name passing but the CCS based \qCCS does not allow to transfer names; \qCCS has operators for choice and recursion but \CQP in \cite{GayNagarajan05} has not.
Therefore, comparing the languages directly would yield negative results in both directions, that do not depend on their treatment of qubits.
To avoid these obvious negative results and to concentrate on the treatment of qubits, we consider \CQS, a strictly less expressive sublanguage of \CQP that removes name passing and simplifies the syntax/semantics, but as we claim does treat qubits in the same way as \CQP.
As second language we consider \OQS that is similar to \qCCS as presented in \cite{YingFengDuanJi09} extended by an operator for a conditional, but as we claim again does treat qubits in the same way as \qCCS.
Accordingly, our focus is not exactly on the languages \CQP and \qCCS but on how they treat qubits.
The language \CQS, for \emph{closed quantum systems}, inherits from \CQP the closed systems with only unitary transformations and has a semantics that is no longer probabilistic, but explicitly deals with probability distributions.
In contrast \OQS, for \emph{open quantum systems}, inherits from \qCCS the open systems and super-operators and a non-probabilistic semantics without explicitly considering probability distributions.
We further discuss the differences between \CQP and \CQS as well as \qCCS and \OQS when we introduce these languages.

We then show that there exists an encoding from \CQS into \OQS that satisfies the quality criteria of Gorla and thereby that the treatment of qubits in \OQS/\qCCS is strong enough to emulate the treatment of qubits in \CQS/\CQP.
We also show that the opposite direction is more difficult, even if we restrict the classical operators in \qCCS.
In fact, the counterexample that we use to prove the non-existence of an encoding considers the treatment of qubits only, \ie relies on the application of a specific super-operator that has no unitary equivalent.

These two results show that the quality criteria can still be applied in the context of quantum based systems and are still meaningful in this setting.
They may, however, not be exhaustive.
Therefore, we discuss directions of additional quality criteria that might be relevant for quantum based systems.

Our encoding satisfies compositionality, name invariance \wrt channel names and qubit names, strong operational correspondence, divergence reflection, success sensitiveness, and that the encoding preserves the size of quantum registers.
We also show that there is no encoding from \OQS/\qCCS into \CQS/\CQP that satisfies compositionality, operational correspondence, and success sensitiveness, where we consider a variant \qCCS with a measurement operator as given in \cite{FengDuanJiYing07,FengDuanYing12}.

\paragraph{Summary.}
We need a number of preliminaries:
Quantum based systems are briefly discussed in §\ref{sec:qbs}, the considered process calculi are introduced in §\ref{sec:processCalculi}, and §\ref{sec:criteria} presents the quality criteria of Gorla.
§\ref{sec:encoding} introduces the encoding from \CQS into \OQS and comments on its correctness.
The negative result from \OQS/\qCCS with a measurement operator into \CQS/\CQP is presented in §\ref{sec:separation}.
In §\ref{sec:criteriaQBS} we discuss directions for criteria specific to quantum based systems.
We conclude in §\ref{sec:conclusions}.
The present work extends and revises \cite{schmittPetersDeng22,schmittPetersDengTec22}.
In particular, we restore the negative result in §\ref{sec:separation}, since unfortunately the counterexample used in \cite{schmittPetersDeng22} was an invalid super-operator.
Moreover, we revise both of the considered languages to get closer to the original versions of \qCCS and \CQP and more clearly describe the differences to their respective prototypes.
We present detailed proofs of the mentioned results and provide more explanations.


\section{Quantum Based Systems}
\label{sec:qbs}

We briefly introduce the aspects of quantum based systems, which are needed for the rest of this paper.
For more details, we refer to the books by Nielsen and Chuang \cite{NielsenChuang10}, Gruska \cite{gruska09}, and Rieffel and Polak \cite{RieffelPolak00}.

A \emph{quantum bit} or \emph{qubit} is a physical system which has the two base states: $\Ket{0}$ and $\Ket{1}$. These states correspond to one-bit classical values.
The general state of a quantum system is a \emph{superposition} or linear combination of base states, concretely $\Ket{\psi} = \alpha \Ket{0} + \beta \Ket{1}$. Thereby, $\alpha$ and $\beta$ are complex numbers such that $\Length{\alpha}^2 + \Length{\beta}^2 = 1$, \eg $ \Ket{0} = 1\Ket{0} + 0\Ket{1} $.
Further, a state can be represented by column vectors $ \Ket{\psi} = \begin{pmatrix} \alpha\\ \beta \end{pmatrix} = \alpha \Ket{0} + \beta \Ket{1} $, which sometimes for readability will be written in the format $ \Transp{\Tuple{\alpha, \beta}} $, where $ \Transp{} $ stands for transpose.
The vector space of these vectors is a \emph{Hilbert space}, denoted by $ \hilbert $.
It forms the state space of a quantum based system.
In \cite{YingFengDuanJi09} finite-dimensional and countably infinite-dimensional Hilbert spaces are considered, where the latter are treated as tensor products of countably infinitely many finite-dimensional Hilbert spaces.
For this work finite-dimensional Hilbert spaces are sufficient.

The basis $ \Set{\Ket{0}, \Ket{1}} $ is called \emph{standard basis} or \emph{computational basis}, but sometimes there are other orthonormal bases of interest, especially the \emph{diagonal} or \emph{Hadamard} basis consisting of the vectors $ \Ket{+} = \frac{1}{\sqrt{2}}\Tuple{\Ket{0} + \Ket{1}} $ and $ \Ket{-} = \frac{1}{\sqrt{2}}\Tuple{\Ket{0} - \Ket{1}} $. We assume the standard basis in the following.

The evolution of a closed quantum system can be described by \emph{unitary transformations} \cite{NielsenChuang10}. A unitary transformation $U$ is represented by a complex-valued matrix such that the effect of $U$ onto a state of a qubit is calculated by matrix multiplication.
It holds that $ \Adjoint{U}U = \opI$, where $ \Adjoint{U} $ is the adjoint of $U$ and $\opI$ is the \emph{identity matrix}. Thereby, $\opI$ is one of the \emph{Pauli matrices} together with $ \opX $, $ \opY $, and $ \opZ $.
Another important unitary transformation is the \emph{Hadamard} transformation $\opH$, as it creates the superpositions $\opH\Ket{0} = \Ket{+}$ and $\opH\Ket{1} = \Ket{-}$.
\begin{align*}
	\opI =
	\begin{pmatrix}
		1 & 0 \\
		0 & 1
	\end{pmatrix} \quad
	\opX =
	\begin{pmatrix}
		0 & 1 \\
		1 & 0
	\end{pmatrix} \quad
	\opY =
	\begin{pmatrix}
		0 & -i \\
		i & 0
	\end{pmatrix} \quad
	\opZ =
	\begin{pmatrix}
		1 & 0 \\
		0 & -1
	\end{pmatrix} \quad
	\opH = \dfrac{1}{\sqrt{2}}
	\begin{pmatrix}
		1 & 1 \\
		1 & -1
	\end{pmatrix}
\end{align*}
All of these five unitary transformations are applied to a single qubit.
As mentioned above, $ \opI $ is identity.
$ \opX $ performs the quantum version of a bit-flip.
It interchanges the amplitudes, \ie $ \opX \Transp{\Tuple{\alpha, \beta}} = \Transp{\Tuple{\beta, \alpha}} $.
Intuitively, $ \opY $ moves a qubit by the imaginary $ i $, \ie $ \opY\Transp{\Tuple{\alpha, \beta}} = \Transp{\Tuple{-i \beta,  i \alpha}} $.
The transformation $ \opZ $, that is sometimes called phase flip, leaves the upper component of the vector unchanged but flips the sign of the second component, \ie $ \opZ \Transp{\Tuple{\alpha, \beta}} = \Transp{\Tuple{\alpha, -\beta}} $.
Hadamard $ \opH $ intuitively moves a qubit halfway between the base states $ \Ket{0} $ and $ \Ket{1} $, \eg $ \opH\Ket{0} = \opH \Transp{\Tuple{1, 0}} = \Ket{+} = \frac{1}{\sqrt{2}}\Tuple{\Ket{0} + \Ket{1}} $ and $ \opH\Ket{+} = \Ket{0} $.

Another key feature of quantum computing is the \emph{measurement}. Measuring a qubit $q$ in state $\Ket{\psi} = \alpha\Ket{0} + \beta\Ket{1}$ results in $0$ (leaving it in $\Ket{0}$) with probability $\Length{\alpha}^2$ and in $1$ (leaving it in $\Ket{1}$) with probability $\Length{\beta}^2$.

By combining qubits, we create \emph{multi-qubit systems}. Therefore the spaces $U$ and $V$ with bases $\Set{u_0, \ldots, u_i, \ldots}$ and $\Set{v_0, \ldots, v_j, \ldots}$ are joined using the \emph{tensor product} into one space $U \tensorProd V$ with basis $\Set{u_0 \tensorProd v_0, \ldots, u_i \tensorProd v_j, \ldots}$. So a system consisting of $n$ qubits has a $2^n$-dimensional space with standard bases $\Ket{00\ldots0} \ldots \Ket{11\ldots1}$.
Within these systems we can measure a single or multiple qubits.
As an example for measurement, consider the 2-qubit system with the basis $ \Set{ \Ket{00}, \Ket{01}, \Ket{10}, \Ket{11} } $ and the general state $\alpha\Ket{00} + \beta\Ket{01} + \gamma\Ket{10} + \delta\Ket{11}$ with $\Length{\alpha}^2 + \Length{\beta}^2 + \Length{\gamma}^2 + \Length{\delta}^2 = 1$.
A measurement of the first qubit gives result $0$ with probability $\Length{\alpha}^2 + \Length{\beta}^2$ and leaves the system in state $\dfrac{1}{\sqrt{\Length{\alpha}^2 + \Length{\beta}^2}}(\alpha\Ket{00} + \beta\Ket{01})$. The result $1$ is given with probability $\Length{\gamma}^2 + \Length{\delta}^2$.
In this case the system has state $\dfrac{1}{\sqrt{\Length{\gamma}^2 + \Length{\delta}^2}}(\gamma\Ket{10} + \delta\Ket{11})$. Further, the measurement of both qubits simultaneously gives result $ 0 $ for both qubits with probability $\Length{\alpha}^2$ (leaving the system in state $\Ket{00}$), result $ 0 $ for the first and $ 1 $ for the second qubit with probability $ \Length{\beta}^2 $ (leaving the system in state $\Ket{01}$) and so on.
We use binary numbers to refer to measurement results, \ie for two qubits the measurement results are $ 00 $, $ 01 $, $ 10 $, or $ 11 $.

In multi-qubit systems unitary transformations can be performed on single or several qubits.
As an example for an unitary transformation, consider the transformation $\opX$ on both qubits of a 2-qubit system in state $\Ket{00}$ simultaneously, we use the unitary transformation $\opX \tensorProd \opX$. The result of $(\opX \tensorProd \opX) \Ket{00}$ is the state $\Ket{11}$. To apply $\opX$ only to the second qubit, we use $\opI \tensorProd \opX$ and $(\opI \tensorProd \opX) \Ket{00} = \Ket{01}$.
The Pauli matrix $ \opI $ denotes the identity matrix in $ 2^1 $ dimensional space.
By slightly abusing notation we also use $ \opI_{\Set{q_1, \ldots, q_n} } $ or simply $ \opI $ to denote identity in $ 2^n $ dimensional space for all natural numbers $ n $.

The multi-qubit systems can exhibit \emph{entanglement}, meaning that states of qubits are correlated, \eg in $ \frac{1}{\sqrt{2}}\left( \Ket{00} + \Ket{11} \right) $ which is one of the so-called \emph{Bell pairs}.
Here, a measurement of the first qubit in the computational basis results in $0$ (leaving the state $\Ket{00}$) with probability $\frac{1}{2}$ and in $1$ (leaving the state $\Ket{11}$) with probability $\frac{1}{2}$. In both cases a subsequent measurement of the second qubit in the same basis gives the same result as the first measurement with probability 1. The effect also occurs if the entangled qubits are physically separated. Because of this, states with entangled qubits cannot be written as a tensor product of single-qubit states.

States of quantum systems can also be described by \emph{density matrices} or \emph{density operators}.
In contrast to the vector description of states, density matrices allow to describe the states of open systems.
A density operator in a Hilbert space $ \hilbert $ is a linear operator $ \rho $ on it, such that $ \Adjoint{\Ket{\psi}}\rho\Ket{\psi} \geq 0 $ for all $ \Ket{\psi} $ and $ \Trace{\rho} = 1 $, where the trace $ \Trace{\rho} $ is the sum of elements on the main diagonal of the matrix $ \rho $.
A positive operator $ \rho $ is called a partial density operator if $ \Trace{\rho} \leq 1 $.
We write $ \partDen $ for the set of (partial) density operators on $ \hilbert $.
For every state $ \Ket{\psi} $ in the above described vector representation, we obtain the corresponding density matrix by the outer product $ \OuterProduct{\psi}{\psi} = \Ket{\psi} \Adjoint{\Ket{\psi}} $.
For example, consider again the $ 2 $-qubit system in general state $ \Ket{\psi} = \alpha\Ket{00} + \beta\Ket{01} + \gamma\Ket{10} + \delta\Ket{11} $ which corresponds to the vector $ \Transp{\Tuple{\alpha, \beta, \gamma, \delta}} $.
The corresponding density matrix is given as:
\begin{align*}
	\OuterProduct{\psi}{\psi} =
	\Ket{\psi} \Adjoint{\Ket{\psi}} =
	\begin{pmatrix}
		\alpha \\
		\beta \\
		\gamma \\
		\delta
	\end{pmatrix}
	\Tuple{\ComplexCon{\alpha}, \ComplexCon{\beta}, \ComplexCon{\gamma}, \ComplexCon{\delta}} = 
	\begin{pmatrix}
		\alpha\ComplexCon{\alpha} & \alpha\ComplexCon{\beta} & \alpha\ComplexCon{\gamma} & \alpha\ComplexCon{\delta} \\
		\beta\ComplexCon{\alpha} & \beta\ComplexCon{\beta} & \beta\ComplexCon{\gamma} & \beta\ComplexCon{\delta} \\
		\gamma\ComplexCon{\alpha} & \gamma\ComplexCon{\beta} & \gamma\ComplexCon{\gamma} & \gamma\ComplexCon{\delta} \\
		\delta\ComplexCon{\alpha} & \delta\ComplexCon{\beta} & \delta\ComplexCon{\gamma} & \delta\ComplexCon{\delta}
	\end{pmatrix}
\end{align*}
where the adjoint $ \Adjoint{\Ket{\psi}} = \Tuple{\ComplexCon{\alpha}, \ComplexCon{\beta}, \ComplexCon{\gamma}, \ComplexCon{\delta}} $ is the conjugate transpose of $ \Ket{\psi} $.
Here, $ \ComplexCon{x} $ denotes the complex conjugate of $ x $.
For real numbers $ a $ and $ b $, the complex conjugate of $ a + ib $ is $ a - ib $.
Such states, \ie states that result from the outer product of a vector with itself, are called \emph{pure states}.
Additionally, density matrices can represent \emph{mixed states}, that arise either when the system is not fully known or when one wants to describe a system which is entangled with another.
Every density matrix can be represented as $ \sum_{i} p_i \OuterProduct{\psi_i}{\psi_i} $, called \emph{sum representation}, \ie by an ensemble of pure states $ \Ket{\psi_i} $ with their probabilities $ p_i \geq 0 $ and $ \sum_i p_i = 1 $.

We often use density matrix to refer to a state of a potentially open system and call the transformations on these states \emph{super-operators}.
Note that unitary transformations can only describe transitions in closed systems.
Super-operators are strictly more expressive, since they can also express interaction with an (unknown) environment.
Example~\ref{exa:phaseFlipChannel} in Section~\ref{sec:separation} presents a super-operator that does not resemble any unitary transformation.
This super-operator can be used to model a specific kind of noise in quantum communication.
Intuitively, noise is a form of partial entanglement with an unkown environment.
Note that the channels that are used to transfer qubit-systems in \CQP, \CQS, \qCCS, and \OQS, are modelled as noise-free channels, \ie noise has to be added explicitly by respective super-operators as discussed in \cite{YingFengDuanJi09}.
There are different ways to define super-operators, \eg via the sum representation.

\begin{defi}[Super-Operator, Operator-Sum Representation, \cite{NielsenChuang10}]
	\label{def:superopSum}
	Let $ \rho $ be the initial state of a system, $ \Ket{e_1}, \ldots, \Ket{e_n} $ be an orthonormal basis for the (finite dimensional) state space of the environment, and $ \rho_{\mathsf{env}} = \Ket{e_0}\Bra{e_0} $ be the initial state of the environment.
	A \emph{super-operator} $ \StateTrans{\opE}{}{\rho} $ on the system $ \rho $ is an operator $ \opE $ which is defined as $ \StateTrans{\opE}{}{\rho} = \sum_i E_i \rho \Adjoint{E_i} $, where $ E_i = \Bra{e_i} U \Ket{e_0} $ is an operator on the state space of the system.
	Thereby, the operators $ \Set{E_i} $ are known as \emph{operation elements} for the quantum operation $ \opE $, which have to satisfy $ \sum_i \Adjoint{E_i}E_i \leq \opI $.
	The super-operator $ \opE $ is \emph{trace-preserving} if $ \sum_i \Adjoint{E_i}E_i = \opI $.
\end{defi}

For every unitary transformation $ U $, $ \StateTrans{U}{}{\rho} = U \rho \Adjoint{U} $ is a trace-preserving super-operator.
Let $ \Set{M_m} $ such that $ \sum_m \Adjoint{M_m} M_m = \opI $.
Then, by \cite{YingFengDuanJi09}, $ \Set{M_m} $ is a collection of measurement operators.
We usually let $ m $ refer to the measurement outcome.
For each $ m $, let $ \opE_m\Tuple{\rho} = M_m \rho \Adjoint{M_m} $ for any state $ \rho \in \partDen $.
Moreover, let $ \opE\Tuple{\rho} = \sum_m M_m \rho \Adjoint{M_m} $ for any state $ \rho \in \partDen $.
Then $ \opE_m $ is a super-operator, which is not necessarily trace-preserving, whereas $ \opE $ is a trace-preserving super-operator (see Example~2.5 in \cite{YingFengDuanJi09}).

According to \cite{NielsenChuang10} the equation $ \opE\Tuple{\rho} = \sum_i E_i \rho \Adjoint{E_i} $ from Definition~\ref{def:superopSum}, is a re-statement of $ \opE\Tuple{\rho} = \TraceEnv{U \left(\rho \tensorProd \rho_{\mathsf{env}} \right)\Adjoint{U}} $, where $ \TraceEnv{} $ is a partial trace over the environment to obtain the reduced state of the system.
Within this equation it is assumed, that the environment starts in a pure state.
This assumption can be made without loss of generality, since we are free to introduce an extra system purifying the environment, if it starts in a mixed state.
Another assumption made within this equation is that the system and the environment start in a product state.
This is not true in general, as quantum systems constantly interact with their environment by which correlations are created.
Nonetheless, in many cases of practical interest it is reasonable to make this assumption, as by bringing a quantum system to a specific state these correlations are destroyed, leaving the system in a pure state.
We refer to \cite{NielsenChuang10} for further informations on super-operators.


\section{Process Calculi}
\label{sec:processCalculi}

A \emph{process calculus} is a language $ \lang = \left\langle \config, \step \right\rangle $ that consists of a set of \emph{configurations} $ \config $ (its syntax) and a relation $ \step : \config \times \config $ on configurations (its reduction semantics).
To range over the configurations we use the upper case letters $ C, C', \ldots $.
Further, a configuration $ C $ contains a \emph{term} out of the set of (process) terms $ \proc $ on which we range over using the upper case letters $ P, Q, P', \ldots $.

Assume three pairwise distinct countably-infinite sets $ \names $ of \emph{names}, $ \var $ of \emph{qubit variables}, and $ \binaries $ of \emph{variables for binary numbers}.
We use lower case letters to range over names $ a, c, \ldots $, qubits names $ q, q', x, y, \ldots $, binary numbers $ b, b', \ldots $, and variables for binary numbers $ v, v', \ldots $.
We write $ bv, bv', \ldots $ for objects that are either a binary number or a variable for binary numbers.
Let $ \tau \notin \var \cup \names \cup \binaries $.
The \emph{scope} of a name defines the area in which this name is known and can be used. It can be useful to restrict this scope, for example to forbid interactions between two processes or with an unknown and, hence, potentially untrusted environment. While names with a restricted scope are called \emph{bound names}, the remaining ones are called \emph{free names}.

The \emph{syntax} of a process calculus is usually defined by a context-free grammar defining operators, \ie functions $ \operatorname{op} : \proc^n \rightarrow \proc $ with $ n \geq 0 $.
An operator of arity $0$ is a \emph{constant}.
The \emph{semantics} of a process calculus is given as a \emph{structural operational semantics} consisting of inference rules defined on the operators of the language \cite{Plotkin04}.
The semantics is provided often in two forms, as \emph{reduction semantics} and as \emph{labelled transition semantics}.
We assume that at least the reduction semantics is given, because its treatment is easier in the context of encodings.
As we naturally extend the definition of the syntax to configurations, a \emph{(reduction) step}, written as $ C \step C' $, is a single application of the reduction semantics where $ C' $ is called \emph{derivative}. Let $ C \step $ denote the existence of a step from $ C $.
We write $ C \infiniteSteps $ if $ C $ has an \emph{infinite sequence} of steps and $ \steps $ to denote the \emph{reflexive and transitive closure} of $ \step $.

To reason about environments of terms, we use functions on process terms called contexts. More precisely, a \emph{context} $\Context{}{}{\hole_1, \ldots, \hole_n} : \proc^n \to \proc$ with $n$ holes is a function from $ n $ terms into one term, \ie given $ P_1, \ldots, P_n \in \proc $, the term $ \Context{}{}{P_1, \ldots, P_n} $ is the result of inserting $ P_1, \ldots, P_n $ in the corresponding order into the $ n $ holes of $ \context $.
We naturally extend the definition of contexts to configurations, \ie consider also contexts $ \Context{}{}{\hole_1, \ldots, \hole_n} : \proc^n \to \config $.

A substitution is a finite mapping on either names or qubits or variables for binary numbers defined by a non-empty set $ \Set{ \Subst{h_1}{g_1}, \ldots, \Subst{h_n}{g_n} } = \Set{ \Subst{h_1, \ldots, h_n}{g_1, \ldots, g_n}} $ of renamings, where the $ g_1, \ldots, g_n $ are pairwise distinct.
The application $ P\Set{ \Subst{h_1}{g_1}, \ldots, \Subst{h_n}{g_n} } $ of a substitution on a term is defined as the result of simultaneously replacing all free occurrences of $ g_i $ by $ h_i $ for $ i \in \Set{ 1, \ldots, n } $, possibly applying $ \alpha $-conversion to avoid capture or name clashes.
For all names in $ \names \setminus \Set{ g_1, \ldots, g_n } $ or qubits in $ \var \setminus \Set{ g_1, \ldots, g_n } $ or variables in $ \binaries \setminus \Set{ g_1, \ldots, g_n } $ the substitution behaves as the identity mapping.
Substitutions on qubits additionally cannot translate different qubits to the same qubit, since this might violate the no-cloning principle.
More on substitutions of qubits can be found, \eg, in \cite{YingFengDuanJi09}.
We naturally extend substitutions to mappings that instantiate variables for binary numbers by binary numbers.
We equate terms and configurations modulo alpha conversion on (qubit) names.

For the last criterion of \cite{gorla10} in Section~\ref{sec:criteria}, we need a special constant $ \success $, called \emph{success(ful termination)}, in both considered languages.
Therefore, we add $ \success $ to the grammars of both languages without explicitly mentioning them.
Success is used as a barb, where $ \HasBarb{C}{\success} $ if the term contained in the configuration $ C $ has an unguarded occurrence of $ \success $ and $ \ReachBarb{C}{\success} = \exists C' \logdot C \steps C' \wedge \HasBarb{C'}{\success} $, to implement some form of (fair) testing.


\subsection{A Calculus for Closed Quantum Systems}
\label{sec:CQS}

Communicating Quantum Processes (\CQP) is introduced in \cite{GayNagarajan05}.
\CQP is further studied \eg in \cite{DavidsonGayNagarajanPuthoor12} to study quantum error correction, in \cite{FrankeArnoldGayPuthoor13,FrankeArnoldGayPuthoor14} to describe and analyse linear optical quantum computing, or in \cite{GayPuthoor12}, where it is extended to be able to describe d-dimensional quantum systems.

As indicated in Section~\ref{sec:introduction}, we build \CQS by inheriting some ideas of \CQP.
However, the resulting language \CQS is strictly less expressive than \CQP.
We simplify the definition of \CQP by removing name passing and contexts, the additional layer on expressions in the syntax and semantics, do not allow to construct channel names from expressions, and by using a monadic version of communication in that only qubits can be transmitted.
Then we add a standard conditional operator, that allows to compare two binary numbers.
\CQP in \cite{GayNagarajan05} does not have such a conditional, but as stated in footnote 3 in \cite{GayNagarajan05} the language can easily be extended by an operator to test the result of measurement---just as the conditional we add here.
We claim, however that the treatment of qubits, in particular the manipulations of the quantum register as well as the communication of qubits, is the same as in \CQP.
Let $ \Binary{i} $ return the binary number representing the natural number $ i $.

\begin{defi}[\CQS]
	The \CQS \emph{terms}, denoted by $ \procCQS $, are given by:
	\begin{align*}
		P & \deffTerms \nilCQS \sep P \mid P \sep \InpCQS{c}{x}{P} \sep \OutCQS{c}{x}{P} \sep \UnitaryCQS{\tilde{x}}{U}{P}\\
		& \sep \MeasCQS{v}{\tilde{x}}{P} \sep \NewCQS{c}{P} \sep \QubitCQS{x}{P} \sep \CondCQS{bv}{bv'}{P}
	\end{align*}
	\CQS \emph{configurations} $ \configCQS $ are given by $ \ConfigCQS{\sigma}{\phi}{P} $ or $ \boxplus_{0 \leq i < 2^r} p_i \bullet \ConfigCQS{\sigma_i}{\phi}{P\Set{\Subst{\Binary{i}}{v}}} $, where $ \sigma, \sigma_i $ have the form $ q_0, \ldots, q_{n-1} = \Ket{\psi} $ with $ \Ket{\psi} = \sum_{i = 0}^{2^n-1}\alpha_i \Ket{\psi_i} $, $ r \leq n $, $ \phi $ is the list of channels in the system, and $ P \in \procCQS $.
\end{defi}

The syntax of \CQS is $ \pi $-calculus like.
The inactive process is denoted by $ \nilCQS $ and $ P \mid P $ defines parallel composition.
A term $ \InpCQS{c}{x}{P} $ receives a qubit $ q \in \var $ over channel $ c \in \names $ and proceeds as $ P\Set{\Subst{q}{x}} $.
Similarly, $ \OutCQS{c}{x}{P} $ first sends a qubit $ x \in \var $ over channel $ c \in \names $ before proceeding as $ P $.
The term $ \UnitaryCQS{\tilde{x}}{U}{P} $ applies the unitary transformation $ U $ to the qubits in sequence $ \tilde{x} $ and then proceeds as $ P $.
The process $ \MeasCQS{v}{\tilde{x}}{P} $ measures the qubits in $ \tilde{x} $ with $ \Length{\tilde{x}} > 0 $ and saves the result in the variable $ v $ for binary numbers.
The terms $ \NewCQS{c}{P} $ and $ \QubitCQS{x}{P} $ create a fresh, global channel $ a \in \names $ and a fresh qubit $ q_n \in \var $ (for a quantum register $ \sigma = q_0, \ldots, q_{n - 1} $) and then proceed as $ P\Set{\Subst{a}{c}} $ and $ P\Set{\Subst{q_n}{x}} $, respectively.

The configuration $ \boxplus_{0 \leq i < 2^r} p_i \bullet C_i $ denotes a probability distribution over configurations $ C_i = \ConfigCQS{\sigma_i}{\phi}{P\Set{\Subst{\Binary{i}}{v}}} $, where $ \sum_i p_i = 1 $ and where the terms within the configurations $ C_i $ may differ only by instantiating a variable $ v $ by the binary number $ \Binary{i} $.
It results from measuring the first $ r $ qubits, where $ p_i $ is the probability of obtaining result $ \Binary{i} $ from measuring the qubits $ q_0, \ldots, q_{r - 1} $ and $ C_i $ is the configuration of case $ i $ after the measurement.
Indeed we restrict our attention to probability distributions of configurations that may be the result of measuring a state of a single configuration.
In particular, this means that the states $ \sigma_i $ of a probability distribution have to reflect the possible outcomes of the measurement, \ie for a single qubit $ \sigma_0 = q = \Ket{0} $ and $ \sigma_1 = q = \Ket{1} $.
We may also write a distribution as $ p_1 \bullet C_1 \boxplus \ldots \boxplus p_j \bullet C_j $ with $ j = 2^r - 1 $.
We equate $ \ConfigCQS{\sigma_0}{\phi}{P} $ and $ \boxplus_{0 \leq i < 2^0} 1 \bullet \ConfigCQS{\sigma_i}{\phi}{P\Set{\Subst{\Binary{i}}{v}}} $, \ie if $ r = 0 $ then we assume that $ v $ is not free in $ P $.

The variable $ x \in \var $ is bound in $ P $ by $ \InpCQS{c}{x}{P} $ and $ \QubitCQS{x}{P} $.
Similarly, the variable $ v \in \binaries $ is bound in $ P $ by $ \MeasCQS{v}{\tilde{x}}{P} $ and the variable $ c \in \names $ is bound in $ P $ by $ \NewCQS{c}{P} $.
A variable is free if it is not bound.
Let $ \FreeQubits{P} $, $ \FreeChan{P} $, and $ \FreeBinV{P} $ denote the sets of free qubits, free channels, and free variables for binary numbers in $ P $, respectively.

The state $ \sigma $ is represented by a list of qubits $ q_0, \ldots, q_{n-1} $ as well as a linear combination $ \Ket{\psi} = \sum_{i = 0}^{2^n-1}\alpha_i \Ket{\psi_i} $ which can also be rewritten by a vector $ \Transp{\Tuple{\alpha_0, \ldots, \alpha_{2^n-1}}} $, where $ \Transp{} $ stands for transpose.
As done in \cite{GayNagarajan05}, we sometimes write as an abbreviated form $ \sigma = q_0, \ldots, q_{n-1} $ or $ \sigma = \Ket{\psi} $.

\begin{figure}[t]
	\centering
	\begin{tabular}{r l}
		\ruleRMeasureCQS & $ \ConfigCQS{\sigma}{\phi}{\MeasCQS{v}{q_0, \ldots q_{r-1}}{P}} $\\
		& \quad $ \step \boxplus_{0 \leq m < 2^r} p_m \bullet \ConfigCQS{\sigma_m'}{\phi}{P\Set{\Subst{\Binary{m}}{v}}} $ \vspace{0.25em}\\
		\ruleRTransCQS & $ \ConfigCQS{q_0, \ldots, q_{n-1} = \Ket{\psi}}{\phi}{\UnitaryCQS{q_0, \ldots, q_{r-1}}{U}{P}} $ \\
		& \quad $ \step \ConfigCQS{q_0, \ldots, q_{n-1} = \Tuple{U \tensorProd \opI_{\Set{q_r, \ldots, q_{n-1}}}}\Ket{\psi}}{\phi}{P} $ \vspace{0.25em}\\
		\ruleRPermCQS & $ \ConfigCQS{q_0, \ldots, q_{n-1} = \Ket{\psi}}{\phi}{P} \step \ConfigCQS{q_{\pi\Tuple{0}}, \ldots, q_{\pi\Tuple{n-1}} = \prod\Ket{\psi}}{\phi}{P\pi} $ \vspace{0.25em}\\
		\ruleRProbCQS & $ \boxplus_{0 \leq i < 2^r} p_i \bullet \ConfigCQS{\sigma_i}{\phi}{P\Set{\Subst{\Binary{i}}{v}}} $\\
		& \quad $ \step \ConfigCQS{\sigma_{j}}{\phi}{P\Set{\Subst{\Binary{j}}{v}}} $ \quad where $ p_j \neq 0 $ and $ r > 0 $ \vspace{0.25em}\\
		\ruleRNewCQS & $ \ConfigCQS{\sigma}{\phi}{\NewCQS{c}{P}} \step \ConfigCQS{\sigma}{\phi, a}{P\Set{\Subst{a}{c}}} $ \quad where $ a $ is fresh \vspace{0.25em}\\
		\ruleRQbitCQS & $ \ConfigCQS{q_0, \ldots, q_{n - 1} = \Ket{\psi}}{\phi}{\QubitCQS{x}{P}} $ \\
		& \quad $ \step \ConfigCQS{q_0, \ldots, q_{n - 1}, q_n = \Ket{\psi} \tensorProd \Ket{0}}{\phi}{P\Set{\Subst{q_n}{x}}} $\\
		\ruleRCommCQS & $ \ConfigCQS{\sigma}{\phi}{\OutCQS{c}{q}{P} \mid \InpCQS{c}{x}{Q}} \step \ConfigCQS{\sigma}{\phi}{P \mid Q\Set{\Subst{q}{x}}} $ \vspace{0.25em}\\
		\ruleRParCQS & $ \dfrac{\ConfigCQS{\sigma}{\phi}{P} \step \boxplus_{0 \leq i < 2^r} p_i \bullet \ConfigCQS{\sigma_i'}{\phi'}{P'\Set{\Subst{\Binary{i}}{v}}}}{\ConfigCQS{\sigma}{\phi}{P \mid Q} \step \boxplus_{0 \leq i < 2^r} p_i \bullet \ConfigCQS{\sigma_i'}{\phi'}{P'\Set{\Subst{\Binary{i}}{v}} \mid Q}} $ \vspace{0.25em}\\
		\ruleRCongCQS & $ \dfrac{Q \equiv P \quad \ConfigCQS{\sigma}{\phi}{P} \step \boxplus_{0 \leq i < 2^r} p_i \bullet \ConfigCQS{\sigma_i'}{\phi'}{P'\Set{\Subst{\Binary{i}}{v}}} \quad P' \equiv Q'}{\ConfigCQS{\sigma}{\phi}{Q} \step \boxplus_{0 \leq i < 2^r} p_i \bullet \ConfigCQS{\sigma_i'}{\phi'}{Q'\Set{\Subst{\Binary{i}}{v}}}} $ \vspace{0.25em}\\
		\ruleRCondCQS & $ \dfrac{b = b'}{\ConfigCQS{\sigma}{\phi}{\CondCQS{b}{b'}{P}} \step \ConfigCQS{\sigma}{\phi}{P}} $
	\end{tabular}
	\caption{Semantics of \CQS}
	\label{fig:semanticsCQS}
\end{figure}

The semantics of \CQS is defined by the reduction rules in Figure~\ref{fig:semanticsCQS}.
These rules are inspired by the semantics of \CQP in \cite{GayNagarajan05} but do not require a second layer for expressions, since we simplified the syntax, and drop the label of Rule~\ruleRProbCQS.
Accordingly, \CQS in contrast to \CQP does not have a probabilistic transition system, but replaces probabilistic steps by non-deterministic steps.
We do that, because the encodability criteria that we study here (see Section~\ref{sec:criteria}) do not consider probabilistic transitions systems.
We discuss this issue in Section~\ref{sec:criteriaQBS}.
Moreover, we add the Rule~\ruleRCondCQS to reduce conditionals.
Rule~\ruleRMeasureCQS measures the first $ r $ qubits of $ \sigma $, where $ \sigma = \alpha_0\Ket{\psi_0} + \cdots + \alpha_{2^n-1}\Ket{\psi_{2^n-1}} $, $ \sigma_m' = \dfrac{\alpha_{l_m}}{\sqrt{p_m}}\Ket{\psi_{l_m}} + \cdots + \dfrac{\alpha_{u_m}}{\sqrt{p_m}}\Ket{\psi_{u_m}} $, $ l_m = 2^{n-r}m $, $ u_m = 2^{n-r}\Tuple{m+1}-1 $, and $ p_m = \Length{\alpha_{l_m}}^2 + \cdots + \Length{\alpha_{u_m}}^2 $.
As a result a probability distribution over the possible base vectors is generated, where $\sigma_m'$ is the accordingly updated qubit vector and $ \Binary{m} $ is the respective measurement outcome.
Rule~\ruleRTransCQS applies the unitary transformation $ U $ on the first $ r $ qubits.
In contrast to \cite{GayNagarajan05}, we explicitly list in the subscript of $ \opI $ the qubits it is applied to.
As the rules \ruleRMeasureCQS and \ruleRTransCQS operate on the first $ r $ qubits within $ \sigma $, Rule~\ruleRPermCQS allows to permute the qubits in $ \sigma $. Thereby, $ \pi $ is a permutation and $ \prod $ is the corresponding unitary operator.

The Rule~\ruleRProbCQS reduces a probability distribution with $ r > 0 $ to a single of its configurations $ \ConfigCQS{\sigma_j}{\phi}{P\Set{\Subst{\Binary{j}}{v}}} $ with non-zero probability $ p_j $.
In contrast to \cite{GayNagarajan05} we drop the label indicating the probability $ p_j $ of the chosen case.
The rules~\ruleRNewCQS and \ruleRQbitCQS create new channels and qubits and update the list of channel names or the qubit vector.
Thereby, a new qubit is initialised to $ \Ket{0} $ and $ \Ket{\psi} \tensorProd \Ket{0} $ is reshaped into a $ \left(2^{n+1}\right) $-vector.
The Rule~\ruleRCommCQS defines communication in the style of the $ \pi $-calculus.
Rule~\ruleRParCQS allows reduction to take place under parallel contexts and Rule~\ruleRCongCQS enables the use of structural congruence as in the $ \pi $-calculus.
The structural congruence of \CQS is defined, similarly to \cite{GayNagarajan05}, as the smallest congruence containing $\alpha$-equivalence that is closed under the following rules:
\begin{displaymath}
	P \mid 0 \equiv P \quad\quad P \mid Q \equiv Q \mid P \quad\quad P \mid \left(Q \mid R \right) \equiv \left( P \mid Q \right) \mid R
\end{displaymath}
Moreover, $ \ConfigCQS{\sigma}{\phi}{P} \equiv \ConfigCQS{\sigma'}{\phi}{Q} $ if $ P \equiv Q $ and $ \sigma = \sigma' $ or if $ \ConfigCQS{\sigma'}{\phi}{Q} $ is obtained from $ \ConfigCQS{\sigma}{\phi}{P} $ by alpha conversion on the qubit names in $ \sigma $.
Finally, Rule~\ruleRCondCQS unguards the continuation $ P $ of a conditional if its condition is satisfied, which checks equality of two binary numbers $ b $ and $ b' $.

As \CQP also \CQS is augmented with a type system to ensure that two parallel components cannot share access to the same qubits, which is the realisation of the no-cloning principle of qubits in \CQP.
We use a very simple type system compared to \cite{GayNagarajan05}, which is possible since we significantly simplified \CQS in comparison to \CQP and since we require the sets $ \names $, $ \var $, and $ \binaries $ to be pairwise distinct.
Remember that we equate configurations and terms modulo alpha conversion.
We use this in the type system to ensure that there are no name clashes, \ie that no two bound variables have the same name and no bound variable has the same name as a free variable.
We extend this convention to also require that no variable of a qubit has the name $ q_i $ for any natural number $ i $ such that \ruleRQbitCQS does not cause name clashes.
The \CQS \emph{types}, denoted by $ \typesCQS $, are given by:
\begin{align*}
	T & \deffTerms \binariesType \sep \OpType{n}
\end{align*}
The data type $ \binariesType $ is used for binary numbers.
The type $ \OpType{n} $ is used for unitary transformations that are applied to $ n $ qubits.

\begin{figure}[t]
	\centering
	\begin{displaymath}\begin{array}{c}
			\ruleTBin \; \dfrac{b \text{ is a binary number}}{\vdash \At{b}{\binariesType}}
			\hspace{2em}
			\ruleTOp \; \dfrac{U \text{ is a unitary transformation on } n \text{ qubits}}{\vdash \At{U}{\OpType{n}}}
			\vspace{0.5em}\\
			\ruleTNil \; \Sigma \vdash \nilCQS
			\hspace{2em}
			\ruleTSuc \; \Sigma \vdash \success
			\hspace{2em}
			\ruleTPar \; \dfrac{\Sigma_1 \vdash P \quad \Sigma_2 \vdash Q \quad \Sigma_1 \cap \Sigma_2 = \emptyset}{\Sigma_1 \cup \Sigma_2 \vdash P \mid Q}
			\vspace{0.5em}\\
			\ruleTIn \; \dfrac{c \in \names \quad x \in \var \setminus \Sigma \quad \Sigma \cup \Set{x} \vdash P}{\Sigma \vdash \InpCQS{c}{x}{P}}
			\hspace{2em}
			\ruleTOut \; \dfrac{c \in \names \quad x \in \var \cap \Sigma \quad \Sigma \setminus \Set{x} \vdash P}{\Sigma \vdash \OutCQS{c}{x}{P}}
			\vspace{0.5em}\\
			\ruleTTrans \; \dfrac{x_1, \ldots, x_n \in \var \cap \Sigma \quad \vdash \At{U}{\OpType{n}} \quad \Sigma \vdash P}{\Sigma \vdash \UnitaryCQS{x_1, \ldots, x_n}{U}{P}}
			\hspace{2em}
			\ruleTNew \; \dfrac{c \in \names \quad \Sigma \vdash P}{\Sigma \vdash \NewCQS{c}{P}}
			\vspace{0.5em}\\
			\ruleTMsure \; \dfrac{v \in \binaries \quad x_1, \ldots x_n \in \var \cap \Sigma \quad \Sigma \vdash P}{\Sigma \vdash \MeasCQS{v}{x_1, \ldots, x_n}{P}}
			\hspace{2em}
			\ruleTQbit \; \dfrac{x \in \var \setminus \Sigma \quad \Sigma \cup \Set{x} \vdash P}{\Sigma \vdash \QubitCQS{x}{P}}
			\vspace{0.5em}\\
			\ruleTCond \; \dfrac{\left( bv \in \binaries \vee {\vdash \At{bv}{\binariesType}} \right) \quad \left( bv' \in \binaries \vee {\vdash \At{bv'}{\binariesType}} \right) \quad \Sigma \vdash P}{\Sigma \vdash \CondCQS{bv}{bv'}{P}}
	\end{array}\end{displaymath}
	\caption{Typing Rules for \CQS}
	\label{fig:typingRulesCQS}
\end{figure}

Type \emph{judgements} for processes are of the form $ \Sigma \vdash P $, where $ \Sigma $ is a set of qubit names and $ P \in \procCQS $.
The set $ \Sigma $ is supposed to contain all free qubit names in the process as we show in Lemma~\ref{lem:typingFreeQubitsCQS}.
A type judgement $ \Sigma \vdash P $ holds if it can be derived from the rules in Figure~\ref{fig:typingRulesCQS}.
These rules are inspired by \cite{GayNagarajan05}.
By Rule~\ruleTPar parallel processes do not use the same qubits, since they can be typed \wrt to distinct sets $ \Sigma_1 $ and $ \Sigma_2 $.
Rule~\ruleTIn checks that the variable used in inputs is from $ \var $ but not yet known to the continuation $ P $, \ie not in $ \Sigma $.
Conversely, \ruleTOut ensures that the transmitted qubit $ x $ in outputs was known before, \ie in $ x \in \var \cup \Sigma $, but is no longer available to the continuation $ P $ after sending it away.
To ensure the latter, $ P $ is checked against $ \Sigma \setminus \Set{x} $.
Rule~\ruleTQbit checks whether the new qubit $ x $ was not known before by $ x \in \var \setminus \Sigma $ and then adds $ x $ to $ \Sigma $ for the analyse of the remaining process.
The remaining rules are self-explanatory.

We show three properties of the type system.
Since the focus of this paper is on encodability criteria and not type systems of process calculi, the proofs of these properties can be found in the Appendix~\ref{sec:typeSystemCQS}.
First we capture the intuition behind $ \Sigma $, as capturing at least all free qubit names of a process.

\begin{lem}[Free Qubits]
	If $ \Sigma \vdash P $ then $ \FreeQubits{P} \subseteq \Sigma $.
	\label{lem:typingFreeQubitsCQS}
\end{lem}

Then we have the standard preservation property.

\begin{lem}[Preservation]
	If $ \Sigma \vdash P $ and $ \ConfigCQS{\sigma}{\phi}{P} \step \boxplus_{0 \leq i < 2^r} p_i \bullet \ConfigCQS{\sigma_i'}{\phi'}{P_i} $ or if $ \Sigma \vdash P_k' $ for all $ 0 \leq k < 2^t $ and $ \boxplus_{0 \leq k < 2^t} p_k' \bullet \ConfigCQS{\sigma}{\phi}{P_k'} \step \boxplus_{0 \leq i < 2^r} p_i \bullet \ConfigCQS{\sigma_i'}{\phi'}{P_i} $ then there is some $ \Sigma' \in \Set{\Sigma, \Sigma \cup \Set{q_n} } $ for some fresh $ q_n $ such that $ \Sigma' \vdash P_i $ for all $ 0 \leq i < 2^{r} $.
	\label{lem:preservationCQS}
\end{lem}

Finally, Lemma~\ref{lem:typingFreeQubitsCQS} ensures the no-cloning principle for well-typed \CQS-terms, since their parallel components cannot have access to the same qubit.
With Lemma~\ref{lem:preservationCQS} the principle is then also preserved in all derivatives.

\begin{lem}[Unique Ownership of Qubits]
	If $ \Sigma \vdash P \mid Q $ then $ \FreeQubits{P} \cap \FreeQubits{Q} = \emptyset $.
	\label{lem:typingCQSUniqueOwnership}
\end{lem}

Note that Lemma~\ref{lem:typingCQSUniqueOwnership} is an adaptation of the Theorem~2 in \cite{GayNagarajan05}---that there ensures the no cloning principle---to the present simpler type system.

As an example in \CQS we consider an implementation of the quantum teleportation protocol \cite{bennettBassardCrepeauJoszsaPeresWootters93}. The quantum teleportation protocol is a procedure for transmitting a quantum state via a non-quantum medium. This protocol is particularly important: not only it is a fundamental component of several more complex protocols, but it is likely to be a key enabling technology for the development of the quantum repeaters \cite{RiedmattenMarcikicTittelZbindenCollinsGisin04} which will be necessary in large-scale quantum communication networks.
The following example is an adaptation of the quantum teleportation example in Figure~3 of \cite{GayNagarajan05} adapted to \CQS.
Note that the original quantum teleportation protocol in \cite{bennettBassardCrepeauJoszsaPeresWootters93,GayNagarajan05}
does not require to transmit qubits but only two bits of classical information obtained from measuring qubits.
Since we stripped \CQS from the ability to transmit classical information, we have to cheat in the following example.
After measuring the relevant qubits, the qubits themselves and not the result of their measurement is transmitted.
However, since measurement transfers the respective qubits into base states, the respective communication does not carry any additional information than the result of measurement.
Of course the relevance of quantum teleportation steams from the fact that the original protocol does not need to transfer qubits.

\begin{exa}[Quantum Teleportation]
	\label{exa:teleportationSource}
	Consider the \CQS-configuration $ S $
	\begin{align*}
		S &= \ConfigCQS{q_0, q_1, q_2 = \frac{1}{\sqrt{2}}\Ket{100} + \frac{1}{\sqrt{2}}\Ket{111}}{\emptyset}{\mathit{System}\Tuple{q_0, q_1, q_2}}
	\end{align*}
	where
	\begin{align*}
		\mathit{System}\Tuple{q_0, q_1, q_2} ={}& \NewCQS{c}{\left( \mathit{Alice}\Tuple{q_0, q_1} \mid \mathit{Bob}\Tuple{q_2} \right)}\\
		\mathit{Alice}\Tuple{q_0, q_1} ={}& \UnitaryCQS{q_0, q_1}{\opCNot}{}\UnitaryCQS{q_0}{\opH}{}\MeasCQS{v_0}{q_0, q_1}{}\OutCQS{c}{q_0}{}\OutCQS{c}{q_1}{}\nilCQS\\
		\mathit{Bob}\Tuple{q_2} ={}& \InpCQS{c}{x_0}{}\InpCQS{c}{x_1}{}\MeasCQS{v}{x_0, x_1}{}\big( \CondCQS{v}{00}{\success}\\
		& {} \mid \CondCQS{v}{01}{\UnitaryCQS{q_2}{\opX}{\success}} \mid \CondCQS{v}{10}{\UnitaryCQS{q_2}{\opZ}{\success}}\\
		& {} \mid \CondCQS{v}{11}{\UnitaryCQS{q_2}{\opY}{\success}} \big)
	\end{align*}
	Alice and Bob each possess one qubit ($q_1$ for Alice and $q_2$ for Bob) of an entangled pair in state $\frac{1}{\sqrt{2}} \Ket{00} + \frac{1}{\sqrt{2}}\Ket{11}$. $q_0$ is the second qubit owned by Alice. Within this example it is in state $\Ket{1}$, but in general it can be in an arbitrary state. It is the qubit whose state will be teleported to $q_2$ and therefore to Bob.

	By Figure~\ref{fig:semanticsCQS}, $ S $ can do the following steps
	\allowdisplaybreaks
	\begin{align*}
		S \step{}& \ConfigCQS{\Ket{\psi_0}}{c}{\mathit{Alice}\Tuple{q_0, q_1} \mid \mathit{Bob}\Tuple{q_2}}\\
		\step{}& \ConfigCQS{\Ket{\psi_1}}{c}{\UnitaryCQS{q_1}{\opH}{}\MeasCQS{v_0}{q_0, q_1}{}\OutCQS{c}{q_0}{}\OutCQS{c}{q_1}{}\nilCQS \mid \mathit{Bob}\Tuple{q_2}}\\
		\step{}& \ConfigCQS{\Ket{\psi_2}}{c}{\MeasCQS{v_0}{q_0, q_1}{}\OutCQS{c}{q_0}{}\OutCQS{c}{q_1}{}\nilCQS \mid \mathit{Bob}\Tuple{q_2}}\\
		\step{}& \frac{1}{4} \bullet \ConfigCQS{q_0, q_1, q_2, = \Ket{001}}{c}{\OutCQS{c}{q_0}{}\OutCQS{c}{q_1}{}\nilCQS \mid \mathit{Bob}\Tuple{q_2}} \boxplus{}\\
		& \frac{1}{4} \bullet \ConfigCQS{q_0, q_1, q_2, = \Ket{010}}{c}{\OutCQS{c}{q_0}{}\OutCQS{c}{q_1}{}\nilCQS \mid \mathit{Bob}\Tuple{q_2}} \boxplus{}\\
		& \frac{1}{4} \bullet \ConfigCQS{q_0, q_1, q_2, = \Ket{101}}{c}{\OutCQS{c}{q_0}{}\OutCQS{c}{q_1}{}\nilCQS \mid \mathit{Bob}\Tuple{q_2}} \boxplus{}\\
		& \frac{1}{4} \bullet \ConfigCQS{q_0, q_1, q_2, = \Ket{110}}{c}{\OutCQS{c}{q_0}{}\OutCQS{c}{q_1}{}\nilCQS \mid \mathit{Bob}\Tuple{q_2}} = S^*
	\end{align*}
	with $ \Ket{\psi_0} = q_0, q_1, q_2 = \frac{1}{\sqrt{2}}\Ket{100} + \frac{1}{\sqrt{2}}\Ket{111} $, $ \Ket{\psi_1} = q_0, q_1, q_2 = \frac{1}{\sqrt{2}}\Ket{110} + \frac{1}{\sqrt{2}}\Ket{101} $, and $ \Ket{\psi_2} = q_0, q_1, q_2 = \frac{1}{2}\Ket{001} + \frac{1}{2}\Ket{010} - \frac{1}{2}\Ket{101} - \frac{1}{2}\Ket{110} $.

	All configurations within the probability distribution in $S^*$ have the same probability.
	We can \eg choose the first one by using Rule~\ruleRProbCQS with $ \Ket{\psi_3} = q_0, q_1, q_2 = \Ket{001} $.
	\begin{align*}
		S^{*} &\step \ConfigCQS{\Ket{\psi_3}}{c}{\OutCQS{c}{q_0}{}\OutCQS{c}{q_1}{}\nilCQS \mid \mathit{Bob}\Tuple{q_2}}\\
		&\step \ConfigCQS{\Ket{\psi_3}}{c}{\begin{array}{l}
			\OutCQS{c}{q_1}{}\nilCQS \mid \InpCQS{c}{x_1}{}\MeasCQS{v}{q_0, x_1}{}\big( \CondCQS{v}{00}{\success}\\
			\mid \CondCQS{v}{01}{\UnitaryCQS{q_2}{\opX}{\success}} \mid \CondCQS{v}{10}{\UnitaryCQS{q_2}{\opZ}{\success}}\\
			\mid \CondCQS{v}{11}{\UnitaryCQS{q_2}{\opY}{\success}} \big)
		\end{array}}\\
		&\step \ConfigCQS{\Ket{\psi_3}}{c}{\begin{array}{l}
			\MeasCQS{v}{q_0, q_1}{}\big( \CondCQS{v}{00}{\success} \mid \CondCQS{v}{01}{\UnitaryCQS{q_2}{\opX}{\success}}\\
			\mid \CondCQS{v}{10}{\UnitaryCQS{q_2}{\opZ}{\success}} \mid \CondCQS{v}{11}{\UnitaryCQS{q_2}{\opY}{\success}} \big)
		\end{array}}\\
		&\step \ConfigCQS{\Ket{\psi_3}}{c}{\begin{array}{l}
			\CondCQS{00}{00}{\success} \mid \CondCQS{00}{01}{\UnitaryCQS{q_2}{\opX}{\success}}\\
			\mid \CondCQS{00}{10}{\UnitaryCQS{q_2}{\opZ}{\success}} \mid \CondCQS{00}{11}{\UnitaryCQS{q_2}{\opY}{\success}}
		\end{array}}\\
		&\step \ConfigCQS{\Ket{\psi_3}}{c}{\begin{array}{l}
			\success \mid \CondCQS{00}{01}{\UnitaryCQS{q_2}{\opX}{\success}}\\
			\mid \CondCQS{00}{10}{\UnitaryCQS{q_2}{\opZ}{\success}} \mid \CondCQS{00}{11}{\UnitaryCQS{q_2}{\opY}{\success}}
		\end{array}} \tag*{\qed}
	\end{align*}
\end{exa}


\subsection{A Calculus for Open Quantum Systems}
\label{sec:OQS}

The algebra of quantum processes (\qCCS) is first introduced in \cite{FengDuanJiYing07} and further investigated \eg in \cite{YingFengDuanJi09,FengDuanYing12,YasudaKubotaKakutani14} as a process calculus for quantum based systems and to study observational equivalences in the quantum setting or in \cite{KubotaKakutaniKatoKawanoSakurada12,KubotaKakutaniKatoKawanoSakurada13} to study quantum crypto protocols.
As \qCCS is designed to model open systems, its states are described by density matrices or operators.
We are mainly interested in the variant of \qCCS presented in \cite{YingFengDuanJi09}, because it has the rare feature of introducing a quantum based calculus without a probabilistic transition system.
Indeed earlier as well as later variants of \qCCS \eg in \cite{FengDuanJiYing07,FengDuanYing12} use probabilistic transition systems.
The main reason for probabilistic transition systems in most quantum based systems is measurement, since its outcome is often a probability distribution.
In \cite{YingFengDuanJi09} measurement can be performed by a super-operator and the resulting probability distribution on potentially different measurement results is captured in the density matrix that represents the state after measurement.
Since they refrain from providing a measurement-operator, they can introduce a non-probabilistic transition system.
Unfortunately, without a separate operator for measurement there is no way in \cite{YingFengDuanJi09} to directly get the results of measurement; although the resulting alteration of the state does of course influence the further behaviour.
Remember that the state of a qubit cannot be read but only measured, so it is not possible to extract this information directly from the state after measurement.
Because of that, we add for \OQS an additional operator, a conditional to compare binary numbers and the outcome of measurement, to the syntax of \qCCS as presented in \cite{FengDuanJiYing07}.

\begin{defi}[\OQS]
	The \OQS \emph{terms}, denoted by $ \procOQS $, are given by:
	\begin{align*}
		P & \deffTerms \RecOQS{A}{\tilde{x}} \sep \nilOQS \sep \tau.P \sep \opE\List{X}.P \sep \InpOQS{c}{x}{P} \sep \OutOQS{c}{x}{P}\\
		& \sep P + P \sep P \parOQS P \sep \ResOQS{P}{L} \sep \CondOQS{bv}{e}{P}
	\end{align*}
	where
	\begin{align*}
		e \deffTerms bv \sep \opM\List{X}
	\end{align*}
	The \OQS \emph{configurations} $ \configOQS $ are given by $ \ConfigOQS{P}{\rho} $, where $ P \in \procOQS $ and $ \rho \in \partDen $.
\end{defi}

Process constants $ \RecOQS{A}{\tilde{x}} $, where $ \tilde{x} = x_1, \ldots, x_n $ is a sequence of pairwise distinct quantum variables, allow recursive definitions of terms.
An inactive process is denoted by $ \nilOQS $ and the term $ \tau.P $ executes the silent action and proceeds as $ P $.
The application of a super-operator $ \opE $ on the qubits in the finite set $ X \subseteq \var $ is performed by the term $ \opE\List{X}.P $.
The terms $ \InpOQS{c}{x}{P} $ and $ \OutOQS{c}{x}{P} $ model input and output on channel $ c \in \names $ to transfer a single qubit $ x \in \var $.
Choice and parallel composition are obtained from CCS and given by $ P + P $ and $ P \parOQS P $.
The term $ \ResOQS{P}{L} $ restricts the scope of all channels within $ L \subseteq \names $ to $ P $.
Finally, the conditional $ \CondOQS{bv}{e}{P} $ continues as $ P $ if either $ bv $ and $ e $ are the same binary number or $ bv $ is the binary number that results from measuring \wrt the standard basis the finite set of qubits $ X \subseteq \var $.
We use $ \opM $ to denote the super-operator for measurement in the standard base.

By slightly abusing notation, we use $ \var $ to also denote the current set of qubit names of a given density matrix $ \rho $.
The variable $ x $ is bound in $ P $ by $ \InpOQS{c}{x}{P} $ and the channels in $ L $ are bound in $ P $ by $ \ResOQS{P}{L} $.
A variable/channel is free if it is not bound.
Let $ \FreeChan{P} $ and $ \FreeQubits{P} $ denote the sets of free channels and free qubits in $ P $, respectively.
For each process constant scheme $ A $, a defining equation $ \RecOQS{A}{\tilde{x}} \stackrel{def}{=} P $ with $ P \in \procOQS $ and $ \FreeQubits{P} \subseteq \tilde{x} $ is assumed.
As done in \cite{YingFengDuanJi09}, we require the following two conditions:
\begin{align*}
	\OutOQS{c}{x}{P} \in \procOQS &\text{ implies } x \notin \FreeQubits{P} & \tag{Cond1}\label{condA}\\
	P \parOQS Q \in \procOQS &\text{ implies } \FreeQubits{P} \cap \FreeQubits{Q} = \emptyset & \tag{Cond2}\label{condB}
\end{align*}
These conditions ensure the no-cloning principle of qubits within \qCCS and \OQS.

The semantics of \OQS is defined by the inference rules in Figure~\ref{fig:semanticsOQS}.
We start with a labelled variant of the semantics from \cite{YingFengDuanJi09} for \qCCS, add the Rule~\ruleCondOQS for the new conditional, and then add the Rule~\ruleRedOQS to obtain a reduction semantics.
We omit the symmetric forms of the rules \ruleChoiceOQS, \ruleIntlOQS, and \ruleCommOQS.
Let $ \ChanNames{\alpha} $ return the possibly empty set of channels in the label $ \alpha $.

\begin{figure}[t]
	\centering
	\begin{displaymath}\begin{array}{c}
		\ruleTauOQS \; \ConfigOQS{\tau.P}{\rho} \LabelledStep{\tau} \ConfigOQS{P}{\rho}
		\hspace{1.5em}
		\ruleInputOQS \; \ConfigOQS{\InpOQS{c}{x}{P}}{\rho} \LabelledStep{\InpLabelOQS{c}{y}} \ConfigOQS{P\Set{\Subst{y}{x}}}{\rho} \quad y \notin \FreeQubits{\InpOQS{c}{x}{P}}
		\vspace{0.25em}\\
		\ruleOutputOQS \; \ConfigOQS{\OutOQS{c}{x}{P}}{\rho} \LabelledStep{\OutLabelOQS{c}{x}} \ConfigOQS{P}{\rho}
		\hspace{2em}
		\ruleOperOQS \; \ConfigOQS{\opE\List{X}.P}{\rho} \LabelledStep{\tau} \ConfigOQS{P}{\opE_X\Tuple{\rho}}
		\vspace{0.25em}\\
		\ruleChoiceOQS \; \dfrac{\ConfigOQS{P}{\rho} \LabelledStep{\alpha} \ConfigOQS{P'}{\rho'}}{\ConfigOQS{P + Q}{\rho} \LabelledStep{\alpha} \ConfigOQS{P'}{\rho'}}
		\hspace{1.5em}
		\ruleDefOQS \; \dfrac{\ConfigOQS{P\Set{\Subst{\tilde{y}}{\tilde{x}}}}{\rho} \LabelledStep{\alpha} \ConfigOQS{P'}{\rho'}}{\ConfigOQS{\RecOQS{A}{\tilde{y}}}{\rho} \LabelledStep{\alpha} \ConfigOQS{P'}{\rho'}} \quad \RecOQS{A}{\tilde{x}} \defeq P
		\vspace{0.25em}\\
		\ruleResOQS \; \dfrac{\ConfigOQS{P}{\rho} \LabelledStep{\alpha} \ConfigOQS{P'}{\rho'}}{\ConfigOQS{\ResOQS{P}{L}}{\rho} \LabelledStep{\alpha} \ConfigOQS{\ResOQS{P'}{L}}{\rho'}} \quad \ChanNames{\alpha} \cap L = \emptyset
		\vspace{0.25em}\\
		\ruleIntlOQS \; \dfrac{\ConfigOQS{P}{\rho} \LabelledStep{\alpha} \ConfigOQS{P'}{\rho'}}{\ConfigOQS{P \parOQS Q}{\rho} \LabelledStep{\alpha} \ConfigOQS{P' \parOQS Q}{\rho'}} \quad \text{if } \alpha = \InpLabelOQS{c}{x} \text{ then } x \notin \FreeQubits{Q}
		\vspace{0.25em}\\
		\ruleCommOQS \; \dfrac{\ConfigOQS{P}{\rho} \LabelledStep{\InpLabelOQS{c}{x}} \ConfigOQS{P'}{\rho} \quad \ConfigOQS{Q}{\rho} \LabelledStep{\OutLabelOQS{c}{x}} \ConfigOQS{Q'}{\rho}}{\ConfigOQS{P \parOQS Q}{\rho} \LabelledStep{\tau} \ConfigOQS{P' \parOQS Q'}{\rho}}
		\hspace{2em}
		\ruleRedOQS \; \dfrac{\ConfigOQS{P}{\rho} \LabelledStep{\tau} \ConfigOQS{P'}{\rho'}}{\ConfigOQS{P}{\rho} \step \ConfigOQS{P'}{\rho'}}
		\vspace{0.25em}\\
		\ruleCondOQS \; \dfrac{\begin{array}{c} \ConfigOQS{P}{\rho} \LabelledStep{\alpha} \ConfigOQS{P'}{\rho'}\\ \left( e = b' \wedge b = b' \wedge \rho' = \rho \right) \vee \left( e = \opM\List{X} \wedge b \in \opM\List{X}\Tuple{\rho} \wedge \rho' = \opM\List{X}\Tuple{\rho} \right) \end{array}}{\ConfigOQS{\CondOQS{b}{e}{P}}{\rho} \LabelledStep{\alpha} \ConfigOQS{P'}{\rho'}}
	\end{array}\end{displaymath}
	\caption{Semantics of \OQS}
	\label{fig:semanticsOQS}
\end{figure}

Rule~\ruleOperOQS implements the application of a super-operator $ \opE $.
It updates the state of the configuration by applying $ \opE $.
To simplify the definition of a reduction semantics, we use (in contrast to \cite{YingFengDuanJi09}) the label $ \tau $.

Rule~\ruleInputOQS ensures that the received qubits are fresh in the continuation of the input.
The rules \ruleIntlOQS and its symmetric rule \ruleIntrOQS forbid to receive qubits within parallel contexts that do posses this qubit.
Rule \ruleResOQS allows to do a step under a restriction.
Rule~\ruleCondOQS allows a step of the continuation of a conditional if its condition is satisfied.
Therefore either $ b $ and $ e $ need to be the same binary number and then the state $ \rho $ is not updated or $ e = \opM\List{X} $ and $ b $ is one of the binary numbers that results from measuring in the standard basis the qubits in $ X $ in the state $ \rho $ with non-zero probability.
In the latter case the state $ \rho $ has to be updated according to the measurement operation.
For instance if $ \rho $ is a 2-qubit system with the qubits $ q_0, q_1 $ in state $ \OuterProduct{0}{0} \tensorProd \OuterProduct{1}{1} $ then $ \CondOQS{01}{\opM\List{q_0, q_1}}{\tau.\nilOQS} \step \nilOQS $ but $ \CondOQS{b}{\opM\List{q_0, q_1}}{P} $ cannot reduce for any $ b \neq 01 $.
Note that to decide whether $ b \in \opM\List{X}\Set{\rho} $ the system indeed has to measure the qubits; it is not sufficient to apply any super-operator on $ \rho $ even if it has the same effect on $ \rho $ as measurement.
Since we cannot read the qubit we have to measure it, to learn anything about its state.
The other rules are self-explanatory.

Similar to \CQS, structural congruence for \OQS is the smallest congruence containing $ \alpha $-equivalence that is closed under the following rules:
\begin{displaymath}
	P \parOQS \nilOQS \equiv P \quad\quad P \parOQS Q \equiv Q \parOQS P \quad\quad P \parOQS \left( Q \parOQS R \right) \equiv \left( P \parOQS Q \right) \parOQS R
\end{displaymath}
Moreover, $ \ConfigOQS{P}{\rho} \equiv \ConfigOQS{Q}{\rho} $ if $ P \equiv Q $ or if $ \ConfigOQS{Q}{\rho} $ is obtained from $ \ConfigOQS{P}{\rho} $ by alpha conversion on the qubit names in $ \var $.


\section{Encodings and Quality Criteria}
\label{sec:criteria}

Let $ \sourceLang = \langle \sourceConfig, \longmapsto_\source \rangle $ and $ \targetLang = \langle \targetConfig, \longmapsto_\target \rangle $ be two process calculi, denoted as \emph{source} and \emph{target} language.
An \emph{encoding} from $ \sourceLang $ into $ \targetLang $ is a function $ \Enc{\cdot} : \sourceConfig \to \targetConfig $.
We often use $ S, S', \ldots $ and $ T, T', \ldots $ to range over $\sourceConfig$ and $\targetConfig$, respectively.

To analyse the quality of encodings and to rule out trivial or meaningless encodings, they are augmented with a set of quality criteria. In order to provide a general framework, Gorla in \cite{gorla10} suggests five criteria well suited for language comparison. They are divided into two structural and three semantic criteria. The structural criteria include
\begin{enumerate*}[(1)]
	\item \emph{compositionality} and
	\item \emph{name invariance}. The semantic criteria are
	\item \emph{operational correspondence},
	\item \emph{divergence reflection}, and
	\item \emph{success sensitiveness}.
\end{enumerate*}
We start with these criteria for classical systems.

Note that a behavioural relation $ \preceq $ on the target is assumed for operational correspondence.
Moreover, $ \preceq $ needs to be success sensitive, \ie $ T_1 \preceq T_2 $ implies $ \ReachBarb{T_1}{\success} $ iff $ \ReachBarb{T_2}{\success} $.
As discussed in \cite{petersGlabbeek15}, we pair operational correspondence as of \cite{gorla10} with correspondence simulation.

\begin{defi}[Correspondence Simulation, \cite{petersGlabbeek15}]
	\label{def:correspondenceSimulation}
	A relation $ \mathcal{R} $ is a \emph{(weak) labelled correspondence simulation} if for each $ \left( T_1, T_2 \right) \in \mathcal{R} $:
	\begin{compactitem}
		\item For all $ T_1 \LabelledStep{\alpha} T_1' $, there exists $ T_2' $ such that $ T_2 \LabelledStep{\alpha} T_2' $ and $ \left( T_1', T_2' \right) \in \mathcal{R} $.
		\item For all $ T_2 \LabelledStep{\alpha} T_2' $, there exists $ T_1'', T_2'' $ such that $ T_1 \steps\LabelledStep{\alpha} T_1'' $, $ T_2' \steps T_2'' $, and $ \left( T_1'', T_2'' \right) \in \mathcal{R} $.
		\item $ \ReachBarb{T_1}{\success} $ iff $ \ReachBarb{T_2}{\success} $.
	\end{compactitem}
	$ T_1 $ and $ T_2 $ are \emph{correspondence similar}, denoted as $ T_1 \preceq T_2 $, if a correspondence simulation relates them.
\end{defi}

Intuitively, an encoding is compositional if the translation of an operator is the same for all occurrences of that operator in a term. Hence, the translation of that operator can be captured by a context that is allowed in \cite{gorla10} to be parametrised on the free names of the respective source configuration.

\begin{defi}[Compositionality, \cite{gorla10}]
	\label{def:compositionality}
	The encoding $ \Enc{\cdot} $ is \emph{compositional} if, for every operator $ \mathbf{op} $ with arity $ n $ of $ \sourceLang $ and for every subset of names $ N $, there exists a context $ \Context{N}{\mathbf{op}}{\hole_1, \ldots, \hole_n} $ such that, for all $ S_1, \ldots, S_n $ with $ \FreeChan{S_1} \cup \ldots \cup \FreeChan{S_n} = N $, it holds that $ \Enc{\mathbf{op}\left( S_1, \ldots, S_n \right)} = \Context{N}{\mathbf{op}}{\Enc{S_1}, \ldots, \Enc{S_n}} $.
\end{defi}

Name invariance ensures that encodings are independent of specific variables in the source.
In \cite{gorla10} name invariance is defined modulo a so-called renaming policy.
Since our encoding in Section~\ref{sec:encoding} translates variables to themselves and name invariance is not relevant for the separation result in Section~\ref{sec:separation}, we do not need a renaming policy.
This simplifies the definition of name invariance such that an encoding is name invariant if it preserves and reflects substitutions.

\begin{defi}[Name Invariance]
	\label{def:nameInvariance}
	The encoding $ \Enc{\cdot} $ is \emph{name invariant} if, for every $ S \in \sourceConfig $ and every substitution $ \gamma $ on names, it holds that $ \Enc{S\gamma} = \Enc{S}\gamma $.
\end{defi}

The first semantic criterion is operational correspondence. It consists of a soundness and a completeness condition. \emph{Completeness} requires that every computation of a source term can be emulated by its translation. \emph{Soundness} requires that every computation of a target term corresponds to some computation of the corresponding source term.

\begin{defi}[Operational Correspondence, \cite{gorla10}]
	\label{def:operationalCorrespondence}
	An encoding $ \Enc{\cdot} $ is \emph{operationally corresponding} \wrt $ \preceq $ if it is:\\
	\begin{tabular}{l l}
		Complete: & For all $ S \steps S' $, there exists $ T $ such that $ \Enc{S} \steps T $ and $ \Enc{S'} \preceq T $.\\
		Sound: & For all $ \Enc{S} \steps T $, there exists $ S', T' $ such that $ S \steps S' $, $ T \steps T' $,\\
		& and $ \Enc{S'} \preceq T' $.
	\end{tabular}
\end{defi}

The next criterion concerns the role of infinite computations.

\begin{defi}[Divergence Reflection, \cite{gorla10}]
	An encoding $ \Enc{\cdot} $ \emph{reflects divergence} if, for every $ S $, $ \Enc{S} \infiniteSteps $ implies $ S \infiniteSteps $.
\end{defi}

The last criterion links the behaviour of source terms to the behaviour of their encodings.
Success sensitiveness requires that source configurations reach success if and only if their literal translations do.

\begin{defi}[Success Sensitiveness, \cite{gorla10}]
	$ \Enc{\cdot} $ is \emph{success sensitive} if, for every $ S $, $ \ReachBarb{S}{\success} $ iff $ \ReachBarb{\Enc{S}}{\success} $.
\end{defi}

Moreover, $ \preceq $ needs to be success sensitive, \ie $ T_1 \preceq T_2 $ implies $ \ReachBarb{T_1}{\success} $ iff $ \ReachBarb{T_2}{\success} $, as required by Definition~\ref{def:correspondenceSimulation}.
Without this requirement the relation that is induced---as described in  \cite{petersGlabbeek15,peters19}---by operational correspondence between the source and target is trivial without some notion of barbs.
To some up, we use the following notion of \emph{good} encoding, where good refers to classical criteria only.

\begin{defi}[Classical Criteria]
	\label{def:good}
	The encoding $ \Enc{\cdot} $ is good, if it is compositional, name invariant, operational corresponding \wrt $ \preceq $, divergence reflecting, and success sensitive, where $ \preceq $ is success sensitive.
\end{defi}

There are several other criteria for classical systems that we could have considered (cf.~\cite{peters19}).
Since \CQS is a typed language, we may consider a criterion for types as discussed \eg in \cite{KouzapasPerezYoshida16}. As only one language is typed, it suffices to require that the encoding is defined for all terms of the source language.
We could also consider a criterion for the preservation of distributability as discussed \eg in \cite{petersNestmannGoltz13}, since distribution and communication between distributed locations is of interest.
Indeed our encoding satisfies this criterion, because it translates the parallel operator homomorphically.
However, already the basic framework of Gorla, on that we rely here, suffices to observe principal design principles of quantum based systems as we discuss with the no-cloning property in Section~\ref{sec:criteriaQBS}.


\section{Encoding Quantum Based Systems}
\label{sec:encoding}

Our encoding, from well-typed \CQS-configurations into \OQS-configurations that satisfy the conditions~\ref{condA} and \ref{condB}, is given by Definition~\ref{def:encoding}.

\begin{defi}[Encoding $ \Enc{\cdot} $ from \CQS into \OQS]
	\begin{align*}
		\Enc{\ConfigCQS{\sigma}{\phi}{P}} &= \ConfigOQS{\ResOQS{\Enc{P}}{\phi}}{\rho_{\sigma}}\\
		\Enc{\boxplus_{0 \leq i < 2^r} p_i \bullet \ConfigCQS{\sigma_i}{\phi}{P\Set{\Subst{\Binary{i}}{v}}}} &= \ConfigOQS{\ResOQS{{\EncDist{q_0, \ldots, q_{r - 1}}{v}{\Enc{P}}}}{\phi}}{\rho_{\boxplus}}\\
		\Enc{\nilCQS} &= \nilOQS\\
		\Enc{P \mid Q} &= \Enc{P} \parOQS \Enc{Q}\\
		\Enc{\InpCQS{c}{x}{P}} &= \InpOQS{c}{x}{\Enc{P}}\\
		\Enc{\OutCQS{c}{q}{P}} &= \OutOQS{c}{q}{\Enc{P}}\\
		\Enc{\UnitaryCQS{\tilde{q}}{U}{P}} &= U\List{\tilde{q}}.\Enc{P}\\
		\Enc{\MeasCQS{v}{\tilde{q}}{P}} &= \opM\List{\tilde{q}}.\EncDist{\tilde{q}}{v}{\Enc{P}}\\
		\Enc{\NewCQS{c}{P}} &= \tau.{\left( \ResOQS{\Enc{P}}{\Set{c}} \right)}\\
		\Enc{\QubitCQS{x}{P}} &= \opE_{\Ket{0}}\List{\var}.{\left( \Enc{P}\Set{\Subst{q_{\Length{\var}}}{x}} \right)}\\
		\Enc{\CondCQS{bv}{bv'}{P}} &= \CondOQS{bv}{bv'}{\tau.\Enc{P}}\\
		\Enc{\success} &= \success
	\end{align*}
	where $ \rho_{\sigma} = \OuterProduct{\psi}{\psi} $ for $ \sigma = \Ket{\psi} $, $ \rho_{\boxplus} = \sum_i p_i\OuterProduct{\psi_i}{\psi_i} $ for $ \sigma_i = \Ket{\psi_i} $,
	\begin{align*}
		\EncDist{\tilde{q}}{v}{Q} =
		\begin{cases}
			Q & \text{, if } \tilde{q} \text{ is empty}\\
			\!\!\begin{array}[t]{l}
				\CondOQS{0..0}{\opM\List{\tilde{q}}}{\tau.Q\Set{\Subst{0..0}{v}}} + \ldots +{}\\
				\CondOQS{\Binary{2^{\Length{\tilde{q}} - 1}}}{\opM\List{\tilde{q}}}{\tau.Q\Set{\Subst{\Binary{2^{\Length{\tilde{q}} - 1}}}{v}}}
			\end{array} & \text{, otherwise}
		\end{cases},
	\end{align*}
	$ \opE_{\Ket{0}}\List{\var} $ adds a new qubit $ q_{\Length{\var}} $ initialised with $ 0 $ to the current state $ \rho $.
	\label{def:encoding}
\end{defi}

The translation of configurations maps the vector $ \sigma $ to the density matrix $ \rho_{\sigma} $ (obtained by the outer product) and restricts all names in $\phi$ to the translation of the sub-term.
In the translation of probability distributions, the state $ \rho_{\boxplus} $ is the sum of the density matrices obtained from the $ \sigma_i $ multiplied with their respective probability.
Again, the names in $ \phi $ are restricted in the translation.
The nondeterminism in choosing one of the possible branches of the probability distribution in \CQS by \ruleRProbCQS is translated into the \OQS-choice $ \EncDist{\tilde{q}}{v}{\Enc{P}} $ with $ \tilde{q} = q_0, \ldots, q_{r - 1} $, where each case is guarded by a conditional to compare to a possible outcome of measurement followed by the continuation with a substitution to hand the result of measurement to the process.
Note that, the translation of a configuration $ \ConfigCQS{\sigma}{\phi}{P} $ is a special case of the second line.
A practical motivated encoding example using such a \OQS-choice is given in Example~\ref{exa:teleportation}.

The application of unitary transformations and the creation of new qubits are translated to the corresponding super-operators.
Measurement is translated into the super-operator for measurement followed by the choice $ \EncDist{\tilde{q}}{v}{\Enc{P}} $ over the branches of the possible outcomes of measurement, \ie after the first measurement the translation is similar to the translation of a probability distribution in the second case.
Note that we measure twice in this translation.
The outer measurement---that is a super-operator for measurement---dissolves entanglement on the measured qubits and ensures that the density matrix after this first measurement is the sum of the density matrices of the respective cases in the distribution (compare with $ \rho_{\boxplus} $ and Example~\ref{exa:measurment}).
The measurements within $ \EncDist{\tilde{q}}{v}{\Enc{P}} $---that are not performed by a super-operator but require to indeed physically measure the qubits---then check whether the respective case $ i $ occurs with non-zero probability and adjust the density matrix to this result of measurement if case $ i $ is picked.
The creation of new channel names is translated to restriction, where a $\tau$-guard simulates the step that is necessary in \CQS to create a new channel.
The restriction ensures that this new name cannot be confused with any other translated source term name.
Since in the derivative of a source term step creating a new channel the new channel is added to $ \phi $ in the configuration, we restrict all channels in $ \phi $.
A condition in \CQS is translated to a conditional in \OQS.
We add a $ \tau $ to guard the continuation of the conditional in the target, since resolving a conditional in \CQS (in contrast to \OQS) requires a step.
The remaining translations are homomorphic.

\begin{exa}
	\label{exa:measurment}
	Consider the \CQS-configuration $ S = \ConfigCQS{\sigma}{\phi}{\MeasCQS{v}{q_0}{P}} $, where $ \sigma = q_0, q_1 = \frac{1}{\sqrt{2}}\Ket{00} + \frac{1}{\sqrt{2}}\Ket{11} = \Ket{\psi} $ consists of two entangled qubits.
	By Rule~\ruleRMeasureCQS in Figure~\ref{fig:semanticsCQS}, $ S \step S' = \frac{1}{2} \bullet \ConfigCQS{\sigma = q_0, q_1 = \Ket{00}}{\phi}{P\Set{\Subst{00}{v}}} \boxplus \frac{1}{2} \bullet \ConfigCQS{\sigma = q_0, q_1 = \Ket{11}}{\phi}{P\Set{\Subst{11}{x}}} $, where we omitted branches with probability zero.

	By Definition~\ref{def:encoding}, $ \Enc{S} = \ConfigOQS{\ResOQS{{\left( \opM\List{q_0}.\EncDist{q_0}{v}{\Enc{P}} \right)}}{\phi}}{\rho} $ with $ \rho = \OuterProduct{\psi}{\psi} $.
	By the rules \ruleOperOQS and \ruleRedOQS in Figure~\ref{fig:semanticsOQS}, then $ \Enc{S} \step T = \ConfigOQS{\ResOQS{\EncDist{q_0}{x}{\Enc{P}}}{\phi}}{\StateTrans{\opM}{q_0}{\rho}} $.
	Accordingly, the probability distribution in $ S' $ is mapped on a choice in $ T $.
	The outer measurement $ \opM\List{q_0} $ resolves the entanglement and yields a density matrix that is the sum of the density matrices of the choice branches, \ie $ \StateTrans{\opM}{q_0}{\rho} = \OuterProduct{00}{00}\rho\Adjoint{\OuterProduct{00}{00}} + \OuterProduct{11}{11}\rho\Adjoint{\OuterProduct{11}{11}} $.
	\qed
\end{exa}

By analysing the encoding function, we observe that for all source terms the type system of \CQS ensures that their literal translation satisfies the conditions~\ref{condA} and \ref{condB}. Hence, the encoding is defined on all source terms.

\begin{cor}
	For all $ S \in \configCQS $ the term $ \Enc{S} $ is defined.
\end{cor}

Before we can start to prove the quality of our encoding, \ie that it satisfies the criteria in Definition~\ref{def:good}, we have to fix a relation $ \preceq $ on the target language \OQS that is used in the definition of operational correspondence in Definition~\ref{def:operationalCorrespondence}.
We instantiate $ \preceq $ with correspondence similarity as given in Definition~\ref{def:correspondenceSimulation}.
In the literature, operational correspondence is often considered \wrt a bisimulation on the target; simply because bisimilarity is a standard behavioural equivalence in process calculi, whereas correspondence simulation is not.
For our encoding, we cannot use bisimilarity.

\begin{exa}
	\label{exa:correspondenceSimulation}
	Consider $ S = \ConfigCQS{\sigma}{c}{\MeasCQS{v}{q}{P} \mid Q} $, where $ S $ is a 1-qubit system with $ \sigma = q = \Ket{+} $ and $ P, Q \in \procCQS $ with $ \FreeChan{P} = \Set{ c } = \FreeChan{Q} $ and $ v \notin \FreeBinV{Q} $.
	By the rules~\ruleRMeasureCQS and \ruleRParCQS of Figure~\ref{fig:semanticsCQS},
	\begin{align*}
		S &\step S' = \frac{1}{2} \bullet \ConfigCQS{\sigma = q = \Ket{0}}{c}{P\Set{\Subst{0}{v}} \mid Q} \boxplus \frac{1}{2} \bullet \ConfigCQS{\sigma = q = \Ket{1}}{c}{P\Set{\Subst{1}{v}} \mid Q},
	\end{align*}
	\ie \ruleRParCQS pulls the parallel component $ Q $ into the probability distribution that results from measuring $ q $.
	Since our encoding is compositional---and indeed we require compositionality, the translation $ \Enc{S} $ behaves slightly differently.
	By Definition~\ref{def:encoding}, $ \Enc{S} = \ConfigOQS{\ResOQS{{\left( \opM\List{q}.\EncDist{q}{v}{\Enc{P}} \parOQS \Enc{Q} \right)}}{\Set{c}}}{\rho} $, where here $ \EncDist{q}{v}{\Enc{P}} = \CondOQS{0}{\opM\List{q}}{\tau.\Enc{P}\Set{\Subst{0}{v}}} + \CondOQS{1}{\opM\List{q}}{\tau.\Enc{P}\Set{\Subst{1}{v}}} $, $ \rho = \OuterProduct{+}{+} $, and $ \Enc{S'} = \ConfigOQS{\ResOQS{\EncDist{q}{v}{\Enc{P} \parOQS \Enc{Q}}}{\Set{c}}}{\rho'} $ with $ \rho' = \frac{1}{2}\OuterProduct{0}{0} + \frac{1}{2}\OuterProduct{1}{1} $.
	By Figure~\ref{fig:semanticsOQS}, $ \Enc{S} \step T = \ConfigOQS{\ResOQS{{\left( \EncDist{q}{v}{\Enc{P}} \parOQS \Enc{Q} \right)}}{\Set{c}}}{\rho'} $, because $ \StateTrans{\opM}{q}{\rho} = \rho' $.
	Unfortunately, $ \Enc{S'} $ and $ T $ are not bisimilar.
	As a counterexample consider $ P = \OutCQS{c}{q}{\nilCQS} $ and $ Q = \NewCQS{c'}{\InpCQS{c}{x}{\MeasCQS{v}{x}{\CondCQS{v}{0}{\success}}}} $.
	The problem is, that a step on $ \Enc{Q} $ in $ \Enc{S'} $ forces us to immediately pick a case and resolve the choice, whereas after performing the same step on $ \Enc{Q} $ in $ T $ all cases of the choice remain available.
	After emulating the first step of $ \Enc{Q} $ in $ \Enc{S'} $, either we reach a configuration that has to reach success eventually or we reach a configuration that cannot reach success; whereas there is just one way to do the respective step in $ T $ and in the resulting configuration success may or may not be reached depending on the next step.
	Fortunately, $ \Enc{S'} $ and $ T $ are correspondence similar.
	\qed
\end{exa}

The encoding $ \Enc{\cdot} $ in Definition~\ref{def:encoding} emulates a source term step by exactly one step on the target, except for source term steps on \ruleRPermCQS that are not emulated at all.
Steps on \ruleRPermCQS are necessary in \CQP and \CQS, because they assume that unitary transformations and measurement is always applied to the first $ r $ qubits.
With \ruleRPermCQS the quantum register is permuted to bring the relevant qubits to the front.
In \OQS this is not necessary.
Lemma~\ref{lem:permutation} captures this observation, by showing that the translation of source term steps on \ruleRPermCQS are indistinguishable in the target modulo $ \preceq $.

\begin{lem}
	If $ S \step S' $ is by \ruleRPermCQS, then $ \Enc{S} \preceq \Enc{S'} $ and $ \Enc{S'} \preceq \Enc{S} $.
	\label{lem:permutation}
\end{lem}

\begin{proof}
	Since $ S \step S' $ is by \ruleRPermCQS, there are $ q_0, \ldots, q_{n - 1}, \psi, \phi, P, \pi $, and $ \Pi $ such that:
	\begin{align*}
		S = \ConfigCQS{q_0, \ldots, q_{n - 1} = \Ket{\psi}}{\phi}{P}
		\quad \text{and} \quad
		S' = \ConfigCQS{q_{\pi(0)}, \ldots, q_{\pi(n - 1)} = \Pi\Ket{\psi}}{\phi}{P\pi}
	\end{align*}
	Let $ \Ket{\psi'} = \Pi\Ket{\psi} $ be the state that results from applying the unitary transformation $ \Pi $.
	Then $ \Enc{S} = \ConfigOQS{\ResOQS{\Enc{P}}{\phi}}{\rho} $ with $ \rho = \OuterProduct{\psi}{\psi} $ and $ \Enc{S'} = \ConfigOQS{\ResOQS{\Enc{P\pi}}{\phi}}{\rho'} $ with $ \rho' = \OuterProduct{\psi'}{\psi'} $.
	Note that, $ \Pi_{q_0, \ldots, q_{n - 1}}(\rho) = \rho' $, where $ \Pi\List{q_0, \ldots, q_{n - 1}} $ is the super-operator obtained from the unitary transformation $ \Pi $.
	By Lemma~\ref{lem:qubitInvariance}, $ \Enc{S'} = \ConfigOQS{\ResOQS{\Enc{P\pi}}{\phi}}{\rho'} = \ConfigOQS{\ResOQS{{\left( \Enc{P}\pi \right)}}{\phi}}{\rho'} $.
	Since \OQS-terms such as $ \ResOQS{\Enc{P}}{\phi} $ and $ \ResOQS{{\left( \Enc{P}\pi \right)}}{\phi} $ do not address qubits by their position in the density matrix but their name, $ \mathcal{R} = \left\lbrace \left( \ConfigOQS{Q}{\rho_Q}, \ConfigOQS{Q\pi}{\rho_Q'} \right) \mid \Pi_{q_0, \ldots, q_{n - 1}}(\rho_Q) = \rho_Q' \right\rbrace $ is a bisimulation and thus $ \mathcal{R} $ as well as $ \mathcal{R}^{-1} $ are correspondence simulations.
	Then $ \Enc{S} \preceq \Enc{S'} $ and $ \Enc{S'} \preceq \Enc{S} $.
\end{proof}

Since structural congruence is defined similarly on \CQS and \OQS, does consider in both cases only alpha conversion, the inactive process, and parallel composition, and since $ \Enc{\cdot} $ translates the inactive process and parallel composition homomorphically, the encoding preserves structural congruence.

\begin{lem}[Preservation of Structural Congruence, $ \Enc{\cdot} $]
	\label{lem:preservationSC}
	\begin{align*}
		&\forall C_1, C_2 \in \configCQS \logdot C_1 \equiv C_2 \text{ implies } \Enc{C_1} \equiv \Enc{C_2} & \text{and}\\
		&\forall S_1, S_2 \in \procCQS \logdot S_1 \equiv S_2 \text{ implies } \Enc{S_1} \equiv \Enc{S_2}
	\end{align*}
\end{lem}

\begin{proof}
	By straightforward induction on the rules of structural congruence.
\end{proof}

By \cite{gorla10}, good encodings are allowed to use a renaming policy that structures the way in that the translations of source term names are used in target terms and how to treat names that are introduced by the encoding function.
The encoding $ \Enc{\cdot} $ simply translates names by themselves and does not introduce any other names.
Because of that, we can choose the identity relation as renaming policy and are able to prove a stronger variant of name invariance.
Note that, name invariance considers substitutions on names only.

\begin{lem}[Name Invariance, $ \Enc{\cdot} $]
	\label{lem:nameInvariance}
	Let $ \gamma $ be a substitution on names.
	\begin{align*}
		\forall S_C \in \configCQS \logdot \Enc{S_C\gamma} = \Enc{S_C}\gamma \quad \text{ and } \quad \forall S \in \procCQS \logdot \Enc{S\gamma} = \Enc{S}\gamma
	\end{align*}
\end{lem}

\begin{proof}
	Assume a substitution $ \gamma $ on names.
	Let $ S_C = \ConfigCQS{\sigma}{\phi}{S} $.
	Then $ S_C\gamma = \ConfigCQS{\sigma}{\phi\gamma}{S\gamma} $.
	Moreover, let $ \sigma = \Ket{\psi} $ and $ \rho = \OuterProduct{\psi}{\psi} $.
	Then $ \Enc{S_C\gamma} = \ConfigOQS{\ResOQS{\Enc{S\gamma}}{(\phi\gamma)}}{\rho} = \ConfigOQS{\ResOQS{{\left( \Enc{S}\gamma \right)}}{(\phi\gamma)}}{\rho} = \ConfigOQS{\ResOQS{\Enc{S}}{\phi}}{\rho}\gamma = \Enc{S_C}\gamma $ holds if $ \Enc{S\gamma} = \Enc{S}\gamma $.

	Similarly, let $ S_C = \boxplus_{0 \leq i < 2^r} p_i \bullet \ConfigCQS{\sigma_i}{\phi}{S\Set{\Subst{\Binary{i}}{v}}} $.
	Then we have $ S_C\gamma = \boxplus_{0 \leq i < 2^r} p_i \bullet \ConfigCQS{\sigma_i}{\phi\gamma}{S\Set{\Subst{\Binary{i}}{v}}\gamma} $.
	Moreover, let $ \sigma_i = \Ket{\psi_i} $, $ \rho = \sum_i pi_i \OuterProduct{\psi_i}{\psi_i} $, $ \tilde{q} = q_0, \ldots, q_{r - 1} $, and $ r = \Length{\tilde{q}} \leq n $.
	Then we have $ \Enc{S_C\gamma} = \ConfigOQS{\ResOQS{\EncDist{\tilde{q}}{v}{\Enc{S\gamma}}}{{\left( \phi\gamma \right)}}}{\rho} = \ConfigOQS{\ResOQS{\EncDist{\tilde{q}}{v}{\Enc{S}\gamma}}{{\left( \phi\gamma \right)}}}{\rho} = \ConfigOQS{\ResOQS{\EncDist{\tilde{q}}{v}{\Enc{S}}}{\phi}}{\rho}\gamma = \Enc{S_C}\gamma $ holds if $ \Enc{S\gamma} = \Enc{S}\gamma $.

	We show $ \Enc{S\gamma} = \Enc{S}\gamma $ by induction on the structure of $ S $.
	\begin{compactdesc}
		\item[Case $ S = \nilCQS $] In this case $ S\gamma = S $ and, thus, $ \Enc{S\gamma} = \Enc{S} = \nilOQS = \Enc{S}\gamma $.
		\item[Case $ S = P \mid Q $] In this case $ S\gamma = P\gamma \mid Q\gamma $. By the induction hypothesis, $ \Enc{P\gamma} = \Enc{P}\gamma $ and $ \Enc{Q\gamma} = \Enc{Q}\gamma $. Then $ \Enc{S\gamma} = \Enc{P\gamma} \parOQS \Enc{Q\gamma} = \Enc{P}\gamma \parOQS \Enc{Q}\gamma = {\left( \Enc{P} \parOQS \Enc{Q} \right)}\gamma = \Enc{S}\gamma $.
		\item[Case $ S = \InpCQS{c}{x}{P} $] In this case $ S\gamma = \InpCQS{{\left( c\gamma \right)}}{x}{{\left( P\gamma \right)}} $. By the induction hypothesis, $ \Enc{P\gamma} = \Enc{P}\gamma $. Then we have $ \Enc{S\gamma} = \InpOQS{{\left( c\gamma \right)}}{x}{\Enc{P\gamma}} = \InpOQS{{\left( c\gamma \right)}}{x}{{\left( \Enc{P}\gamma \right)}} = {\left( \InpOQS{c}{x}{\Enc{P}} \right)}\gamma = \Enc{S}\gamma $.
		\item[Case $ S = \OutCQS{c}{q}{P} $] In this case $ S\gamma = \OutCQS{{\left( c\gamma \right)}}{q}{{\left( P\gamma \right)}} $. By the induction hypothesis, $ \Enc{P\gamma} = \Enc{P}\gamma $. Then we have $ \Enc{S\gamma} = \OutOQS{{\left( c\gamma \right)}}{q}{\Enc{P\gamma}} = \OutOQS{{\left( c\gamma \right)}}{q}{{\left( \Enc{P}\gamma \right)}} = {\left( \OutOQS{c}{q}{\Enc{P}} \right)}\gamma = \Enc{S}\gamma $.
		\item[Case $ S = \UnitaryCQS{\tilde{q}}{U}{P} $] In this case $ S\gamma = \UnitaryCQS{\tilde{q}}{U}{{\left( P\gamma \right)}} $. By the induction hypothesis, $ \Enc{P\gamma} = \Enc{P}\gamma $. Then $ \Enc{S\gamma} = U\List{\tilde{q}}.\Enc{P\gamma} = U\List{\tilde{q}}.{\left( \Enc{P}\gamma \right)} = {\left( U\List{\tilde{q}}.\Enc{P} \right)}\gamma = \Enc{S}\gamma $.
		\item[Case $ S = \MeasCQS{v}{\tilde{q}}{P} $] In this case $ S\gamma = \MeasCQS{v}{\tilde{q}}{{\left( P\gamma \right)}} $.
		By the induction hypothesis, $ \Enc{P\gamma} = \Enc{P}\gamma $.
		Then $ \Enc{S\gamma} = \opM\List{\tilde{q}}.\EncDist{\tilde{q}}{v}{\Enc{P\gamma}} = \opM\List{\tilde{q}}.\EncDist{\tilde{q}}{v}{\Enc{P}\gamma} = {\left( \opM\List{\tilde{q}}.\EncDist{\tilde{q}}{v}{\Enc{P}} \right)}\gamma = \Enc{S}\gamma $.
		\item[Case $ S = \NewCQS{c}{P} $] In this case $ S\gamma = \NewCQS{d}{{\left( P'\gamma \right)}} $, where $ d $ is fresh and $ P' = P\Set{\Subst{d}{c}} $. By the induction hypothesis, $ \Enc{P'\gamma} = \Enc{P'}\gamma $ and $ \Enc{P'} = \Enc{P}\Set{\Subst{d}{c}} $. Then $ \Enc{S\gamma} = \tau.{\left( \ResOQS{\Enc{P'\gamma}}{\Set{d}} \right)} = {\left( \tau.{\left( \ResOQS{\Enc{P'}}{\Set{d}} \right)} \right)}\gamma = {\left( \tau.{\left( \ResOQS{\Enc{P}}{\Set{c}} \right)} \right)}\gamma = \Enc{S}\gamma $.
		\item[Case $ S = \QubitCQS{x}{P} $] In this case $ S\gamma = \QubitCQS{x}{{\left( P\gamma \right)}} $. By the induction hypothesis, $ \Enc{P\gamma} = \Enc{P}\gamma $. Then $ \Enc{S\gamma} = \opE_{\Ket{0}}\List{\var}.{\left( \Enc{P\gamma}\Set{\Subst{q_{\Length{\var}}}{x}} \right)} = {\left( \opE_{\Ket{0}}\List{\var}.{\left( \Enc{P}\Set{\Subst{q_{\Length{\var}}}{x}} \right)} \right)}\gamma = \Enc{S}\gamma $.
		\item[Case $ S = \CondCQS{bv}{bv'}{P} $] In this case $ S\gamma = \CondCQS{bv}{bv'}{P\gamma} $. By the induction hypothesis, $ \Enc{P\gamma} = \Enc{P}\gamma $. Then we have $ \Enc{S\gamma} = \left( \CondOQS{bv}{bv'}{\tau.\Enc{P\gamma}} \right) = \left( \CondOQS{bv}{bv'}{\tau.{\left( \Enc{P}\gamma \right)}} \right) = {\left( \CondOQS{bv}{bv'}{\tau.\Enc{P}} \right)}\gamma = \Enc{S}\gamma $.
			\qedhere
	\end{compactdesc}
\end{proof}

For the proof of operational correspondence, we also need qubit invariance, \ie that also substitutions on qubits are preserved and reflected by the encoding function.
The proof of qubit invariance is very similar to the proof of name invariance.

\begin{lem}[Qubit Invariance, $ \Enc{\cdot} $]
	\label{lem:qubitInvariance}
	Let $ \gamma $ be a substitution on qubit names.
	\begin{align*}
		\forall S_C \in \configCQS \logdot \Enc{S_C\gamma} = \Enc{S_C}\gamma \quad \text{ and } \quad \forall S \in \procCQS \logdot \Enc{S\gamma} = \Enc{S}\gamma
	\end{align*}
\end{lem}

\begin{proof}
	Assume a substitution $ \gamma $ on qubit names.
	Let $ S_C = \ConfigCQS{\sigma}{\phi}{S} $.
	Then $ S_C\gamma = \ConfigCQS{\sigma\gamma}{\phi}{S\gamma} $.
	Moreover, let $ \sigma = \Ket{\psi} $ and $ \rho = \OuterProduct{\psi}{\psi} $.
	Then $ \Enc{S_C\gamma} = \ConfigOQS{\ResOQS{\Enc{S\gamma}}{\phi}}{\rho\gamma} = \ConfigOQS{\ResOQS{{\left( \Enc{S}\gamma \right)}}{\phi}}{\rho\gamma} = \ConfigOQS{\Enc{S}}{\rho}\gamma = \Enc{S_C}\gamma $ holds if $ \Enc{S\gamma} = \Enc{S}\gamma $.

	Similarly, let $ S_C = \boxplus_{0 \leq i < 2^r} p_i \bullet \ConfigCQS{\sigma_i}{\phi}{S\Set{\Subst{\Binary{i}}{v}}} $.
	Then we have $ S_C\gamma = \boxplus_{0 \leq i < 2^r} p_i \bullet \ConfigCQS{\sigma_i\gamma}{\phi}{S\Set{\Subst{\Binary{i}}{v}}\gamma} $.
	Moreover, let $ \sigma_i = \Ket{\psi_i} $, $ \rho = \sum_i pi_i \OuterProduct{\psi_i}{\psi_i} $, $ \tilde{q} = q_0, \ldots, q_{r - 1} $, and $ r = \Length{\tilde{q}} \leq n $.
	Then we have $ \Enc{S_C\gamma} = \ConfigOQS{\ResOQS{\EncDist{\tilde{q}}{v}{\Enc{S\gamma}}}{\phi}}{\rho\gamma} = \ConfigOQS{\ResOQS{\EncDist{\tilde{q}}{v}{\Enc{S}\gamma}}{\phi}}{\rho\gamma} = \ConfigOQS{\ResOQS{\EncDist{\tilde{q}}{v}{\Enc{S}}}{\phi}}{\rho}\gamma = \Enc{S_C}\gamma $ holds if $ \Enc{S\gamma} = \Enc{S}\gamma $.

	We show $ \Enc{S\gamma} = \Enc{S}\gamma $ by induction on the structure of $ S $.
	\begin{compactdesc}
		\item[Case $ S = \nilCQS $] In this case $ S\gamma = S $ and, thus, $ \Enc{S\gamma} = \Enc{S} = \nilOQS = \Enc{S}\gamma $.
		\item[Case $ S = P \mid Q $] In this case $ S\gamma = P\gamma \mid Q\gamma $. By the induction hypothesis, $ \Enc{P\gamma} = \Enc{P}\gamma $ and $ \Enc{Q\gamma} = \Enc{Q}\gamma $. Then $ \Enc{S\gamma} = \Enc{P\gamma} \parOQS \Enc{Q\gamma} = \Enc{P}\gamma \parOQS \Enc{Q}\gamma = {\left( \Enc{P} \parOQS \Enc{Q} \right)}\gamma = \Enc{S}\gamma $.
		\item[Case $ S = \InpCQS{c}{x}{P} $] In this case $ S\gamma = \InpCQS{c}{y}{{\left( P'\gamma \right)}} $, where $ y $ is fresh and $ P' = P\Set{\Subst{y}{x}} $, \ie we use alpha conversion to ensure that the variable that stores the received qubit is fresh in $ \gamma $.
		By the induction hypothesis, $ \Enc{P'\gamma} = \Enc{P'}\gamma $ and $ \Enc{P'} = \Enc{P}\Set{\Subst{y}{x}} $.
		Then we have $ \Enc{S\gamma} = \InpOQS{c}{y}{\Enc{P'\gamma}} = \InpOQS{c}{y}{{\left( \Enc{P'}\gamma \right)}} = {\left( \InpOQS{c}{y}{\Enc{P'}} \right)}\gamma = {\left( \InpOQS{c}{x}{\Enc{P}} \right)}\gamma = \Enc{S}\gamma $.
		\item[Case $ S = \OutCQS{c}{q}{P} $] In this case $ S\gamma = \OutCQS{c}{q\gamma}{{\left( P\gamma \right)}} $.
		By the induction hypothesis, $ \Enc{P\gamma} = \Enc{P}\gamma $.
		Then we have $ \Enc{S\gamma} = \OutOQS{c}{{\left( q\gamma \right)}}{\Enc{P\gamma}} = \OutOQS{c}{{\left( q\gamma \right)}}{{\left( \Enc{P}\gamma \right)}} = {\left( \OutOQS{c}{q}{\Enc{P}} \right)}\gamma = \Enc{S}\gamma $.
		\item[Case $ S = \UnitaryCQS{\tilde{q}}{U}{P} $] In this case $ S\gamma = \UnitaryCQS{\tilde{q}\gamma}{U}{{\left( P\gamma \right)}} $.
		By the induction hypothesis, $ \Enc{P\gamma} = \Enc{P}\gamma $.
		Then $ \Enc{S\gamma} = U\List{\tilde{q}\gamma}{\Enc{P\gamma}} = U\List{\tilde{q}\gamma}{{\left( \Enc{P}\gamma \right)}} = {\left( U\List{\tilde{q}}{\Enc{P}} \right)}\gamma = \Enc{S}\gamma $.
		\item[Case $ S = \MeasCQS{v}{\tilde{q}}{P} $] In this case $ S\gamma = \MeasCQS{v}{\tilde{q}\gamma}{{\left( P\gamma \right)}} $.
		By the induction hypothesis, $ \Enc{P\gamma} = \Enc{P}\gamma $.
		Then we have $ \Enc{S\gamma} = \opM\List{\tilde{q}\gamma}.\EncDist{\tilde{q}\gamma}{v}{\Enc{P\gamma}} = \opM\List{\tilde{q}\gamma}.\EncDist{\tilde{q}\gamma}{v}{\Enc{P}\gamma} = {\left( \opM\List{\tilde{q}}.\EncDist{\tilde{q}}{v}{\Enc{P}} \right)}\gamma = \Enc{S}\gamma $.
		\item[Case $ S = \NewCQS{x}{P} $] In this case $ S\gamma = \NewCQS{x}{{\left( P\gamma \right)}} $. By the induction hypothesis, $ \Enc{P\gamma} = \Enc{P}\gamma $. Then $ \Enc{S\gamma} = \tau.{\left( \ResOQS{\Enc{P\gamma}}{\Set{x}} \right)} = \tau.{\left( \ResOQS{{\left( \Enc{P}\gamma \right)}}{\Set{x}} \right)} = {\left( \tau.{\left( \ResOQS{\Enc{P}}{\Set{x}} \right)} \right)}\gamma = \Enc{S}\gamma $.
		\item[Case $ S = \QubitCQS{x}{P} $] In this case $ S\gamma = \QubitCQS{y}{{\left( P'\gamma \right)}} $, where $ y $ is fresh and $ P' = P\Set{\Subst{y}{x}} $.
		By the induction hypothesis, $ \Enc{P'\gamma} = \Enc{P'}\gamma $ and in particular $ \Enc{P'} = \Enc{P}\Set{\Subst{y}{x}} $.
		Therefore, we have $ \Enc{S\gamma} = \opE_{\Ket{0}}\List{\var}.{\left( \Enc{P'\gamma}\Set{\Subst{q_{\Length{\var}}}{y}} \right)} = \opE_{\Ket{0}}\List{\var}.{\left( \Enc{P'}\gamma\Set{\Subst{q_{\Length{\var}}}{y}} \right)} =$\linebreak${\left( \opE_{\Ket{0}}\List{\var}.{\left( \Enc{P'}\Set{\Subst{q_{\Length{\var}}}{y}} \right)} \right)}\gamma = {\left( \opE_{\Ket{0}}\List{\var}.{\left( \Enc{P}\Set{\Subst{q_{\Length{\var}}}{x}} \right)} \right)}\gamma = \Enc{S}\gamma $.
		\item[Case $ S = \CondCQS{bv}{bv'}{P} $] In this case $ S\gamma = \CondCQS{bv}{bv'}{P\gamma} $. By the induction hypothesis, $ \Enc{P\gamma} = \Enc{P}\gamma $. Then we have $ \Enc{S\gamma} = \left( \CondOQS{bv}{bv'}{\tau.\Enc{P\gamma}} \right) = \left( \CondOQS{bv}{bv'}{\tau.{\left( \Enc{P}\gamma \right)}} \right) = {\left( \CondOQS{bv}{bv'}{\tau.\Enc{P}} \right)}\gamma = \Enc{S}\gamma $.
			\qedhere
	\end{compactdesc}
\end{proof}

We also show invariance modulo the instantiation of a variable for binary numbers by a number.
Again the proof is very similar to the proofs of name and qubit invariance.

\begin{lem}
	\label{lem:binaryInvariance}
	\begin{align*}
		\forall S_C \in \configCQS \logdot \forall v, b \logdot \Enc{S_C\Set{\Subst{b}{v}}} = \Enc{S_C}\Set{\Subst{b}{v}} \quad \text{ and } \quad \forall S \in \procCQS \logdot \forall v, b \logdot \Enc{S\Set{\Subst{b}{v}}} = \Enc{S}\Set{\Subst{b}{v}}
	\end{align*}
\end{lem}

\begin{proof}
	Let $ S_C = \ConfigCQS{\sigma}{\phi}{S} $.
	Then $ S_C\Set{\Subst{b}{v}} = \ConfigCQS{\sigma}{\phi}{S\Set{\Subst{b}{v}}} $.
	Moreover, let $ \sigma = \Ket{\psi} $ and $ \rho = \OuterProduct{\psi}{\psi} $.
	Then $ \Enc{S_C\Set{\Subst{b}{v}}} = \ConfigOQS{\ResOQS{\Enc{S\Set{\Subst{b}{v}}}}{\phi}}{\rho} = \ConfigOQS{\ResOQS{{\left( \Enc{S}\Set{\Subst{b}{v}} \right)}}{\phi}}{\rho} = \ConfigOQS{\ResOQS{\Enc{S}}{\phi}}{\rho}\Set{\Subst{b}{v}} = \Enc{S_C}\Set{\Subst{b}{v}} $ holds if $ \Enc{S\Set{\Subst{b}{v}}} = \Enc{S}\Set{\Subst{b}{v}} $.

	Let $ S_C = \boxplus_{0 \leq i < 2^r} p_i \bullet \ConfigCQS{\sigma_i}{\phi}{S\Set{\Subst{\Binary{i}}{v'}}} $.
	Then we have $ S_C\Set{\Subst{b}{v}} = \boxplus_{0 \leq i < 2^r} p_i \bullet \ConfigCQS{\sigma_i}{\phi}{\Tuple{S\Set{\Subst{\Binary{i}}{v}}}\Set{\Subst{b}{v}}} $.
	Moreover, let $ \sigma_i = \Ket{\psi_i} $, $ \rho = \sum_i pi_i \OuterProduct{\psi_i}{\psi_i} $, $ \tilde{q} = q_0, \ldots, q_{r - 1} $, and $ r = \Length{\tilde{q}} \leq n $.
	If $ v' = v $ then $ v \notin \FreeBinV{S_C} $ and thus $ v \notin \FreeBinV{\Enc{S_C}} $.
	Then $ \Enc{S_C\Set{\Subst{b}{v}}} = \Enc{S_C} = \Enc{S_C}\Set{\Subst{b}{v}} $.
	Else if $ v' \neq v $ then we have $ \Enc{S_C\Set{\Subst{b}{v}}} = \ConfigOQS{\ResOQS{\EncDist{\tilde{q}}{v}{\Enc{S\Set{\Subst{b}{v}}}}}{\phi}}{\rho} = \ConfigOQS{\ResOQS{\EncDist{\tilde{q}}{v}{\Enc{S}\Set{\Subst{b}{v}}}}{\phi}}{\rho} = \ConfigOQS{\ResOQS{\EncDist{\tilde{q}}{v}{\Enc{S}}}{\phi}}{\rho}\Set{\Subst{b}{v}} = \Enc{S_C}\Set{\Subst{b}{v}} $ holds if $ \Enc{S\Set{\Subst{b}{v}}} = \Enc{S}\Set{\Subst{b}{v}} $.

	We show $ \Enc{S\Set{\Subst{b}{v}}} = \Enc{S}\Set{\Subst{b}{v}} $ by induction on the structure of $ S $.
	\begin{compactdesc}
		\item[Case $ S = \nilCQS $] In this case $ S\Set{\Subst{b}{v}} = S $ and, thus, $ \Enc{S\Set{\Subst{b}{v}}} = \Enc{S} = \nilOQS = \Enc{S}\Set{\Subst{b}{v}} $.
		\item[Case $ S = P \mid Q $] In this case $ S\Set{\Subst{b}{v}} = P\Set{\Subst{b}{v}} \mid Q\Set{\Subst{b}{v}} $. By the induction hypothesis, $ \Enc{P\Set{\Subst{b}{v}}} = \Enc{P}\Set{\Subst{b}{v}} $ and $ \Enc{Q\Set{\Subst{b}{v}}} = \Enc{Q}\Set{\Subst{b}{v}} $. Then $ \Enc{S\Set{\Subst{b}{v}}} = \Enc{P\Set{\Subst{b}{v}}} \parOQS \Enc{Q\Set{\Subst{b}{v}}} = \Enc{P}\Set{\Subst{b}{v}} \parOQS \Enc{Q}\Set{\Subst{b}{v}} = \Tuple{\Enc{P} \parOQS \Enc{Q}}\Set{\Subst{b}{v}} = \Enc{S}\Set{\Subst{b}{v}} $.
		\item[Case $ S = \InpCQS{c}{x}{P} $] In this case $ S\Set{\Subst{b}{v}} = \InpCQS{c}{x}{\Tuple{P\Set{\Subst{b}{v}}}} $. By the induction hypothesis, $ \Enc{P\Set{\Subst{b}{v}}} = \Enc{P}\Set{\Subst{b}{v}} $. Then $ \Enc{S\Set{\Subst{b}{v}}} = \InpOQS{c}{x}{\Enc{P\Set{\Subst{b}{v}}}} = \InpOQS{c}{x}{\Tuple{P\Set{\Subst{b}{v}}}} = \Tuple{\InpOQS{c}{x}{\Enc{P}}}\Set{\Subst{b}{v}} = \Enc{S}\Set{\Subst{b}{v}} $.
		\item[Case $ S = \OutCQS{c}{q}{P} $] In this case $ S\Set{\Subst{b}{v}} = \OutCQS{c}{q}{{\left( P\Set{\Subst{b}{v}} \right)}} $. By the induction hypothesis, $ \Enc{P\Set{\Subst{b}{v}}} = \Enc{P}\Set{\Subst{b}{v}} $. Then we have $ \Enc{S\Set{\Subst{b}{v}}} = \OutOQS{c}{q}{\Enc{P\Set{\Subst{b}{v}}}} = \OutOQS{c}{q}{{\left( \Enc{P}\Set{\Subst{b}{v}} \right)}} = {\left( \OutOQS{c}{q}{\Enc{P}} \right)}\Set{\Subst{b}{v}} = \Enc{S}\Set{\Subst{b}{v}} $.
		\item[Case $ S = \UnitaryCQS{\tilde{q}}{U}{P} $] In this case $ S\Set{\Subst{b}{v}} = \UnitaryCQS{\tilde{q}}{U}{{\left( P\Set{\Subst{b}{v}} \right)}} $. By the induction hypothesis, $ \Enc{P\Set{\Subst{b}{v}}} = \Enc{P}\Set{\Subst{b}{v}} $. Then $ \Enc{S\Set{\Subst{b}{v}}} = U\List{\tilde{q}}.\Enc{P\Set{\Subst{b}{v}}} = U\List{\tilde{q}}.{\left( \Enc{P}\Set{\Subst{b}{v}} \right)} = {\left( U\List{\tilde{q}}.\Enc{P} \right)}\Set{\Subst{b}{v}} = \Enc{S}\Set{\Subst{b}{v}} $.
		\item[Case $ S = \MeasCQS{v'}{\tilde{q}}{P} $] In this case $ S\Set{\Subst{b}{v}} = \MeasCQS{v''}{\tilde{q}}{{\left( P'\Set{\Subst{b}{v}} \right)}} $, where $ v'' $ is fresh and $ P' = P\Set{\Subst{v''}{v'}} $. By the induction hypothesis, $ \Enc{P'\Set{\Subst{b}{v}}} = \Enc{P'}\Set{\Subst{b}{v}} $. Then we have $ \Enc{S\Set{\Subst{b}{v}}} = \opM\List{\tilde{q}}.\EncDist{\tilde{q}}{v''}{\Enc{P'\Set{\Subst{b}{v}}}} = \opM\List{\tilde{q}}.\EncDist{\tilde{q}}{v''}{\Enc{P'}\Set{\Subst{b}{v}}} = {\left( \opM\List{\tilde{q}}.\EncDist{\tilde{q}}{v'}{\Enc{P}} \right)}\Set{\Subst{b}{v}} = \Enc{S}\Set{\Subst{b}{v}} $.
		\item[Case $ S = \NewCQS{c}{P} $] In this case $ S\Set{\Subst{b}{v}} = \NewCQS{c}{{\left( P\Set{\Subst{b}{v}} \right)}} $. By the induction hypothesis, $ \Enc{P\Set{\Subst{b}{v}}} = \Enc{P}\Set{\Subst{b}{v}} $. Then we have $ \Enc{S\Set{\Subst{b}{v}}} = \tau.{\left( \ResOQS{\Enc{P\Set{\Subst{b}{v}}}}{\Set{c}} \right)} =$\linebreak${\left( \tau.\Tuple{\ResOQS{\Enc{P}}{\Set{c}}} \right)}\Set{\Subst{b}{v}} = \Enc{S}\Set{\Subst{b}{v}} $.
		\item[Case $ S = \QubitCQS{x}{P} $] In this case $ S\Set{\Subst{b}{v}} = \QubitCQS{x}{{\left( P\Set{\Subst{b}{v}} \right)}} $. By the induction hypothesis, $ \Enc{P\Set{\Subst{b}{v}}} = \Enc{P}\Set{\Subst{b}{v}} $. Then $ \Enc{S\Set{\Subst{b}{v}}} = \opE_{\Ket{0}}\List{\var}.{\left( \Enc{P\Set{\Subst{b}{v}}}\Set{\Subst{q_{\Length{\var}}}{x}} \right)} = {\left( \opE_{\Ket{0}}\List{\var}.{\left( \Enc{P}\Set{\Subst{q_{\Length{\var}}}{x}} \right)} \right)}\Set{\Subst{b}{v}} = \Enc{S}\Set{\Subst{b}{v}} $.
		\item[Case $ S = \CondCQS{bv}{bv'}{P} $] In this case $ S\Set{\Subst{b}{v}} = \CondCQS{\Tuple{bv\Set{\Subst{b}{v}}}}{\Tuple{bv'\Set{\Subst{b}{v}}}}{P\Set{\Subst{b}{v}}} $. By the induction hypothesis, $ \Enc{P\Set{\Subst{b}{v}}} = \Enc{P}\Set{\Subst{b}{v}} $.
			Then
			\begin{align*}
				\Enc{S\Set{\Subst{b}{v}}} &= \left( \CondOQS{\Tuple{bv\Set{\Subst{b}{v}}}}{\Tuple{bv'\Set{\Subst{b}{v}}}}{\tau.\Tuple{\Enc{P\Set{\Subst{b}{v}}}}} \right)\\
				&= \left( \CondOQS{\Tuple{bv\Set{\Subst{b}{v}}}}{\Tuple{bv'\Set{\Subst{b}{v}}}}{\tau.{\left( \Enc{P}\Set{\Subst{b}{v}} \right)}} \right)\\
				&= {\left( \CondOQS{bv}{bv'}{\tau.\Enc{P}} \right)}\Set{\Subst{b}{v}} = \Enc{S}\Set{\Subst{b}{v}}
				\qedhere
			\end{align*}
	\end{compactdesc}
\end{proof}

Then we show the completeness and soundness parts of operational correspondence. 
For completeness, we have to show how target terms emulate source term steps.
Above we observed that steps on \ruleRPermCQS are not emulated at all, \ie are emulated by an empty sequence of steps, and captured this observation in Lemma~\ref{lem:permutation}.
Moreover, Example~\ref{exa:correspondenceSimulation} illustrates that in translating measurement under parallel composition completeness holds \wrt correspondence simulation but not bisimulation.
All other kinds of source term steps are emulated more tightly by exactly one target term step.

\begin{lem}[Operational Completeness, $\Enc{\cdot}$]
	\label{lem:completeness}
	\begin{align*}
		\forall S, S' \in \configCQS \logdot S \steps S' \text{ implies } \exists T \in \configOQS \logdot \Enc{S} \steps T \wedge \Enc{S'} \preceq T
	\end{align*}
\end{lem}

\begin{proof}
	We first consider a single step $ S \step S' $ and show that we need in this case at most one step in the sequence $ \Enc{S} \steps T $ such that $ \Enc{S'} \preceq T $.
	Therefore, we perform an induction over the derivation of $S \step S'$ using a case split over the rules in Figure~\ref{fig:semanticsCQS}.
	\begin{compactdesc}
		\item[Case~\ruleRMeasureCQS] In this case $ S = \ConfigCQS{\sigma}{\phi}{\MeasCQS{v}{\tilde{q}}{P}} $, $ \tilde{q} = q_0, \ldots, q_{r-1} $ and $ S' = \boxplus_{0 \leq m < 2^r} p_m \bullet \ConfigCQS{\sigma_m'}{\phi}{P\Set{\Subst{\Binary{m}}{x}}} $, where $ r = \Length{\tilde{q}} \leq n $, $ \sigma = q_0, \ldots q_{n-1} = \Ket{\psi} = \alpha_0\Ket{\psi_0} + \cdots + \alpha_{2^n-1}\Ket{\psi_{2^n-1}} $, and $ \sigma_m' = \Ket{\psi_m'} $.
			The corresponding encodings are given by
			\begin{align*}
				\Enc{S} = \ConfigOQS{\ResOQS{{\left( \opM\List{\tilde{q}}.\EncDist{\tilde{q}}{v}{\Enc{P}} \right)}}{\phi}}{\rho}
				\quad \text{and} \quad
				\Enc{S'} = \ConfigOQS{\ResOQS{\EncDist{\tilde{q}}{v}{\Enc{P}}}{\phi}}{\rho'},
			\end{align*}
			where $ \rho = \OuterProduct{\psi}{\psi} $ and $ \rho' = \sum_m p_m\OuterProduct{\psi_m'}{\psi_m'} $.
			We observe that $ \Enc{S} $ can emulate the step $ S \step S' $ by applying the super-operator $ \opM\List{\tilde{q}} $ using the Rule~\ruleOperOQS, \ie by
			\begin{align*}
				\Enc{S} \step \ConfigOQS{\ResOQS{\EncDist{\tilde{q}}{v}{\Enc{P}}}{\phi}}{\StateTrans{\opM}{\tilde{q}}{\rho}} = T.
			\end{align*}
			Further, $ \sigma_m' = \Ket{\psi_m'} = \dfrac{\alpha_{l_m}}{\sqrt{p_m}}\Ket{\psi_{l_m}} + \cdots + \dfrac{\alpha_{u_m}}{\sqrt{p_m}}\Ket{\psi_{u_m}} $ with $ l_m = 2^{n-r}m $, $ u_m = 2^{n-r}\Tuple{m+1}-1 $, and $ p_m = \Length{\alpha_{l_m}}^2 + \cdots + \Length{\alpha_{u_m}}^2 $.
			Let $ \Projector{m} $ be the base vector for $ \Binary{m} $ in the standard base.
			Since $ \StateTrans{\opM}{\tilde{q}}{\rho} =  \sum_m \Projector{m}\rho\Adjoint{\Projector{m}} $, and $\rho' = \sum_m p_m\OuterProduct{\psi_m'}{\psi_m'}$, then $\rho' = \StateTrans{\opM}{\tilde{q}}{\rho} $.
			Note that $ \Ket{\psi_m'} = \dfrac{\Projector{m}\Ket{\psi}}{\sqrt{\Trace{\Adjoint{\Projector{m}}\Projector{m}\OuterProduct{\psi}{\psi}}}} $.
			Therefore, the measurement using the super-operator $ \opM\List{\tilde{q}} $ applied to $ \rho $ produces the same probability distribution as measuring $ \sigma $ with $\MeasCQS{v}{q_0, \ldots, q_{r-1}}{P} $ (modulo the different representations of the qubits). It follows $ T = \Enc{S'} $, \ie $ \Enc{S'} \preceq T $.
		\item[Case~\ruleRTransCQS] In this case $ S = \ConfigCQS{\sigma}{\phi}{\UnitaryCQS{\tilde{q}}{U}{P}} $ and $ S' = \ConfigCQS{\sigma'}{\phi}{P} $, where $ \tilde{q} = q_0, \ldots, q_{r-1} $, $ r = \Length{\tilde{q}} \leq n $, $ \sigma = q_0, \ldots q_{n-1} = \Ket{\psi} $, $ \sigma' = q_0, \ldots q_{n-1} = \Ket{\psi'}$, and $ \Ket{\psi'} $ is the result of applying $ U $ on the first $ r $ qubits in $ q_0, \ldots q_{n-1} = \Ket{\psi} $.
			The corresponding encodings are given by
			\begin{align*}
				\Enc{S} = \ConfigOQS{\ResOQS{{\left( U\List{\tilde{q}}.\Enc{P} \right)}}{\phi}}{\rho}
				\quad \text{and} \quad
				\Enc{S'} = \ConfigOQS{\ResOQS{\Enc{P}}{\phi}}{\rho'},
			\end{align*}
			where $ \rho = \OuterProduct{\psi}{\psi} $ and $ \rho' = \OuterProduct{\psi'}{\psi'} $.
			We observe that $ \Enc{S} $ can emulate the step $ S \step S' $ by applying the super-operator $ U\List{\tilde{q}} $ using the Rule~\ruleOperOQS, \ie by
			\begin{align*}
				\Enc{S} \step \ConfigOQS{\ResOQS{\Enc{P}}{\phi}}{\StateTrans{U}{\tilde{q}}{\rho}} = T.
			\end{align*}
			Further, $ \sigma' = (U \tensorProd \opI_{\Set{q_r, \ldots, q_{n-1}}}) \Ket{\psi} = \Ket{\psi'} $ and $ \StateTrans{U}{\tilde{q}}{\rho} = (U \tensorProd \opI_{\var-\tilde{q}}) \cdot \rho \cdot \Adjoint{(U \tensorProd \opI_{\var-\tilde{q}})} $.
			Moreover, since $ \rho' = \OuterProduct{\psi'}{\psi'} $ and $ \Set{q_r, \ldots, q_{n-1}} = \var - \tilde{q} $, it follows $ \rho' = \StateTrans{U}{\tilde{q}}{\rho} $ and therefore $ T = \Enc{S'} $, \ie $ \Enc{S'} \preceq T $.
		\item[Case~\ruleRPermCQS] In this case, $ \Enc{S'} \preceq \Enc{S} $, because of Lemma~\ref{lem:permutation}.
			We choose $ T = \Enc{S} $ such that $ \Enc{S} \steps T $ (by doing 0 steps) and $ \Enc{S'} \preceq T $.
		\item[Case~\ruleRCommCQS] In this case $ S = \ConfigCQS{\sigma}{\phi}{\OutCQS{c}{q}{P} \mid \InpCQS{c}{x}{Q}} $ and $ S' = \ConfigCQS{\sigma}{\phi}{P \mid Q\Set{\Subst{q}{x}}} $, where $ \sigma = q_0, \ldots, q_{n-1} = \Ket{\psi} $. The corresponding encodings are given by
			\begin{align*}
				\Enc{S} = \ConfigOQS{\ResOQS{{\left( \OutOQS{c}{q}{\Enc{P}} \parOQS \InpOQS{c}{x}{\Enc{Q}} \right)}}{\phi}}{\rho}
				\quad \text{and} \quad
				\Enc{S'} = \ConfigOQS{\ResOQS{{\left( \Enc{P} \parOQS \Enc{Q\Set{\Subst{q}{x}}} \right)}}{\phi}}{\rho},
			\end{align*}
			where $ \rho = \OuterProduct{\psi}{\psi} $. We observe that $ \Enc{S} $ can emulate the step $ S \step S' $ using the rules \ruleCommOQS, \ruleInputOQS, and \ruleOutputOQS by
			\begin{align*}
				\Enc{S} \step \ConfigOQS{\ResOQS{{\left( \Enc{P} \parOQS \Enc{Q}\Set{\Subst{q}{x}} \right)}}{\phi}}{\rho} = T.
			\end{align*}
			By Lemma~\ref{lem:qubitInvariance}, $ \Enc{Q\Set{\Subst{q}{x}}} = \Enc{Q}\Set{\Subst{q}{x}} $. Then $ T = \Enc{S'} $, \ie $ \Enc{S'} \preceq T $.
		\item[Case~\ruleRNewCQS] In this case $ S = \ConfigCQS{\sigma}{\phi}{\NewCQS{d}{P}} $ and $ S' = \ConfigCQS{\sigma}{\phi, c}{P\Set{\Subst{c}{d}}} $, where $ c $ is fresh and $ \sigma = q_0, \ldots, q_{n-1} = \Ket{\psi} $.
			The corresponding encodings are given by
			\begin{align*}
				\Enc{S} = \ConfigOQS{\ResOQS{{\left( \tau.{\left( \ResOQS{\Enc{P}}{\Set{d}} \right)} \right)}}{\phi}}{\rho}
				\quad \text{and} \quad
				\Enc{S'} = \ConfigOQS{\ResOQS{\Enc{P\Set{\Subst{c}{d}}}}{\phi, c}}{\rho},
			\end{align*}
			where $ \rho = \OuterProduct{\psi}{\psi} $.
			We observe that $ \Enc{S} $ can emulate the step $ S \step S' $ by reducing $ \tau $ using Rule~\ruleTauOQS, \ie by
			\begin{align*}
				\Enc{S} \step \ConfigOQS{\ResOQS{\Enc{P}}{{\left( \phi \cup \Set{d} \right)}}}{\rho} = T.
			\end{align*}
			By Lemma~\ref{lem:nameInvariance}, $ \Enc{P\Set{\Subst{c}{d}}} = \Enc{P}\Set{\Subst{c}{d}} $.
			Since $ c $ is fresh, then $ \ResOQS{\Enc{P\Set{\Subst{c}{d}}}}{{\left( \phi \cup \Set{ c } \right)}} = \ResOQS{\Enc{P}\Set{\Subst{c}{d}}}{{\left( \phi \cup \Set{ c } \right)}} = \ResOQS{\Enc{P}}{{\left( \phi \cup \Set{ d } \right)}} $.
			Then $ T = \Enc{S'} $, \ie $ \Enc{S'} \preceq T $.
		\item[Case~\ruleRQbitCQS] In this case $ S = \ConfigCQS{\sigma}{\phi}{\QubitCQS{x}{P}} $ and $ S' = \ConfigCQS{\sigma'}{\phi}{P\Set{\Subst{q_n}{x}}} $, where $ \sigma = q_0, \ldots, q_{n-1} = \Ket{\psi} $, $ \sigma' = q_0, \ldots, q_{n-1}, q_n = \Ket{\psi'} $, $ \Ket{\psi'} = \Ket{\psi} \tensorProd \Ket{0} $, $ \var = q_0, \ldots, q_{n - 1} $, and $ q $ is fresh.
			The corresponding encodings are given by the following terms
			\begin{align*}
				\Enc{S} = \ConfigOQS{\ResOQS{{\left( \opE_{\Ket{0}}\List{\var}.{\left( \Enc{P}\Set{\Subst{q_n}{x}} \right)} \right)}}{\phi}}{\rho}
				\quad \text{and} \quad
				\Enc{S'} = \ConfigOQS{\ResOQS{\Enc{P\Set{\Subst{q_n}{x}}}}{\phi}}{\rho'},
			\end{align*}
			where $ \rho = \OuterProduct{\psi}{\psi} $ and $ \rho' = \OuterProduct{\psi'}{\psi'} $.
			We observe that $ \Enc{S} $ can emulate the step $ S \step S' $ by applying the super-operator $ \opE_{\Ket{0}}\List{\var} $ using the Rule~\ruleOperOQS, \ie by
			\begin{align*}
				\Enc{S} \step \ConfigOQS{\ResOQS{{\left( \Enc{P}\Set{\Subst{q_n}{x}} \right)}}{\phi}}{\StateTrans{\opE}{\Ket{0},\var}{\rho}} = T.
			\end{align*}
			By Lemma~\ref{lem:qubitInvariance}, $ \Enc{P\Set{\Subst{q_n}{x}}} = \Enc{P}\Set{\Subst{q_n}{x}} $.
			Further, $ \StateTrans{\opE}{\Ket{0}, \var}{\rho} = \rho' $. Then $ T = \Enc{S'} $, \ie $ \Enc{S'} \preceq T $.
		\item[Case~\ruleRParCQS] In this case $ S = \ConfigCQS{\sigma}{\phi}{P \mid Q} $, $ S' = \boxplus_{0 \leq i < 2^r} p_i \bullet \ConfigCQS{\sigma_i'}{\phi'}{P'\Set{\Subst{\Binary{i}}{v}} \mid Q} $, $ S_P = \ConfigCQS{\sigma}{\phi}{P} \in \configCQS $, and $ S_P \step S_P' = \boxplus_{0 \leq i < 2^r} p_i \bullet \ConfigCQS{\sigma_i'}{\phi'}{P'\Set{\Subst{\Binary{i}}{v}}} $, where $ \sigma = q_0, \ldots, q_{n-1} = \Ket{\psi} $, $ \sigma_i' = q_0, \ldots, q_{n - 1} = \Ket{\psi_i'} $, $ \tilde{q} = q_0, \ldots, q_{r - 1} $, $ r = \Length{\tilde{q}} \leq n $, and $ v $ is fresh in $ Q $.
			By the induction hypothesis, there is some $ T_P \in \configOQS $ such that $ \Enc{S_P} \steps T_P $ and $ \Enc{S_P'} \preceq T_P $.
			Then either (1) $ r = 0 $ and $ S_P' = \boxplus_{0 \leq i < 2^0} p_i \bullet \ConfigCQS{\sigma_i'}{\phi'}{P'\Set{\Subst{\Binary{i}}{v}}} = \ConfigCQS{\sigma_0'}{\phi'}{P'} $, because there is just one case in the probability distribution, or (2) $ r > 0 $ and the probability distribution in $ S_P' $ contains more than one case:
			\begin{compactenum}[(1)]
				\item The corresponding encodings are given by
					\[\begin{array}{lclclcl}
						\Enc{S} & = & \ConfigOQS{\ResOQS{{\left( \Enc{P} \parOQS \Enc{Q} \right)}}{\phi}}{\rho}, & \hspace{2em} & \Enc{S'} & = & \ConfigOQS{\ResOQS{{\left( \Enc{P'} \parOQS \Enc{Q} \right)}}{\phi'}}{\rho'},\\
						\Enc{S_P} & = & \ConfigOQS{\ResOQS{\Enc{P}}{\phi}}{\rho}, & \text{ and } & \Enc{S_P'} & = & \ConfigOQS{\ResOQS{\Enc{P'}}{\phi'}}{\rho'},
					\end{array}\]
					where $ \rho = \OuterProduct{\psi}{\psi} $ and $ \rho' = \OuterProduct{\psi_0'}{\psi_0'} $.
					Since $ \Enc{S_P'} \preceq T_P $, then $ T_P = \ConfigOQS{\ResOQS{T_P'}{\phi'}}{\rho'} $ for some $ T_P' $.
					By the Rule~\ruleRedOQS in Figure~\ref{fig:semanticsOQS}, $ \Enc{S_P} \steps T_P $ implies $ \Enc{S_P} \LabelledStep{\tau}\ldots\LabelledStep{\tau} T_P $ and, by \ruleResOQS, then $ \ConfigOQS{\Enc{P}}{\rho} \steps \ConfigOQS{T_P'}{\rho'} $.
					Then $ \Enc{S} $ can emulate the step $ S \step S' $ using the sequence $ \ConfigOQS{\Enc{P}}{\rho} \steps \ConfigOQS{T_P'}{\rho'} $ and the rules~\ruleIntlOQS and \ruleResOQS by
					\begin{align*}
						\Enc{S} \steps \ConfigOQS{\ResOQS{{\left( T_P' \parOQS \Enc{Q} \right)}}{\phi'}}{\rho'} = T.
					\end{align*}
					Since $ \Enc{S_P} \steps T_P $ contains at most one step, so does $ \Enc{S} \steps T $.
					Finally, we show that $ \Enc{S_P'} \preceq T_P $ implies $ \Enc{S'} \preceq T $:
					\begin{compactitem}
						\item Assume $ \Enc{S'} \LabelledStep{\alpha} C_1 $.
							Then either $ \Enc{Q} $ performs a step on its own, $ \Enc{P'} $ does a step on its own, or they perform a communication step together.
							In the second and third case, $ \Enc{S_P'} \preceq T_P $ ensures that for every $ \Enc{S_P'} = \ConfigOQS{\ResOQS{\Enc{P'}}{\phi'}}{\rho'} \LabelledStep{\alpha'} T_1 $ there is some $ T_P = \ConfigOQS{\ResOQS{T_P'}{\phi'}}{\rho'} \LabelledStep{\alpha'} T_1' $ such that $ T_1 \preceq T_1' $.
							With that, in all three cases, $ T \LabelledStep{\alpha} C_1' $ such that $ C_1 \preceq C_1' $.
						\item Assume $ T \LabelledStep{\alpha} C_1' $.
							Then either $ \Enc{Q} $ performs a step on its own, $ T_P' $ does a step on its own, or they perform a communication step together.
							In the second and third case, $ \Enc{S_P'} \preceq T_P $ ensures that for every $ T_P = \ConfigOQS{\ResOQS{T_P'}{\phi'}}{\rho'} \LabelledStep{\alpha'} T_1' $ there are some $ \Enc{S_P'} = \ConfigOQS{\ResOQS{\Enc{P'}}{\phi'}}{\rho'} \steps\LabelledStep{\alpha'} T_2 $ and $ T_1' \steps T_2' $ such that $ T_2 \preceq T_2' $.
							With that, in all three cases, $ \Enc{S'} \steps\LabelledStep{\alpha} C_2 $ and $ C_1' \steps C_2' $ such that $ C_2 \preceq C_2' $.
						\item Since correspondence simulation is stricter than weak trace equivalence, $ \Enc{S'} $ and $ T $ have the same weak traces and thus $ \ReachBarb{\Enc{S'}}{\success} $ iff $ \ReachBarb{T}{\success} $.
					\end{compactitem}
				\item Since $ v $ is fresh in $ Q $, the corresponding encodings are given by
					\[\begin{array}{lclclcl}
						\hspace{4,5em}\Enc{S} & = & \ConfigOQS{\ResOQS{{\left( \Enc{P} \parOQS \Enc{Q} \right)}}{\phi}}{\rho}, && \Enc{S'} & = & \ConfigOQS{\ResOQS{{\left( \EncDist{\tilde{q}}{v}{\Enc{P'} \parOQS \Enc{Q}} \right)}}{\phi'}}{\rho'},\\
						\hspace{4,5em}\Enc{S_P} & = & \ConfigOQS{\ResOQS{\Enc{P}}{\phi}}{\rho}, && \Enc{S_P'} & = & \ConfigOQS{\ResOQS{\EncDist{\tilde{q}}{v}{\Enc{P'}}}{\phi'}}{\rho'},
					\end{array}\]
					where $ \rho = \OuterProduct{\psi}{\psi} $ and $ \rho' = \sum_i p_i \OuterProduct{\psi_i'}{\psi_i'} $.
					Since $ \Enc{S_P'} \preceq T_P $, then $ T_P =$\linebreak$\ConfigOQS{\ResOQS{\EncDist{\tilde{q}}{v}{T_P'}}{\phi'}}{\rho'} $ for some $ T_P' $.
					By the Rule~\ruleRedOQS in Figure~\ref{fig:semanticsOQS}, $ \Enc{S_P} \steps T_P $ implies $ \Enc{S_P} \LabelledStep{\tau}\ldots\LabelledStep{\tau} T_P $ and, by Rule~\ruleResOQS, then $ \ConfigOQS{\Enc{P}}{\rho} \steps \ConfigOQS{\EncDist{\tilde{q}}{v}{T_P'}}{\rho'} $.
					Then $\Enc{S}$ can emulate the step $S \step S'$ using the sequence $ \ConfigOQS{\Enc{P}}{\rho} \steps \ConfigOQS{\EncDist{\tilde{q}}{v}{T_P'}}{\rho'} $ and the rules~\ruleIntlOQS and \ruleResOQS by
					\begin{align*}
						\Enc{S} \steps \ConfigOQS{\ResOQS{{\left( \EncDist{\tilde{q}}{v}{T_P'} \parOQS \Enc{Q} \right)}}{\phi'}}{\rho'} = T.
					\end{align*}
					Since $ \Enc{S_P} \steps T_P $ contains at most one step, so does $ \Enc{S} \steps T $.
					Finally, we show that $ \Enc{S_P'} \preceq T_P $ implies $ \Enc{S'} \preceq T $:
					\begin{compactitem}
						\item Assume $ \Enc{S'} \LabelledStep{\alpha} C_1 $.
							Then this step reduces the choice to one branch with non-zero probability with \ruleCondOQS and \ruleChoiceOQS and in this branch the respective super-operator to adjust the density matrix to the chosen result of measurement with \ruleOperOQS, \ie $ \alpha = \tau $ and $ C_1 = \ConfigOQSlb{\ResOQS{{\left( \Enc{P'}\Set{\Subst{\Binary{j}}{v}} \parOQS \Enc{Q} \right)}}{\phi'}}{{\opE_{\Binary{j}, \tilde{q}}}{\left( \rho' \right)}} $ for some $ 0 \leq j < 2^r $ with $ p_j \neq 0 $, where $ \opE_{\Binary{j}, \tilde{q}} $ is measurement with the expected result $ \Binary{j} $ and will adapt the state of the measured qubits to $ \Binary{j} $.
							Then $ T \LabelledStep{\alpha} C_1' = \ConfigOQS{\ResOQS{{\left( T_P'\Set{\Subst{\Binary{j}}{v}} \parOQS \Enc{Q} \right)}}{\phi'}}{{\opE_{\Binary{j}, \tilde{q}}}{\left( \rho' \right)}} $.
							Because of $ \ConfigOQS{\ResOQS{\EncDist{\tilde{q}}{v}{\Enc{P'}}}{\phi'}}{\rho'} = \Enc{S_P'} \preceq T_P = \ConfigOQS{\ResOQS{\EncDist{\tilde{q}}{v}{T_P'}}{\phi'}}{\rho'} $ and Lemma~\ref{lem:binaryInvariance}, then $ C_1 \preceq C_1' $.
						\item Assume $ T \LabelledStep{\alpha} C_1' $.
							Then either the choice on the left is reduced or $ \Enc{Q} $\linebreak performs a step on its own.
							In the former case, $ \alpha = \tau $, $ C_1' =$\linebreak$\ConfigOQS{\ResOQS{{\left( T_P'\Set{\Subst{\Binary{j}}{v}} \parOQS \Enc{Q} \right)}}{\phi'}}{{\opE_{\Binary{j}, \tilde{q}}}\Tuple{\rho'}} $, $ 0 \leq j < 2^r $, and $ p_j \neq 0 $.
							Then $ \Enc{S'} \LabelledStep{\alpha} C_1 = \ConfigOQS{\ResOQS{{\left( \Enc{P'}\Set{\Subst{\Binary{j}}{v}} \parOQS \Enc{Q} \right)}}{\phi'}}{\opE_{j, \tilde{q}}\Tuple{\rho'}} $.
							Because of Lemma~\ref{lem:binaryInvariance} and\linebreak $ \ConfigOQS{\ResOQS{\EncDist{\tilde{q}}{v}{\Enc{P'}}}{\phi'}}{\rho'} = \Enc{S_P'} \preceq T_P = \ConfigOQS{\ResOQS{\EncDist{\tilde{q}}{v}{T_P'}}{\phi'}}{\rho'} $, then we pick $ C_2' = C_1' $ and $ C_2 = C_1 $ such that $ C_2 \preceq C_2' $.\\
							In the latter case, $ C_1' = \ConfigOQS{\ResOQS{{\left( \EncDist{\tilde{q}}{v}{T_P'} \parOQS T_Q \right)}}{\phi''}}{\rho''} $.
							Then we pick an arbitrary case $ 0 \leq j < 2^r $ of the probability distribution with non-zero probability $ p_j \neq 0 $ such that $ \Enc{S'} \step \ConfigOQS{\ResOQS{{\left( \Enc{P'}\Set{\Subst{\Binary{j}}{v}} \parOQS \Enc{Q} \right)}}{\phi'}}{\opE_{\Binary{j}, \tilde{q}}\Tilde{\rho'}} \LabelledStep{\alpha} C_2 $ with $ C_2 = \ConfigOQS{\ResOQS{{\left( \Enc{P'}\Set{\Subst{\Binary{j}}{v}} \parOQS T_Q \right)}}{\phi''}}{\rho'''} $, where $ \rho''' $ is the result of applying the transformation on the matrix in the step $ T \LabelledStep{\alpha} C_1' $ (if there is any) to the density matrix $ \opE_{\Binary{j}, \tilde{q}}\Tuple{\rho'} $.
							Because of the non-cloning principle, applying the super-operator $ \opE_{\Binary{j}, \tilde{q}} $ on $ \rho'' $ again yields $ \rho''' $, because $ \opE_{\Binary{j}, \tilde{q}} $ and the super-operator (if any) applied in $ T \LabelledStep{\alpha} C_1' $ need to operate on different sets of qubits.
							Hence, $ C_1' \step C_2' = \ConfigOQS{\ResOQS{{\left( T_P'\Set{\Subst{\Binary{j}}{v}} \parOQS T_Q \right)}}{\phi''}}{\rho'''} $.
							Because of $ \ConfigOQS{\ResOQS{\EncDist{\tilde{q}}{v}{\Enc{P'}}}{\phi'}}{\rho'} = \Enc{S_P'} \preceq T_P = \ConfigOQS{\ResOQS{\EncDist{\tilde{q}}{v}{T_P'}}{\phi'}}{\rho'} $ and Lemma~\ref{lem:binaryInvariance}, then $ C_2 \preceq C_2' $.
						\item Since correspondence simulation is stricter than weak trace equivalence, $ \Enc{S'} $ and $ T $ have the same weak traces and thus $ \HasBarb{\Enc{S'}}{\success} $ iff $ \HasBarb{T}{\success} $.
					\end{compactitem}
			\end{compactenum}
		\item[Case~\ruleRCongCQS] In this case $ S = \ConfigCQS{\sigma}{\phi}{Q} $, $ S' = \boxplus_{0 \leq i < 2^r} p_i \bullet \ConfigCQS{\sigma_i'}{\phi'}{Q'\Set{\Subst{\Binary{i}}{v}}} $, $ Q \equiv P $, $ P' \equiv Q' $, and $ S_P = \ConfigCQS{\sigma}{\phi}{P} \step S_P' = \boxplus_{0 \leq i < 2^r} p_i \bullet \ConfigCQS{\sigma_i'}{\phi'}{P'\Set{\Subst{\Binary{i}}{v}}} $, where $ \sigma = q_0, \ldots, q_{n-1} = \Ket{\psi} $ and $ \sigma_i' = q_0, \ldots, q_{n - 1} = \Ket{\psi'} $.
			By Lemma~\ref{lem:preservationSC}, $ Q \equiv P $ implies $ \Enc{S} \equiv \Enc{S_P} $ and $ P' \equiv Q' $ implies $ \Enc{S_P'} \equiv \Enc{S'} $.
			By the induction hypothesis, there is some $ T_P \in \configOQS $ such that $ \Enc{S_P} \steps T_P $ is a sequence of at most one step and $ \Enc{S_P'} \preceq T_P $.
			Because of $ \Enc{S} \equiv \Enc{S_P} $, \ie $ \Enc{S_P} \preceq \Enc{S} $, then there is some $ T \in \configOQS $ such that $ \Enc{S} \steps T $ is a sequence of at most one step and $ T_P \preceq T $.
			Because of $ \Enc{S_P'} \equiv \Enc{S'} $, \ie $ \Enc{S'} \preceq \Enc{S_P'} $, then $ \Enc{S'} \preceq \Enc{S_P'} \preceq T_P \preceq T $, \ie $ \Enc{S'} \preceq T $.
		\item[Case~\ruleRProbCQS] Then $ S = \boxplus_{0 \leq i < 2^r} p_i \bullet \ConfigCQS{\sigma_i}{\phi}{P\Set{\Subst{\Binary{i}}{v}}} $ and $ S' = \ConfigCQS{\sigma_j}{\phi}{P\Set{\Subst{\Binary{j}}{v}}} $ for some $ 0 \leq j < 2^r $ with $ p_j \neq 0 $, where $ \sigma_i = q_0, \ldots, q_{n-1} = \Ket{\psi_i} $, $ \tilde{q} = q_0, \ldots, q_{r - 1} $, and $ r = \Length{\tilde{q}} \leq n $.
			The corresponding encodings are given by
			\begin{align*}
				\Enc{S} = \ConfigOQS{\ResOQS{\EncDist{\tilde{q}}{v}{\Enc{P}}}{\phi}}{\rho}
				\quad \text{ and } \quad
				\Enc{S'} = \ConfigOQS{\ResOQS{\Enc{P\Set{\Subst{\Binary{j}}{v}}}}{\phi}}{\rho'},
			\end{align*}
			where $ \rho = \sum_i p_i\OuterProduct{\psi_i}{\psi_i} $ and $ \rho' = \OuterProduct{\psi_j}{\psi_j} $.
			We observe that $ \Enc{S} $ can emulate the step $ S \step S' $ using the rules~\ruleChoiceOQS, \ruleCondOQS, and \ruleOperOQS by
			\begin{align*}
				\Enc{S} \step \ConfigOQS{\ResOQS{\Enc{P}\Set{\Subst{\Binary{j}}{v}}}{\phi}}{\opE_{\Binary{j}, \tilde{q}}\Tuple{\rho}} = T,
			\end{align*}
			where $ \opE_{\Binary{j}, \tilde{q}} $ is measurement with the expected result $ \Binary{j} $ and will adapt the state of the measured qubits to $ \Binary{j} $.
			By Lemma~\ref{lem:binaryInvariance}, $ \Enc{P\Set{\Subst{\Binary{j}}{v}}} = \Enc{P}\Set{\Subst{\Binary{j}}{v}} $.
			Since we restrict in \CQS our attention to a probability distributions that results from the measurement of qubits, $ \sigma_i = \dfrac{\alpha_{l_i}}{\sqrt{p_i}}\Ket{\psi_{l_i}} + \cdots + \dfrac{\alpha_{u_i}}{\sqrt{p_i}}\Ket{\psi_{u_i}} $ with $ l_i = 2^{n-r}i$, $u_i = 2^{n-r}\Tuple{i+1}-1 $, and $ p_i = \Length{\alpha_{l_i}}^2 + \cdots + \Length{\alpha_{u_i}}^2 $.
			Accordingly, $ \opE_j\List{\tilde{q}} $ sets the system state to $ \dfrac{\StateTrans{\opE}{j,\tilde{q}}{\rho}}{\Trace{\StateTrans{\opE}{j,\tilde{q}}{\rho}}} = \rho' $.
			Then $ T = \Enc{S'} $, \ie $ \Enc{S'} \preceq T $.
		\item[Case~\ruleRCondCQS] Then $ S = \ConfigCQS{\sigma}{\phi}{\CondCQS{b}{b'}{P}} $, $ b = b' $, and $ S' = \ConfigCQS{\sigma}{\phi}{P} $, where $ \sigma_i = q_0, \ldots, q_{n-1} = \Ket{\psi_i} $.
			The corresponding encodings are given by
			\begin{align*}
				\Enc{S} = \ConfigOQS{\ResOQS{\CondOQS{b}{b'}{\tau.\Enc{P}}}{\phi}}{\rho}
				\quad \text{ and } \quad
				\Enc{S'} = \ConfigOQS{\ResOQS{\Enc{P}}{\phi}}{\rho},
			\end{align*}
			where $ \rho = \OuterProduct{\psi}{\psi} $.
			We observe that $ \Enc{S} $ can emulate the step $ S \step S' $ using the rules \ruleCondOQS and \ruleTauOQS by
			\begin{align*}
				\Enc{S} \step \ConfigOQS{\ResOQS{\Enc{P}}{\phi}}{\rho} = T.
			\end{align*}
			Since $ T = \Enc{S'} $, then $ \Enc{S'} \preceq T $.
	\end{compactdesc}
	Finally, the lemma follows from an induction over the number of steps in $ S \steps S' $.
\end{proof}

In the opposite direction, \ie for soundness, we show that every target term step is the result of emulating a source term step.
Thereby, the formulation of soundness allows to perform---after some initial steps $ \Enc{S} \steps T $ that need to be mapped to the source---some additional steps $ T \steps T' $, to catch up with a source term encoding $ \Enc{S'} $.
To avoid the problem described in Example~\ref{exa:correspondenceSimulation}, we use these additional steps on the target to resolve all unguarded choices as they result from translating probability distributions.
Accordingly, the sequence $ S \steps S' $ contains the mapping of the steps in $ \Enc{S} \steps T $, steps to resolve probability distributions to map the steps in $ T \steps T' $, and some additional steps on Rule~\ruleRPermCQS to permute qubits.
The last kind of steps is necessary in the source to prepare for applications of unitary transformations and measurement, \ie these steps surround in $ S \steps S' $ the corresponding mappings of steps in $ \Enc{S} \steps T $ that apply the super-operators for unitary transformations or measurement.

\begin{lem}[Operational Soundness, $ \Enc{\cdot} $]
	\label{lem:soundness}
	\begin{align*}
		& \forall S \in \configCQS \logdot \forall T \in \configOQS \logdot \Enc{S} \steps T \text{ implies }\\
		& \exists S' \in \configCQS \logdot \exists T' \in \configOQS \logdot S \steps S' \wedge T \steps T' \wedge \Enc{S'} \preceq T'
	\end{align*}
\end{lem}

\begin{proof}
	We strengthen the proof goal by replacing $ \preceq $ with equality:
	\begin{align*}
		\forall S \in \configCQS \logdot \forall T \in \configOQS \logdot \Enc{S} \steps T \text{ implies } \exists S' \in \configCQS \logdot S \steps S' \wedge T \steps \Enc{S'}
	\end{align*}
	Moreover, we require that either $ S' = S $ or $ S' $ is not a probability distribution with $ r > 0 $ and that every step in the sequence $ T \steps \Enc{S'} $ reduces a choice.
	Then the proof is by induction on the number of steps in $ \Enc{S} \steps T $.
	The base case for zero steps, \ie $ T = \Enc{S} $, holds trivially by choosing $ S' = S $.
	For the induction step, assume $ \Enc{S} \steps T^* \step T $.
	By the induction hypothesis, there is some $ S^{**} $ such that $ S \steps S^{**} $ and $ T^* \steps \Enc{S^{**}} $, where $ S^{**} $ is not a probability distribution and in $ T^* \steps \Enc{S^{**}} $ only choices are reduced.
	Let $ S^{**} = \ConfigCQS{\sigma^{**}}{\phi^{**}}{P^{**}} $ with $ \sigma^{**} = q_0, \ldots, q_{n^{**} - 1} = \Ket{\psi^{**}} $.
	By Definition~\ref{def:encoding}, then $ \Enc{S^{**}} = \ConfigOQS{\ResOQS{\Enc{P^{**}}}{\phi^{**}}}{\rho^{**}} $ with $ \rho^{**} = \OuterProduct{\psi^{**}}{\psi^{**}} $.

	By Figure~\ref{fig:semanticsOQS}, $ T^* \step T $ was derived from the Rule~\ruleRedOQS, \ie $ T^* \LabelledStep{\tau} T $, and the derivation of $ T^* \LabelledStep{\tau} T $ is based on either (1)~the Axiom~\ruleTauOQS, (2)~the Axiom~\ruleOperOQS, or (3)~both of the Axioms~\ruleInputOQS and \ruleOutputOQS.
	\begin{compactenum}[(1)]
		\item By Definition~\ref{def:encoding}, $ \tau $ cannot guard a branch of a choice.
			Then $ \tau $ (a)~does not guard the subterm of a conditional, or (b)~guards the subterm of a conditional without a measurement, or (c)~guards the subterm of a conditional with a measurement.
			\begin{compactenum}[(a)]
				\item Then $ T^* $ contains an unguarded subterm $ \tau.{\left( \ResOQS{T_{\tau}}{c} \right)} $ that is reduced in the step $ T^* \step T $.
					Because of Definition~\ref{def:encoding} and since $ T^* \steps \Enc{S^{**}} $ reduces only choices, then $ S^{**} $ contains an unguarded subterm $ \NewCQS{c}{P_{\mathsf{new}}} $ that was translated into $ \tau.{\left( \ResOQS{T_{\tau}}{c} \right)} $.
					Then there is some $ S' $ such that $ S^{**} \step S' = \ConfigCQS{\sigma^{**}}{\phi^{**}, c}{P_{\mathsf{new}}'\Set{\Subst{c}{d}}} $, where $ c, d $ are fresh and $ P_{\mathsf{new}}' $ is obtained from $ P^{**} $ by replacing $ \NewCQS{c}{P_{\mathsf{new}}} $ with $ P_{\mathsf{new}}\Set{\Subst{d}{c}} $.
					Then $ S \steps S' $.
					By Lemma~\ref{lem:nameInvariance}, then $ \Enc{S'} = \ConfigOQS{\ResOQS{\Enc{P_{\mathsf{new}}'\Set{\Subst{c}{d}}}}{{\left( \phi^{**}, c \right)}}}{\rho^{**}} = \ConfigOQS{\ResOQS{\Enc{P_{\mathsf{new}}'}}{{\left( \phi^{**}, d \right)}}}{\rho^{**}} $.
					Since $ T^* \step T $ is not in conflict with any of the steps of $ T^* \steps \Enc{S^{**}} $, $ T \steps \Enc{S'} $ performs the sequence $ T^* \steps \Enc{S^{**}} $ starting in $ T $ instead of $ T^* $.
				\item Then $ T^* $ contains an unguarded subterm $ \CondOQS{bv}{bv'}{\tau.T_{\tau}} $ that is reduced in the step $ T^* \step T $.
					Because of Definition~\ref{def:encoding} and since $ T^* \steps \Enc{S^{**}} $ reduces only choices, then $ S^{**} $ contains an unguarded subterm $ \CondCQS{bv}{bv'}{P_{\mathsf{cond}}} $ that was translated into $ \CondOQS{bv}{bv'}{\tau.T_{\tau}} $.
					Then there is some $ S' $ such that $ S^{**} \step S' = \ConfigCQS{\sigma^{**}}{\phi^{**}}{P_{\mathsf{cond}}'} $, where $ P_{\mathsf{cond}}' $ is obtained from $ P^{**} $ by replacing $ \CondCQS{bv}{bv'}{P_{\mathsf{cond}}} $ with $ P_{\mathsf{cond}} $.
					Then $ S \steps S' $.
					Since $ T^* \step T $ is not in conflict with any of the steps of $ T^* \steps \Enc{S^{**}} $, $ T \steps \Enc{S'} $ performs the sequence $ T^* \steps \Enc{S^{**}} $ starting in $ T $ instead of $ T^* $.
				\item By Definition~\ref{def:encoding}, then the $ \tau $ guards the subterm of a conditional within a choice.
					Since $ T^* \steps \Enc{S^{**}} $ and $ S^{**} $ is not a probability distribution (with $ r > 0 $), then $ T^* \steps \Enc{S^{**}} $ reduces this choice but not necessarily to the case $ j $ that contains the considered $ \tau $ guard.
					Accordingly, $ S \steps S^{**} $ contains a step that reduces the corresponding probability distribution, where the respective branch is not further reduced because $ T^* \steps \Enc{S^{**}} $ reduces only choices.
					Then we replace in $ S \steps S^{**} $ and $ T^* \steps \Enc{S^{**}} $ the respective steps reducing the probability distribution and the choice in question by a step that reduces this probability distribution and this choice to case $ j $.
					Note that, because $ T^* \step T $ measures $ \tilde{q} $, case $ j $ has a non-zero probability.
					Finally, we reorder the steps on the target such that $ S \steps S' $ and $ T^* \step T \steps \Enc{S'} $, where $ S' $ is obtained from $ S^{**} $ by adapting the chosen branch to case $ j $.
					Note that this is the only case, in that the state of $ S' $ is not $ \sigma^{**} $, because the adaptation of the branch to case $ j $ also requires to adapt the state accordingly.
			\end{compactenum}
		\item By Definition~\ref{def:encoding}, one of the following super-operators was reduced:
			\begin{compactdesc}
				\item[Case of $ U\List{\tilde{q}} $] By Definition~\ref{def:encoding}, $ U\List{\tilde{q}} $ cannot guard a branch of a choice nor can $ U\List{\tilde{q}} $ guard the subterm of a conditional.
					Then $ T^* $ contains an unguarded subterm $ U\List{\tilde{q}}.T_{\mathsf{U}} $ that is reduced in the step $ T^* \step T $.
					Because of Definition~\ref{def:encoding} and since $ T^* \steps \Enc{S^{**}} $ reduces only choices, then $ S^{**} $ contains an unguarded subterm $ \UnitaryCQS{\tilde{q}}{U}{P_{\mathsf{U}}} $ that was translated into $ U\List{\tilde{q}}.T_{\mathsf{U}} $.
					Then there are some $ S_{\mathsf{perm}}, S_{\mathsf{U}}, S' $ such that $ S^{**} \step S_{\mathsf{perm}} \step S_{\mathsf{U}} \step S' = \ConfigCQS{\sigma'}{\phi^{**}}{P_{\mathsf{U}}'} $, where $ S^{**} \step S_{\mathsf{perm}} $ is by Rule~\ruleRPermCQS and permutes the qubits in $ \tilde{q} $ to the front using a permutation $ \pi $, $ S_{\mathsf{perm}} \step S_{\mathsf{U}} $ performs the unitary transformation, $ S_{\mathsf{U}} \step S' $ permutes the qubits back to their original order, $ \sigma' = \Pi{\left( \Tuple{U \tensorProd \opI_{\Set{q_{\Length{\tilde{q}}}, \ldots, q_{n-1}}}}{\left( \Pi\Ket{\psi^{**}} \right)} \right)} = \Ket{\psi'} $, and $ P_{\mathsf{U}}' $ is obtained from $ P^{**} $ by replacing $ \UnitaryCQS{\tilde{q}}{U}{P_{\mathsf{U}}} $ with $ P_{\mathsf{U}} $.
					Then $ S \steps S' $ and $ \Enc{S'} = \ConfigOQS{\ResOQS{\Enc{P_{\mathsf{U}}'}}{\phi^{**}}}{\rho'} $, where $ \rho' = \StateTrans{U}{\tilde{q}}{\rho^{**}} = \OuterProduct{\psi'}{\psi'} $.
					Since $ T^* \step T $ is not in conflict with any of the steps of $ T^* \steps \Enc{S^{**}} $, $ T \steps \Enc{S'} $ performs the sequence $ T^* \steps \Enc{S^{**}} $ starting in $ T $ instead of $ T^* $.
				\item[Case of $ \opM\List{\tilde{q}} $] By Definition~\ref{def:encoding}, $ \opM\List{\tilde{q}} $ cannot guard a branch of a choice nor can $ \opM\List{\tilde{q}} $ guard the subterm of a conditional.
					Then $ T^* $ contains an unguarded subterm $ \opM\List{\tilde{q}}.T_{\mathsf{M}} $ that is reduced in the step $ T^* \step T $.
					Because of Definition~\ref{def:encoding} and since $ T^* \steps \Enc{S^{**}} $ reduces only choices, then $ S^{**} $ contains an unguarded subterm $ \MeasCQS{v}{\tilde{q}}{P_{\mathsf{M}}} $ that was translated into $ \opM\List{\tilde{q}}.T_{\mathsf{M}} $.
					Then there are some $ S_{\mathsf{perm}}, S_{\mathsf{M}}, S_{\mathsf{dist}}, S' $ such that $ S^{**} \step S_{\mathsf{perm}} \step S_{\mathsf{M}} \step S_{\mathsf{dist}} \step S' = \ConfigCQS{\sigma'}{\phi^{**}}{P_{\mathsf{M}}'\Set{\Subst{\Binary{j}}{v'}}} $, where $ S^{**} \step S_{\mathsf{perm}} $ is by Rule~\ruleRPermCQS and permutes the qubits in $ \tilde{q} $ to the front using a permutation $ \pi $, $ S_{\mathsf{perm}} \step S_{\mathsf{M}} $ performs the measurement, $ S_{\mathsf{M}} \step S_{\mathsf{dist}} $ resolves the resulting probability distribution to an arbitrary case $ j $ with non-zero probability, $ S_{\mathsf{dist}} \step S' $ permutes the qubits back to their original order, $ v' $ is fresh, $ \sigma' = \Ket{\psi'} $ is the result of measuring the qubits $ \tilde{q} $ in $ \sigma^{**} $, and $ P_{\mathsf{M}}' $ is obtained from $ P^{**} $ by replacing $ \MeasCQS{v}{\tilde{q}}{P_{\mathsf{M}}} $ with $ P_{\mathsf{M}}\Set{\Subst{v'}{v}} $.
					Then $ S \steps S' $.
					By Lemma~\ref{lem:binaryInvariance}, then $ \Enc{S'} = \ConfigOQS{\ResOQS{\Enc{P_{\mathsf{M}}'\Set{\Subst{\Binary{j}}{v'}}}}{\phi^{**}}}{\rho'} = \ConfigOQS{\ResOQS{\Enc{P_{\mathsf{M}}'}\Set{\Subst{\Binary{j}}{v'}}}{\phi^{**}}}{\rho'} $, where $ \StateTrans{\opE}{\Binary{j}, \tilde{q}}{\opM_{\tilde{q}}\Tilde{\rho^{**}}} $ sets the system state to $ \rho' = \OuterProduct{\psi'}{\psi'} $.
					Since $ T^* \step T $ is not in conflict with any of the steps of $ T^* \steps \Enc{S^{**}} $, $ T \steps \Enc{S'} $ performs the sequence $ T^* \steps \Enc{S^{**}} $ starting in $ T $ instead of $ T^* $ and one additional step to reduce the choice that is the outermost operator of $ T_{\mathsf{M}} $ to case $ j $.
				\item[Case of $ \opE_{\Ket{0}}\List{\var} $] By Definition~\ref{def:encoding}, $ \opE_{\Ket{0}}\List{\var} $ cannot guard a branch of a choice nor can $ \opE_{\Ket{0}}\List{\var} $ guard the subterm of a conditional.
					Then $ T^* $ contains an unguarded subterm of the form $ \opE_{\Ket{0}}\List{\var}.{\left( T_{\mathsf{qbit}}\Set{\Subst{q_{\Length{\var}}}{x}} \right)} $ that is reduced in the step $ T^* \step T $.
					Because of Definition~\ref{def:encoding} and since $ T^* \steps \Enc{S^{**}} $ reduces only choices, then $ S^{**} $ contains an unguarded subterm $ \QubitCQS{x}{P_{\mathsf{qbit}}} $ that was translated into $ \opE_{\Ket{0}}\List{\var}.{\left( T_{\mathsf{qbit}}\Set{\Subst{q_{\Length{\var}}}{x}} \right)} $.
					Then there is some $ S' $ such that $ S^{**} \step S' = \ConfigCQS{\sigma'}{\phi^{**}}{P_{\mathsf{qbit}}'\Set{\Subst{q_{\Length{\var}}}{y}}} $, where $ y $ is fresh, $ \sigma' = \Ket{\psi^{**}} \tensorProd \Ket{0} = \Ket{\psi'} $, and $ P_{\mathsf{qbit}}' $ is obtained from $ P^{**} $ by replacing $ \QubitCQS{x}{P_{\mathsf{qbit}}} $ with $ P_{\mathsf{qbit}}\Set{\Subst{y}{x}} $.
					Then $ S \steps S' $.
					By Lemma~\ref{lem:qubitInvariance}, then $ \Enc{S'} = \ConfigOQS{\ResOQS{\Enc{P_{\mathsf{qbit}}'\Set{\Subst{q_{\Length{\var}}}{y}}}}{\phi^{**}}}{\rho'} = \ConfigOQS{\ResOQS{\Enc{P_{\mathsf{qbit}}'}\Set{\Subst{q_{\Length{\var}}}{y}}}{\phi^{**}}}{\rho'} $, where $ \rho' = \StateTrans{\opE}{\Ket{0}, \var}{\rho^{**}} = \OuterProduct{\psi'}{\psi'} $.
					Since $ T^* \step T $ is not in conflict with any of the steps of $ T^* \steps \Enc{S^{**}} $, $ T \steps \Enc{S'} $ performs the sequence $ T^* \steps \Enc{S^{**}} $ starting in $ T $ instead of $ T^* $.
			\end{compactdesc}
		\item By Definition~\ref{def:encoding}, inputs or outputs cannot guard a branch of a choice nor can inputs or outputs guard the subterm of a conditional.
			Then $ T^* $ contains two unguarded subterms $ \InpOQS{c}{x}{T_{\mathsf{in}}} $ and $ \OutOQS{c}{q}{T_{\mathsf{out}}} $ that are reduced in the step $ T^* \step T $.
			Because of Definition~\ref{def:encoding} and since $ T^* \steps \Enc{S^{**}} $ reduces only choices, then $ S^{**} $ contains two unguarded subterms $ \InpCQS{c}{x}{P_{\mathsf{in}}} $ and $ \OutCQS{c}{q}{P_{\mathsf{out}}} $ that were translated into $ \InpOQS{c}{x}{T_{\mathsf{in}}} $ and $ \OutOQS{c}{q}{T_{\mathsf{out}}} $.
			Then there is some $ S' $ such that $ S^{**} \step S' = \ConfigCQS{\sigma^{**}}{\phi^{**}}{P_{\mathsf{com}}} $, where $ P_{\mathsf{com}} $ is obtained from $ P^{**} $ by replacing $ \InpCQS{c}{x}{P_{\mathsf{in}}} $ with $ P_{\mathsf{in}}\Set{\Subst{q}{x}} $ and $ \OutCQS{c}{q}{P_{\mathsf{out}}} $ with $ P_{\mathsf{out}} $.
			Then $ S \steps S' $ and $ \Enc{S'} = \ConfigOQS{\ResOQS{\Enc{P_{\mathsf{com}}}}{\phi^{**}}}{\rho^{**}} $.
			Since $ T^* \step T $ is not in conflict with any of the steps of $ T^* \steps \Enc{S^{**}} $, $ T \steps \Enc{S'} $ performs the sequence $ T^* \steps \Enc{S^{**}} $ starting in $ T $ instead of $ T^* $.
			\qedhere
	\end{compactenum}
\end{proof}

Divergence reflection follows from the above soundness proof.

\begin{lem}[Divergence Reflection, $ \Enc{\cdot} $]
	\label{lem:divergenceReflection}
	\begin{align*}
		\forall S \in \configCQS \logdot \Enc{S} \infiniteSteps \text{ implies } S \infiniteSteps
	\end{align*}
\end{lem}

\begin{proof}
	By the variant of soundness that we show in the proof of Lemma~\ref{lem:soundness}, for every sequence $ \Enc{S} \steps T $ there is some $ S' \in \configCQS $ such that $ S \steps S' $ and $ T \steps \Enc{S'} $, where the sequence $ S \steps S' $ is at least as long as $ T \steps \Enc{S'} $ (and often longer).
	Then for every sequence of target term steps there is a matching sequence of source term steps that is at least as long.
	This ensures divergence reflection.
\end{proof}

Success sensitiveness follows from the homomorphic translation of $ \success $ in Definition~\ref{def:encoding} and operational correspondence.

\begin{lem}[Success Sensitiveness, $ \Enc{\cdot} $]
	\label{lem:successSensitiveness}
	\begin{align*}
		\forall S \in \configCQS \logdot \ReachBarb{S}{\success} \text{ iff } \ReachBarb{\Enc{S}}{\success}
	\end{align*}
\end{lem}

\begin{proof}
	By Definition~\ref{def:encoding}, $ \HasBarb{S^*}{\success} $ iff $ \HasBarb{\Enc{S^*}}{\success} $ for all $ S^* $.
	\begin{compactitem}
		\item If $ \ReachBarb{S}{\success} $, then $ S \steps S' $ and $ \HasBarb{S'}{\success} $.
			By Lemma~\ref{lem:completeness}, then $ \Enc{S} \steps T $ and $ \Enc{S'} \preceq T $.
			Since $ \preceq $ is success sensitive and $ \HasBarb{S'}{\success} $ implies $ \HasBarb{\Enc{S'}}{\success} $, then $ \HasBarb{T}{\success} $ and, thus, $ \ReachBarb{\Enc{S}}{\success} $.
		\item If $ \ReachBarb{\Enc{S}}{\success} $, then $ \Enc{S} \steps T $ and $ \HasBarb{T}{\success} $.
			By the proof of Lemma~\ref{lem:soundness}, then $ S \steps S' $ and $ T \steps \Enc{S'} $.
			Since $ \HasBarb{\Enc{S'}}{\success} $ implies $ \HasBarb{S'}{\success} $, then $ \HasBarb{S'}{\success} $ and, thus, $ \ReachBarb{S}{\success} $.
			\qedhere
	\end{compactitem}
\end{proof}

Compositionality follows directly from the encoding function, \ie as we can observe in Definition~\ref{def:encoding} every source term operator is translated in a compositional way.
With that we can show that the encoding $ \Enc{\cdot} $ satisfies the properties
\begin{enumerate*}[(1)]
	\item compositionality,
	\item name invariance,
	\item operational correspondence,
	\item divergence reflection, and
	\item success sensitiveness.
\end{enumerate*}

\begin{thm}
	\label{thm:goodEncoding}
	The encoding $ \Enc{\cdot} $ is good.
\end{thm}

\begin{proof}
	By Definition~\ref{def:encoding}, $ \Enc{\cdot} $ is compositional, because we can derive the required contexts from the right hand side of the equations by replacing the encodings of the respective sub-terms by holes $ \hole $.

	By Lemma~\ref{lem:nameInvariance}, $ \Enc{\cdot} $ is name invariant.

	By Lemma~\ref{lem:completeness} and Lemma~\ref{lem:soundness}, $ \Enc{\cdot} $ is operationally corresponding with respect to the success sensitive correspondence simulation $ \preceq $.

	By Lemma~\ref{lem:divergenceReflection}, $ \Enc{\cdot} $ reflects divergence.

	By Lemma~\ref{lem:successSensitiveness}, $ \Enc{\cdot} $ is success sensitive.
\end{proof}

By \cite{petersGlabbeek15}, Theorem~\ref{thm:goodEncoding} implies that there is a correspondence simulation that relates source terms $ S $ and their literal translations $ \Enc{S} $.
To refer to a more standard equivalence, this also implies that $ S $ and $ \Enc{S} $ are coupled similar (for the relevance of coupled similarity see \eg \cite{bispingNestmannPeters20}).
Proving operational correspondence \wrt a bisimulation would not significantly tighten the connection between the source and the target.
To really tighten the connection such that $ S $ and $ \Enc{S} $ are bisimilar, we need a stricter variant of operational correspondence and for that a more direct translation of probability distributions to avoid the problem discussed in Example~\ref{exa:correspondenceSimulation}.
Indeed \cite{FengDuanYing12} introduces probability distributions to \qCCS and a corresponding alternative of measurement that allows to translate this operator homomorphically.
However, in this study we are more concerned about the quality criteria.
Hence using them to compare languages that treat qubits fundamentally differently is more interesting here.
Moreover, to tighten the connection we would need a probabilistic version of operational correspondence and accordingly a probabilistic version of bisimulation.
Very recently we introduced probabilistic operational correspondence in \cite{schmittPeters23}.

To illustrate the encoding $ \Enc{\cdot} $ on a practical relevant example, we present the translation of the quantum teleportation protocol in Example~\ref{exa:teleportationSource}.

\begin{exa}
	\label{exa:teleportation}
	By Definition~\ref{def:encoding},
	\begin{align*}
		\Enc{S} ={} & \ConfigOQS{\tau.\Tuple{\Enc{\mathit{Alice}\Tuple{q_0, q_1}} \parOQS \Enc{\mathit{Bob}\Tuple{q_2}}}}{\rho_0}\\
		\Enc{\mathit{Alice}\Tuple{q_0, q_1}} ={} & \opCNot\List{q0, q_1}.\opH\List{q_0}.\opM\List{q_0, q_1}.\EncDist{q_0, q_1}{v_0}{\OutOQS{c}{q_0}{\OutOQS{c}{q_1}{\nilOQS}}}\\
		\Enc{\mathit{Bob}\Tuple{q_2}} ={} & \InpOQS{c}{x_0}{}\InpOQS{c}{x_1}{}\opM\List{x_0, x_1}.\EncDist{x_0, x_1}{v}{T_B}\\
		T_B ={} & \CondOQS{v}{00}{\tau.\success} \parOQS \CondOQS{v}{01}{\tau.\opX\List{q_2}.\success} \parOQS{}\\
		& \CondOQS{v}{10}{\tau.\opZ\List{q_2}.\success}  \parOQS \CondOQS{v}{11}{\tau.\opY\List{q_2}.\success}
	\end{align*}
	where $ \rho_0 = \OuterProduct{\psi_0}{\psi_0} $.
	By Figure~\ref{fig:semanticsOQS}, $ \Enc{S} $ can do the following sequence of steps to emulate the sequence in Example~\ref{exa:teleportationSource}
	\begin{align*}
		\Enc{S} &\step \ConfigOQS{\Enc{\mathit{Alice}\Tuple{q_0, q_1}} \parOQS \Enc{\mathit{Bob}\Tuple{q_2}}}{\rho_0}\\
		&\step \ConfigOQS{\opH\List{q_0}.\opM\List{q_0, q_1}.\EncDist{q_0, q_1}{v_0}{\OutOQS{c}{q_0}{\OutOQS{c}{q_1}{\nilOQS}}} \parOQS \Enc{\mathit{Bob}\Tuple{q_2}}}{\rho_1}\\
		&\step \ConfigOQS{\opM\List{q_0, q_1}.\EncDist{q_0, q_1}{v_0}{\OutOQS{c}{q_0}{\OutOQS{c}{q_1}{\nilOQS}}} \parOQS \Enc{\mathit{Bob}\Tuple{q_2}}}{\rho_2}\\
		&\step \ConfigOQS{\EncDist{q_0, q_1}{v_0}{\OutOQS{c}{q_0}{\OutOQS{c}{q_1}{\nilOQS}}} \parOQS \Enc{\mathit{Bob}\Tuple{q_2}}}{\rho_3} = T^*
	\end{align*}
	where $ \rho_1 = \opCNot\List{q_0, q_1}\Tuple{\rho_0} $, $ \rho_2 = \opH\List{q_0}\Tuple{\rho_1} $, $ \rho_3 = \opM\List{q_0, q_1}\Tuple{\rho_2} $, and the state $ \rho_3 $ corresponds to $ \Ket{\psi_2} = q_0, q_1, q_2 = \frac{1}{2}\Ket{001} + \frac{1}{2}\Ket{010} - \frac{1}{2}\Ket{101} - \frac{1}{2}\Ket{110} $ in Example~\ref{exa:teleportationSource}.
	\begin{align*}
		\EncDist{q_0, q_1}{v_0}{\OutOQS{c}{q_0}{\OutOQS{c}{q_1}{\nilOQS}}} ={}& \Tuple{\CondOQS{00}{\opM\List{q_0, q_1}}{\tau.\OutOQS{c}{q_0}{\OutOQS{c}{q_1}{\nilOQS}}}} +{}\\
		& \Tuple{\CondOQS{01}{\opM\List{q_0, q_1}}{\tau.\OutOQS{c}{q_0}{\OutOQS{c}{q_1}{\nilOQS}}}} +{}\\
		& \Tuple{\CondOQS{10}{\opM\List{q_0, q_1}}{\tau.\OutOQS{c}{q_0}{\OutOQS{c}{q_1}{\nilOQS}}}} +{}\\
		& \Tuple{\CondOQS{11}{\opM\List{q_0, q_1}}{\tau.\OutOQS{c}{q_0}{\OutOQS{c}{q_1}{\nilOQS}}}} +{}
	\end{align*}
	As in Example~\ref{exa:teleportationSource} we choose again the first branch:
	\begin{align*}
		T^{*} \step{}& \ConfigOQS{\OutOQS{c}{q_0}{\OutOQS{c}{q_1}{\nilOQS}} \parOQS \Enc{\mathit{Bob}\Tuple{q_2}}}{\rho_4}\\
		\step{}& \ConfigOQS{\OutOQS{c}{q_1}{\nilOQS} \parOQS \InpOQS{c}{x_1}{}\opM\List{q_0, x_1}.\EncDist{q_0, x_1}{v}{T_B}}{\rho_4}\\
		\step{}& \ConfigOQS{\opM\List{q_0, q_1}.\EncDist{q_0, q_1}{v}{T_B}}{\rho_4}\\
		\step{}& \ConfigOQS{\EncDist{q_0, q_1}{v}{T_B}}{\rho_4} = T^{**}
	\end{align*}
	where $ \rho_4 = \opE_{0, q_0, q_1}\Tuple{\rho_3} = \OuterProduct{001}{001} $.
	Note that the measurement in the last of the above steps has no effect on the state, since $ q_0 $ and $ q_1 $ are already both in the base state $ \Ket{0} $.
	Because of that $ \EncDist{q_0, q_1}{v}{T_B} $ can only reduce to the first state of $ T_B $.
	\begin{align*}
		T^{**} &\step \ConfigOQS{\begin{array}{l}
				\CondOQS{00}{00}{\tau.\success} \parOQS \CondOQS{00}{01}{\tau.\opX\List{q_2}.\success} \parOQS{}\\
				\CondOQS{00}{10}{\tau.\opZ\List{q_2}.\success} \parOQS \CondOQS{00}{11}{\tau.\opY\List{q_2}.\success}
			\end{array}}{\rho_4}\\
			&\step \ConfigOQS{\begin{array}{l}
				\success \parOQS \CondOQS{00}{01}{\tau.\opX\List{q_2}.\success} \parOQS{}\\
				\CondOQS{00}{10}{\tau.\opZ\List{q_2}.\success} \parOQS \CondOQS{00}{11}{\tau.\opY\List{q_2}.\success}
			\end{array}}{\rho_4} \tag*{\qed}
	\end{align*}
\end{exa}


\section{Separating Quantum Based Systems}
\label{sec:separation}

Since super-operators are more expressive than unitary transformations, an encoding from \qCCS or \OQS into \CQP or \CQS is more difficult.

\begin{exa}[Phase Flip Channel]
	\label{exa:phaseFlipChannel}
	Consider the operator $ \opQ\Tuple{\rho} = E_0 \rho E_0^{\dagger} + E_1 \rho E_1^{\dagger} $, where $ E_0 = \sqrt{0.5}\opI = \begin{pmatrix} \sqrt{0.5} & 0\\ 0 & \sqrt{0.5} \end{pmatrix} $ and $ E_1 = \sqrt{0.5}\opZ = \begin{pmatrix} \sqrt{0.5} & 0\\ 0 & -\sqrt{0.5} \end{pmatrix} $, that is presented under the name \emph{phase flip} channel in \cite[Section~8.3.3]{NielsenChuang10} (for $ p = 0.5 $) as an operator to introduce noise.
	Note that $ E_0^{\dagger}E_0 + E_1^{\dagger}E_1 = \opI $.
	By Definition~\ref{def:superopSum}, $ \opQ $ is then a trace-preserving super-operator (in sum representation).
	$ \opQ $ sometimes behaves as identity, in particular we have $ \opQ\Tuple{\OuterProduct{0}{0}} = \OuterProduct{0}{0} $ and $ \opQ\Tuple{\OuterProduct{1}{1}} = \OuterProduct{1}{1} $, and sometimes it changes a qubit, in particular we have $ \opQ\Tuple{\OuterProduct{+}{+}} = \begin{pmatrix} 0.5 & 0\\ 0 & 0.5 \end{pmatrix} = \opQ\Tuple{\OuterProduct{-}{-}} $.
	\qed
\end{exa}

It is easy to show that there is no unitary transformation with the behaviour of $ \opQ $.
However, to prove that there is no encoding from \qCCS into \CQP, we have to show additionally that this operator can also not be emulated using measurement.
Therefore, we use the fact that measurement destroys entanglement.
More precisely, we consider 2-qubit systems and use a bell pair as starting state to prove that even with measurement the behaviour of $ \opQ $ cannot be emulated.

\begin{exa}[Counterexample]
	\label{exa:counterexample}
	Consider $ \opQ $ of Example~\ref{exa:phaseFlipChannel} applied to the second bit of a 2-qubit system:
	\begin{align*}
		\opQ_2\Tuple{\rho} ={}& \left( \opI \otimes E_0 \right) \rho \left( \opI \otimes E_0 \right)^{\dagger} + \left( \opI \otimes E_1 \right) \rho \left( \opI \otimes E_1 \right)^{\dagger}\\
		={}&
		\begin{pmatrix}
			\sqrt{0.5} & 0 & 0 & 0\\
			0 & \sqrt{0.5} & 0 & 0\\
			0 & 0 & \sqrt{0.5} & 0\\
			0 & 0 & 0 & \sqrt{0.5}
		\end{pmatrix}
		\rho
		\begin{pmatrix}
			\sqrt{0.5} & 0 & 0 & 0\\
			0 & \sqrt{0.5} & 0 & 0\\
			0 & 0 & \sqrt{0.5} & 0\\
			0 & 0 & 0 & \sqrt{0.5}
		\end{pmatrix}
		+{}\\
		&\begin{pmatrix}
			\sqrt{0.5} & 0 & 0 & 0\\
			0 & -\sqrt{0.5} & 0 & 0\\
			0 & 0 & \sqrt{0.5} & 0\\
			0 & 0 & 0 & -\sqrt{0.5}
		\end{pmatrix}
		\rho
		\begin{pmatrix}
			\sqrt{0.5} & 0 & 0 & 0\\
			0 & -\sqrt{0.5} & 0 & 0\\
			0 & 0 & \sqrt{0.5} & 0\\
			0 & 0 & 0 & -\sqrt{0.5}
		\end{pmatrix}
	\end{align*}
	Accordingly, $ \opQ_2\Tuple{x} = x $ for all $ x \in \Set{\OuterProduct{00}{00}, \OuterProduct{01}{01}, \OuterProduct{10}{10}, \OuterProduct{11}{11}} $, $ \opQ_2\Tuple{\OuterProduct{0+}{0+}} = \begin{pmatrix} 0.5 & 0 & 0 & 0\\ 0 & 0.5 & 0 & 0\\ 0 & 0 & 0 & 0\\ 0 & 0 & 0 & 0 \end{pmatrix} $, and $ \opQ_2\Tuple{\begin{pmatrix} 0.5 & 0 & 0 & 0.5\\ 0 & 0 & 0 & 0\\ 0 & 0 & 0 & 0\\ 0.5 & 0 & 0 & 0.5 \end{pmatrix}} = \begin{pmatrix} 0.5 & 0 & 0 & 0\\ 0 & 0 & 0 & 0\\ 0 & 0 & 0 & 0\\ 0 & 0 & 0 & 0.5 \end{pmatrix} $ for the bell pair that resembles $ \frac{1}{\sqrt{2}}\Ket{00} + \frac{1}{\sqrt{2}}\Ket{11} $.
	To observe this strange behaviour of $ \opQ $ we measure directly or apply Hadamard and then measure.
	Therefore we use the \OQS-terms
	\allowdisplaybreaks
	\begin{align*}
		S_{00} ={}& \CondOQS{00}{\opM\List{q_0, q_1}}{\tau.\success} + \CondOQS{01}{\opM\List{q_0, q_1}}{\tau.\nilOQS} +{}\\
		& \CondOQS{10}{\opM\List{q_0, q_1}}{\tau.\nilOQS} + \CondOQS{11}{\opM\List{q_0, q_1}}{\tau.\nilOQS}\\
		S_{01} ={}& \CondOQS{00}{\opM\List{q_0, q_1}}{\tau.\nilOQS} + \CondOQS{01}{\opM\List{q_0, q_1}}{\tau.\success} +{}\\
		& \CondOQS{10}{\opM\List{q_0, q_1}}{\tau.\nilOQS} + \CondOQS{11}{\opM\List{q_0, q_1}}{\tau.\nilOQS}\\
		S_{10} ={}& \CondOQS{00}{\opM\List{q_0, q_1}}{\tau.\nilOQS} + \CondOQS{01}{\opM\List{q_0, q_1}}{\tau.\nilOQS} +{}\\
		& \CondOQS{10}{\opM\List{q_0, q_1}}{\tau.\success} +  \CondOQS{11}{\opM\List{q_0, q_1}}{\tau.\nilOQS}\\
		S_{11} ={}& \CondOQS{00}{\opM\List{q_0, q_1}}{\tau.\nilOQS} + \CondOQS{01}{\opM\List{q_0, q_1}}{\tau.\nilOQS} +{}\\
		& \CondOQS{10}{\opM\List{q_0, q_1}}{\tau.\nilOQS} + \CondOQS{11}{\opM\List{q_0, q_1}}{\tau.\success}\\
		S_{00+11} ={}& \CondOQS{00}{\opM\List{q_0, q_1}}{\tau.\success} + \CondOQS{01}{\opM\List{q_0, q_1}}{\tau.\nilOQS} +{}\\
		& \CondOQS{10}{\opM\List{q_0, q_1}}{\tau.\nilOQS}  +\CondOQS{11}{\opM\List{q_0, q_1}}{\tau.\success}
	\end{align*}
	such that $ S_{ij} $ reaches success if and only if $ ij $ is measured and $ S_{00+11} $ reaches success if and only if $ 00 $ or $ 11 $ is measured.
	From that we build the \OQS-configurations
	\begin{align*}
		\mathsf{S_{ce1}}\Tuple{\rho} &= \ConfigOQS{\opQ\List{q_1}.S_{00}}{\rho}\\
		\mathsf{S_{ce2}}\Tuple{\rho} &= \ConfigOQS{\opQ\List{q_1}.S_{01}}{\rho}\\
		\mathsf{S_{ce3}}\Tuple{\rho} &= \ConfigOQS{\opQ\List{q_1}.S_{10}}{\rho}\\
		\mathsf{S_{ce4}}\Tuple{\rho} &= \ConfigOQS{\opQ\List{q_1}.S_{11}}{\rho}\\
		\mathsf{S_{ce5}}\Tuple{\rho} &= \ConfigOQS{\opQ\List{q_1}.\opH\List{q_1}.S_{01}}{\rho}\\
		\mathsf{S_{ce6}}\Tuple{\rho} &= \ConfigOQS{\opQ\List{q_1}.S_{00+11}}{\rho}
	\end{align*}
	for the 2-qubit system $ \rho = q_0, q_1 $.
	In particular, we use that $ \mathsf{S_{ce1}}\Tuple{\OuterProduct{00}{00}} $, $ \mathsf{S_{ce2}}\Tuple{\OuterProduct{01}{01}} $, $ \mathsf{S_{ce3}}\Tuple{\OuterProduct{10}{10}} $, and $ \mathsf{S_{ce4}}\Tuple{\OuterProduct{11}{11}} $ must reach success, whereas $ \mathsf{S_{ce5}}\Tuple{\OuterProduct{0+}{0+}} $ may but not must reach success, to show that $ \opQ $ cannot be emulated by unitary transformations.
	Since Hadamard $ \opH $ applied to $ \opQ\Tuple{\OuterProduct{+}{+}} $ is again $ \opQ\Tuple{\OuterProduct{+}{+}} $, we measure in $ \mathsf{S_{ce5}} $ after applying $ \opQ\List{q_1}.\opH\List{q_1} $ either $ 00 $ or $ 01 $ with equal probability.
	In the latter case success $ \success $ is unguarded, whereas the former case does not unguard success, \ie $ \mathsf{S_{ce5}}\Tuple{\OuterProduct{0+}{0+}} $ may but not must reach success.
	Finally, we use that $ \mathsf{S_{ce6}} $ for the bell pair that resembles $ \frac{1}{\sqrt{2}}\Ket{00} + \frac{1}{\sqrt{2}}\Ket{11} $ must reach success, to show that also measurement does not allow to emulate $ \opQ $.
	Note that the first qubit is only relevant for this last step, \ie for $ \mathsf{S_{ce6}} $.
	\qed
\end{exa}

An encoding from \qCCS or \OQS into \CQP or \CQS needs to emulate the behaviour of $ \opQ\List{q_1} $.
Since \CQP and \CQS do not allow for super-operators but only unitary transformations and since there is no unitary transformation with the same effect as $ \opQ\List{q_1} $, there is no good encoding from \OQS into \CQS or \qCCS into \CQP.
To prove this separation result we borrow a technical result from \cite{petersNestmannGoltz13}.
By success sensitiveness, a source term $ S $ reaches success if and only if its literal translation $ \Enc{S} $ reaches success.
As a consequence $ S $ cannot reach success if and only if $ \Enc{S} $ cannot reach success.
The next lemma shows that operational correspondence and success sensitiveness also imply that $ S $ must reach success, \ie reaches success in all finite traces, if and only if $ \Enc{S} $ must reach success.

\begin{lem}
	\label{lem:mustSuccessSensitiveness}
	For all operationally corresponding, success sensitive encodings $ \Enc{\cdot} $ \wrt some success respecting preorder $ \preceq $ on the target and for all source configurations $ S $, $ S $ must reach success in all finite traces iff $ \Enc{S} $ must reach success in all finite traces.
\end{lem}

\begin{proof}
	We consider both directions separately.
	\begin{compactdesc}
		\item[if $ S $ must reach success then also $ \Enc{S} $] Assume the opposite, \ie there is an encoding that satisfies the criteria operational soundness and success sensitiveness, $ \preceq $ is success respecting, and there is some source configuration $ S $ such that for all $ S' $ with $ S \steps S' $ we have $ \ReachBarb{S'}{\success} $, \ie $ S $ must reach success in all finite traces, but there is some target configuration $ T $ such that $ \Enc{S} \steps T $ and $ T $ cannot reach success.
		
		Since $ \Enc{\cdot} $ is operationally sound, $ \Enc{S} \steps T $ implies that there exist some $ S'', T'' $ such that $ S \steps S'' $, $ T \steps T'' $, and $ \Enc{S''} \preceq T'' $.
		Since $ T $ cannot reach success and $ T \steps T'' $, then $ T'' $ cannot reach success.
		Since $ \preceq $ respects success, $ \Enc{S''} \preceq T'' $ and that $ T'' $ cannot reach success imply that $ \Enc{S''} $ cannot reach success.
		Because $ \Enc{\cdot} $ is success sensitive, then also $ S'' $ cannot reach success, which contradicts the assumption that $ S $ must reach success.
		We conclude that if $ S $ must reach success in all finite traces then $ \Enc{S} $ must reach success in all finite traces.
		\item[if $ \Enc{S} $ must reach success then also $ S $] Assume the opposite, \ie there is an encoding that satisfies the criteria operational completeness and success sensitiveness, $ \preceq $ is success respecting, and there is some source configuration $ S $ such that for all $ T $ with $ \Enc{S} \steps T $ we have $ \ReachBarb{T}{\success} $, \ie $ \Enc{S} $ must reach success in all finite traces, but there is some source configuration $ S' $ such that $ S \steps S' $ and $ S' $ cannot reach success.
		
		Since $ \Enc{\cdot} $ is operationally complete, $ S \steps S' $ implies that there exists some $ T' $ such that $ \Enc{S} \steps T' $ and $ \Enc{S'} \preceq T' $.
		Because $ \Enc{\cdot} $ is success sensitive and $ S' $ cannot reach success, then also $ \Enc{S'} $ cannot reach success.
		Since $ \preceq $ respects success, $ \Enc{S'} \preceq T' $ and that $ \Enc{S'} $ cannot reach success imply that $ T' $ cannot reach success.
		Since $ T' $ cannot reach success and $ \Enc{S} \steps T' $, this contradicts the assumption that $ \Enc{S} $ must reach success.
		We conclude that if $ \Enc{S} $ must reach success in all finite traces then $ S $ must reach success in all finite traces.
		\qedhere
	\end{compactdesc}
\end{proof}

To prove the non-existence of an encoding from \OQS into \CQS, we use $ \opQ $ on a 2-quit system as described in Example~\ref{exa:counterexample} as a counterexample and show that it is not possible in \CQS to emulate the behaviour of $ \opQ\List{q_1} $ modulo compositionality, operational correspondence \wrt a success respecting preorder, and success sensitiveness.
More precisely, since there is no unitary transformation with this behaviour and also measurement or additional qubits do not help to emulate this behaviour on the state of the qubit (see the proof of Theorem~\ref{thm:separation}), there is no encoding from \OQS into \CQS that satisfies compositionality, operational correspondence \wrt a success respecting preorder, and success sensitiveness.

\begin{thm}
	\label{thm:separation}
	There is no encoding from \OQS into \CQS that satisfies compositionality, operational correspondence \wrt a success respecting preorder, and success sensitiveness.
\end{thm}

\begin{proof}
	The proof is by contradiction, \ie we assume that there is an encoding $ \Enc{\cdot} $ from \OQS into \CQS that satisfies compositionality, operational correspondence \wrt a success respecting preorder, and success sensitiveness.
	In \OQS we start with a configuration that contains two qubits (represented as a density matrix in $ \rho $).
	The encoding translates this \OQS-configuration into a \CQS-configuration such that its state is captured in a vector $ \sigma $.
	The encoding may use the qubits inside $ \rho $ directly for $ \sigma $, or it may measure these qubits and uses the information gained in this measurement to construct $ \sigma $.
	Remember that it is impossible to determine the exact state of a qubit and hence the entries for the density matrix.
	Using the original qubits directly results in a 2-qubit vector $ \sigma $.
	From measuring the original qubits we cannot gain more than two bit information such that we again capture all the information in a 2-qubit vector $ \sigma $.
	In other words, we can assume that the encoding translates a 2-qubit density matrix $ \rho $ into a 2-qubit vector $ \sigma $, because there is no more information available to justify the use of more qubits in \CQS, \ie systems with more qubits won't provide more information.

	By compositionality, then there is a \CQS-context $ \Context{}{\opQ}{\hole} $ such that
	\allowdisplaybreaks
	\begin{align*}
		\Enc{\mathsf{S_{ce1}}\Tuple{\rho}} &= \ConfigCQS{\sigma}{\phi_2}{\Context{}{\opQ}{T_1}}\\
		\Enc{\mathsf{S_{ce2}}\Tuple{\rho}} &= \ConfigCQS{\sigma}{\phi_2}{\Context{}{\opQ}{T_2}}\\
		\Enc{\mathsf{S_{ce3}}\Tuple{\rho}} &= \ConfigCQS{\sigma}{\phi_2}{\Context{}{\opQ}{T_3}}\\
		\Enc{\mathsf{S_{ce4}}\Tuple{\rho}} &= \ConfigCQS{\sigma}{\phi_2}{\Context{}{\opQ}{T_4}}\\
		\Enc{\mathsf{S_{ce5}}\Tuple{\rho}} &= \ConfigCQS{\sigma}{\phi_2}{\Context{}{\opQ}{T_5}}\\
		\Enc{\mathsf{S_{ce6}}\Tuple{\rho}} &= \ConfigCQS{\sigma}{\phi_2}{\Context{}{\opQ}{T_6}}
	\end{align*}
	where
	\begin{align*}
		\Enc{\ConfigOQS{S_{00}}{\rho}} &= \ConfigCQS{\sigma}{\phi_{\mathcal{M}}}{T_1}\\
		\Enc{\ConfigOQS{S_{01}}{\rho}} &= \ConfigCQS{\sigma}{\phi_{\mathcal{M}}}{T_2}\\
		\Enc{\ConfigOQS{S_{10}}{\rho}} &= \ConfigCQS{\sigma}{\phi_{\mathcal{M}}}{T_3}\\
		\Enc{\ConfigOQS{S_{11}}{\rho}} &= \ConfigCQS{\sigma}{\phi_{\mathcal{M}}}{T_4}\\
		\Enc{\ConfigOQS{\opH\List{q_1}.S_{01}}{\rho}} &= \ConfigCQS{\sigma}{\phi_{\mathcal{M}}}{T_5}\\
		\Enc{\ConfigOQS{S_{00+11}}{\rho}} &= \ConfigCQS{\sigma}{\phi_{\mathcal{M}}}{T_6}
	\end{align*}
	and $ \sigma $ is the translation of $ \rho $.
	Since the \OQS-configurations in the source are parametric on $ \rho $, the behaviour of the resulting \CQS-configurations depends on $ \sigma $.
	By operational correspondence and success sensitiveness, these contexts have to behave exactly as their respective sources \wrt the reachability of success (including the reachability of success in all finite traces as in Lemma~\ref{lem:mustSuccessSensitiveness}).
	Since the behaviour of the translations depends only on $ \sigma $ as input, we can focus on the translation of $ \opQ\List{q_1} $ on the quantum register $ \sigma $ that $ \Context{}{\opQ}{\hole} $ constructs from the input $ \rho $.
	In \CQP as well as \CQS the only operators with direct influence on the quantum register are unitary transformations, measurement, and the creation of new qubits.
	Moreover, \eg by communication or the probability distributions after measurement \CQP-configurations or \CQS-configurations can introduce branching and thus provide different results on different branches.

	With the creation of new qubits the size of the vector is increased.
	Intuitively, $ \Context{}{\opQ}{\hole} $ gets as input a 2-qubit vector and has to produce another 2-qubit vector as output, because $ T_1 - T_6 $ and $ \Context{}{\opQ}{\hole} $ require a 2-qubit vector.
	Because of that, the creation of new qubits can only contribute to $ \Context{}{\opQ}{\hole} $ by allowing to set a qubit to $ \Ket{0} $.
	Since this can also be done by measurement followed by a bit-flip if 1 was measured, we do not need to consider the creation of new qubits, \ie this behaviour is subsumed by the other operations.

	Note that we consider 2-bit vectors.
	Measuring one qubit in \CQP or \CQS creates a probability distribution with two cases that consist of their respective probability, which can be zero, followed by the configuration in the respective case.
	The overall evolution of closed systems---and \CQP and \CQS can express only closed systems---can be described by a unitary transformation.
	Accordingly, for the way in that $ \Context{}{\mathcal{Q}}{\hole} $ manipulates the 2-qubit vector the only relevant effect of measurement is (1)~that it creates branches, (2)~that some of these branches might have a zero-probability \wrt particular inputs but not necessarily all inputs, and (3)~that the evolution of the 2-qubit vector in every of these branches is described by a unitary transformation, at least if we consider as inputs only the values $ \Ket{00} $, $ \Ket{01} $, $ \Ket{10} $, $ \Ket{11} $, $ \Ket{0+} $, and $ \frac{1}{\sqrt{2}}\Ket{00} + \frac{1}{\sqrt{2}}\Ket{11} $.

	There are two sources for branching: either branching results from the probability distribution after measurement or from communication.
	Since the matrix multiplication of two unitary transformations is again a unitary transformation, sequences of unitary transformations can be abbreviated by a single unitary transformation.
	Accordingly, if we consider a single branch without measurement in $ \Context{}{\mathcal{Q}}{\hole} $ from the beginning to the end, the transformation on the 2-qubit vector can be abbreviated by a single unitary transformation that is a $ 4 \times 4 $-matrix.

	Assume that all branches in $ \Context{}{\opQ}{\hole} $ result from communication, \ie $ \Context{}{\opQ}{\hole} $ does not use measurement.
	Then for every branch there is a unitary transformation $ U = \begin{pmatrix} u_{11} & u_{12} & u_{13} & u_{14}\\ u_{21} & u_{22} & u_{23} & u_{24}\\ u_{31} & u_{32} & u_{33} & u_{34}\\ u_{41} & u_{42} & u_{43} & u_{44} \end{pmatrix} $ that emulates $ \opQ\List{q_1} $ in this branch.
	Considering the behaviour of $ \mathsf{S_{ce1}}\Tuple{\OuterProduct{00}{00}} $ it follows
	\begin{align*}
		\begin{pmatrix}
			u_{11} & u_{12} & u_{13} & u_{14}\\
			u_{21} & u_{22} & u_{23} & u_{24}\\
			u_{31} & u_{32} & u_{33} & u_{34}\\
			u_{41} & u_{42} & u_{43} & u_{44}
		\end{pmatrix}
		\begin{pmatrix} 1\\ 0\\ 0\\ 0 \end{pmatrix} =
		\begin{pmatrix} 1\\ 0\\ 0\\ 0 \end{pmatrix}
		\quad
		\begin{array}{ll}
			\text{and therefore } & u_{11} = 1\\
			\text{and } & u_{21} = u_{31} = u_{41} = 0
		\end{array}
	\end{align*}
	for all branches, because $ \Ket{00} = \Transp{\Tuple{1, 0, 0, 0}} $ is the only state such that $ T_1 $ applied to this state always unguards $ \success $ and $ \mathsf{S_{ce1}}\Tuple{\OuterProduct{00}{00}} $ must reach success.
	Repeating this calculation for $ \mathsf{S_{ce2}}\Tuple{\OuterProduct{01}{01}} $, $ \mathsf{S_{ce3}}\Tuple{\OuterProduct{10}{10}} $, and $ \mathsf{S_{ce1}}\Tuple{\OuterProduct{11}{11}} $, we conclude that $ u_{ii} = 1 $ for all $ i \in \Set{1, 2, 3, 4} $ and $ u_{ij} = 0 $ for all $ i, j \in \Set{1, 2, 3, 4} $ with $ i \neq j $, \ie that $ U = \opI \otimes \opI $ is identity.
	But, if we apply this identity transformation $ U $ to $ \Ket{0+} $ we obtain $ \Ket{0+} $ and $ T_5 $ applied in this state cannot reach success, whereas $ \mathsf{S_{ce5}}\Tuple{\OuterProduct{0+}{0+}} $ may reach success.
	This is a contradiction, \ie there is no such unitary transformation that emulates $ \opQ\List{q_1} $.
	Therefore, our assumption that all branches in $ \Context{}{\opQ}{\hole} $ result from communication must be wrong, \ie $ \Context{}{\opQ}{\hole} $ has to measure.

	Of course $ \Context{}{\opQ}{\hole} $ may consist of a sequence of steps containing several measurements.
	From $ \mathsf{S_{ce1}}\Tuple{\OuterProduct{00}{00}} $, $ \mathsf{S_{ce2}}\Tuple{\OuterProduct{01}{01}} $, $ \mathsf{S_{ce3}}\Tuple{\OuterProduct{10}{10}} $, and $ \mathsf{S_{ce4}}\Tuple{\OuterProduct{11}{11}} $ it is obvious that measuring the first qubit does not contribute to the implementation of $ \opQ\List{q_1} $.
	It suffices to consider implementations of $ \Context{}{\opQ}{\hole} $ that measure only the second qubit.
	More precisely, we consider only the last measurement of the second qubit that is performed in $ \Context{}{\opQ}{\hole} $ in each of its branches.
	Without loss of generality we can assume that this measurement was performed \wrt the standard base, because all other cases can be implemented by a unitary transformation right before the measurement.
	Then there are two possible outcomes of every last measurement, $ \Ket{0} $ and $ \Ket{1} $, \ie there are two possible branches but one of them might occur with probability zero.
	As usual we ignore branches that occur with probability zero.
	All transformations in $ \Context{}{\opQ}{\hole} $ after the last measurement can again be subsumed in a single unitary transformation.
	Accordingly, $ \Context{}{\opQ}{\hole} $ does perform some arbitrary initial steps that may contain an arbitrary number of measurements and might produce an arbitrary number of branches and each branch with measurement ends with the final measurement of the second qubit that produces one or two branches whose behaviour after the final measurement can be described respectively by a single unitary transformation.

	We consider once more the case $ \mathsf{S_{ce1}}\Tuple{\OuterProduct{00}{00}} $.
	The last measurement of $ q_2 $ sets in every branch the qubit $ q_2 $ in $ \sigma $ to $ \Ket{0} $ or $ \Ket{1} $.
	Since $ \Ket{00} $ is the only state such that $ T_1 $ applied to this state always unguards $ \success $ and $ \mathsf{S_{ce1}}\Tuple{\OuterProduct{00}{00}} $ must reach success, the unitary transformation after the last measurement has to map the current state in every branch to $ \Ket{00} $.
	Let us call this unitary transformation $ U_0 $.
	Note that for instance $ U_0 = \opI \otimes \opI $ would do the job, if the first qubit is still in state $ \Ket{0} $ before its application.
	Similarly, in all branches in that $ 1 $ was measured, the unitary transformation has to result in $ \Ket{00} $.
	Let us call this transformation $ U_1 $ and note that \eg $ \opI \otimes \opX $ can do this, if the first qubit is still in state $ \Ket{0} $.
	Accordingly, in all branches in that the last measurement results in $ \Ket{0} $ this measurement is followed by $ U_0 $ and in all branches in that the last measurement results in $ \Ket{1} $ this measurement is followed by $ U_1 $, because this ensures that each branch of $ \Context{}{\opQ}{\hole} $ for $ \Ket{00} $ finally results in $ \Ket{00} $ as required by $ T_1 $.

	We apply the same argumentation for $ \Ket{01} $ instead of $ \Ket{00} $ and $ \mathsf{S_{ce2}}\Tuple{\OuterProduct{01}{01}} $ instead of $ \mathsf{S_{ce1}}\Tuple{\OuterProduct{00}{00}} $ to obtain the following: In all branches in that the last measurement results in $ \Ket{0} $ this measurement is followed by $ U_0 $ and in all branches in that the last measurement results in $ \Ket{1} $ this measurement is followed by some $ U_1 $ such that $ U_0, U_1 $ both ensure that the respective branch of $ \Context{}{\opQ}{\hole} $ for $ \Ket{01} $ finally results in $ \Ket{01} $.

	Note that this is not yet a contradiction.
	By compositionality, $ \Context{}{\opQ}{\hole} $ has to be implemented by the same term regardless of whether we start with $ \Ket{00} $ or $ \Ket{01} $ and, thus, the mentioned $ U_0 $ and $ U_1 $ indeed have to be the same in both cases.
	And, obviously, there is no $ U_0 $ that applied to $ \Ket{0} $ for the second qubit sometimes results in $ \Ket{0} $ and sometimes in $ \Ket{1} $.
	But we do not necessarily always have two branches as result of measurement.
	So there are so far still two plausible scenarios:
	Either if we start with $ \Ket{0} $ for the second qubit only $ 0 $ is measured and if we start with $ \Ket{1} $ for the second qubit only $ 1 $ is measured or vice versa.
	In the former case we could \eg pick $ U_0 = U_1 = \opI \times \opI $ and in the latter case we could \eg pick $ U_0 = U_1 = \opI \times \opX $ (if the first qubit remains in its initial state).
	However, we have a contradiction for the case $ \frac{1}{\sqrt{2}}\Ket{00} + \frac{1}{\sqrt{2}}\Ket{11} $.

	In the state $ \frac{1}{\sqrt{2}}\Ket{00} + \frac{1}{\sqrt{2}}\Ket{11} $ measuring the second qubit we obtain either $ 0 $ or $ 1 $ with equal probability.
	Because of that, the implementation of $ \Context{}{\opQ}{\hole} $ will have at least two branches with measurement on the second qubit such that in one branch after the last measurement of the second qubit $ U_0 $ is applied and in the other branch after the last measurement of the second qubit $ U_1 $ is applied.
	For both of the two plausible scenarios that are left, this means that in one branch the second qubit is set to $ \Ket{0} $ and in the other to $ \Ket{1} $.
	Note that the entanglement between the two qubits is destroyed (if not before then by this last measurement).
	Then it cannot be avoided that a subsequent measurement of both qubits will result in different values.
	This is in contradiction to $ \mathsf{S_{ce6}} $, because $ \mathsf{S_{ce6}} $ applied on the considered bell pair must reach success and therefore $ T_6 $ requires two qubits that always return the same value in measurement.

	Accordingly, our original assumption, \ie that there is an encoding $ \Enc{\cdot} $ from \OQS into \CQS that satisfies compositionality, operational correspondence \wrt a success respecting preorder, and success sensitiveness is wrong: there is no such encoding.
\end{proof}

As we claim, the counterexample in Example~\ref{exa:counterexample} can be expressed similarly, \ie with strongly bisimilar behaviour, in variants of \qCCS with measurement operators as in \cite{FengDuanJiYing07,FengDuanYing12}.
Moreover, even the full expressive power of \CQP does not help to correctly emulate this super-operator.
Hence, there is also no encoding from \qCCS into \CQP.

\begin{cor}
	There is no encoding from \qCCS with a measurement operator into \CQP that satisfies compositionality, operational correspondence \wrt a success respecting preorder, and success sensitiveness.
\end{cor}


\section{Quality Criteria for Quantum Based Systems}
\label{sec:criteriaQBS}

Sections~\ref{sec:encoding} and \ref{sec:separation} show that the quality criteria of Gorla in \cite{gorla10} can be applied to quantum based systems and are still meaningful in this setting.
They might, however, not be exhaustive, \ie there might be aspects of quantum based systems that are relevant but not sufficiently covered by this set of criteria.
To obtain these criteria, Gorla studied a large number of encodings, \ie this set of criteria was built upon the experience of many researchers and years of work.
Accordingly, we do not expect to answer the question 'what are good quality criteria for quantum based systems' now, but rather want to start the discussion.

A closer look at the criteria in Section~\ref{sec:criteria} reveals a first candidate for an additional quality criterion.
Name invariance ensures that encodings cannot cheat by treating names differently. It requires that good encodings preserve substitutions to some extend.
\CQP and \qCCS model the dynamics of quantum registers in fundamentally different ways, but both languages address qubits by qubit names.
It seems natural to extend name invariance to also cover qubit names.

As in \cite{gorla10}, we let our definition of qubit invariance depend on a renaming policy $ \varphi $, where this renaming policy is for qubit names.
The renaming policy translates qubit names of the source to tuples of qubit names in the target, \ie $ \varphi: \var \to \var^n $, where we require that $ \varphi(q) \cap \varphi(q') = \emptyset $ whenever $ q \neq q' $.

The new criterion \emph{qubit invariance}, then requires that encodings preserve and reflect substitutions on qubits modulo the renaming policy on qubits.

\begin{defi}[Qubit Invariance]
	The encoding $ \Enc{\cdot} $ is \emph{qubit invariant} if, for every $ S \in \sourceConfig $ and every substitution $ \gamma $ on qubit names, it holds that $ \Enc{S\gamma} = \Enc{S}\gamma' $, where $ \varphi(\gamma(q)) = \gamma'(\varphi(q)) $ for every $ q \in \var $.
\end{defi}

In \cite{gorla10}, name invariance allows the slightly weaker condition $ \Enc{S\gamma} \preceq \Enc{S}\gamma' $ for non-injective substitutions.
In contrast, substitutions on qubits always have to be injective such that they cannot violate the no-cloning principle.
Since $ \Enc{\cdot} $ translates qubit names to themselves and introduces no other qubit names, it satisfies qubit invariance for $ \varphi $ being the identity and $ \gamma' = \gamma $. The corresponding proof is given above in Lemma~\ref{lem:qubitInvariance}.

Note that the qubits discussed so far are so-called \emph{logical qubits}, \ie they are abstractions of the physical qubits.
To implement a single \emph{logical qubit} as of today several \emph{physical qubits} are necessary.
These additional physical qubits are used to ensure stability and fault-tolerance in the implementation of logical qubits.
Since the number of necessary physical qubits can be much larger than the number of logical qubits, already a small increase in the number of logical qubits might seriously limit the practicability of a system.
Accordingly, one may require that encodings preserve the number of logical qubits.

\begin{defi}[Size of Quantum Registers]
	An encoding $ \Enc{\cdot} $ \emph{preserves the size of quantum registers}, if for all $ S \in \sourceConfig $, the number of qubits in $ \Enc{S} $ is not greater than in $ S $.
\end{defi}

Again, the encoding $ \Enc{\cdot} $ in Definition~\ref{def:encoding} satisfies this criterion, which can be verified easily by inspection of the encoding function.

\begin{lem}
	\label{lem:sizeOfRegister}
	The encoding $ \Enc{\cdot} $ preserves the size of quantum registers, \ie for all $ S \in \sourceConfig $, the number of qubits in $ \Enc{S} $ is not greater than in $ S $.
\end{lem}

\begin{proof}
	By Definition~\ref{def:encoding}, the number of qubits in $ \Enc{S} $ is the same as the number of qubits in $ S $.
	Moreover, $ \Enc{\cdot} $ does not introduce new qubits in any of its cases except as the encoding of the creation of a new qubit in the source.
	Because of that, also the derivatives of source term translations have the same number of qubits as their respective source term equivalents.
	Thus, $ \Enc{\cdot} $ preserves the size of quantum registers.
\end{proof}

Similarly to success sensitiveness, requiring the preservation of the size of quantum registers on literal encodings is not enough.
To ensure that all reachable target terms preserve the size of quantum registers, we again link this criterion with the target term relation $ \preceq $.
More precisely, we require that $ \preceq $ is sensible to the size of quantum registers, \ie $ T_1 \preceq T_2 $ implies that the quantum registers in $ T_1 $ and $ T_2 $ have the same size.
The correspondence simulation $ \preceq $ that we used as target relation for the encoding $ \Enc{\cdot} $ is not sensible to the size of quantum registers, but we can easily turn it into such a relation.
Therefore, we simply add the condition that $ \Length{\rho} = \Length{\sigma} $ whenever $ \ConfigOQS{P}{\rho}\mathcal{R}\ConfigOQS{Q}{\sigma} $ to Definition~\ref{def:correspondenceSimulation}.
Fortunately, all of the already shown results remain valid for the altered version of $ \preceq $.

In contrast to \CQP, the semantics of \OQS yields a non-probabilistic transition system, where probabilities are captured in the density matrices.
The encoding $ \Enc{\cdot} $ translates probability distributions into non-deterministic choices.
Thereby, branches with zero probability are correctly eliminated, but all remaining branches are treated similarly and their probabilities are forgotten.
To check also the probabilities of branches, we can strengthen operational correspondence \eg to a labelled variant, where labels capture the probability of a step.
The challenge here is to create a meaningful criterion that correctly accumulates the probabilities in sequences of steps as \eg a single source term step might be translated into a sequence of target term steps, but the product of the probabilities contained in the sequence has to be equal to the probability of the single source term step.
As, to the best of our knowledge, there are no well-accepted probabilistic versions of operational correspondence.
Because of that, we started to study probabilistic versions of operational correspondence and the nature of the relation between source and target they imply.
Just recently we were able to publish three variants of probabilistic operational correspondence \cite{schmittPeters23}.
These criteria allow to more closely and more naturally connect the usually probabilistic quantum based systems.

Another important aspect is in how far the quality criteria capture the fundamental principles of quantum based systems such as the \emph{no-cloning principle}: By the laws of quantum mechanics, it is not possible to exactly copy a qubit.
Technically, such a copying would require some form of interaction with the qubit and this interaction would destroy its superposition, \ie alter its state.
Interestingly, the criteria of Gorla are even strong enough to observe a violation of this principle in the encoding from \CQS into \OQS, \ie if we allow \CQS to violate this principle but require that \OQS respects it, then we obtain a negative result.
Therefore, we remove the type system from \CQS. Without this type system, we can use the same qubit at different locations, violating the no-cloning principle.
As an example, consider $ S = \ConfigCQS{\sigma}{\phi}{\OutCQS{c}{q}{\nilCQS} \mid \OutCQS{c}{q}{\nilCQS}} $.
Then the encoding $ \Enc{\cdot} $ in Definition~\ref{def:encoding} is not valid any more, because $ \Enc{S} = \ConfigOQS{\ResOQS{\Tuple{\OutOQS{c}{q}{\nilOQS} \parOQS \OutOQS{c}{q}{\nilOQS}}}{\phi}}{\rho} $ violates condition~\ref{condB}.
Using $ S $ as counterexample, it should be possible to show that there exists no encoding that satisfies compositionality, operational correspondence, and success sensitiveness.

Of course, even if we succeed with this proof, this does not imply that the criteria are strong enough to sufficiently capture the no-cloning principle.
Indeed, the other direction is more interesting, \ie criteria that rule out encodings such that the source language respects the no-cloning principle but not all literal translations or their derivatives respect it.
We believe that capturing the no-cloning principle and the other fundamental principles of quantum based systems is an interesting research challenge.


\section{Conclusions}
\label{sec:conclusions}

We proved that \CQS can be encoded by \OQS \wrt the quality criteria compositionality, name invariance, operational correspondence, divergence reflection, and success sensitiveness.
Additionally, this encoding satisfies two new, quantum specific criteria: it is invariant to qubit names and preserves the size of quantum registers.
We think that these new criteria are relevant for translations between quantum based systems.

The encoding proves that the way in that \qCCS treats qubits|using density matrices and super-operators|can emulate the way in that \CQP treats qubits.
The other direction is more difficult.
We showed that there exists no encoding from \OQS into \CQS that satisfies compositionality, operational correspondence, and success sensitiveness and claim that this also implies that there is no encoding from \qCCS into \CQP.

The results themselves may not necessarily be very surprising.
The unitary transformations used in \CQS/\CQP are a subset of the super-operators used in \OQS/\qCCS and also density matrices can express more than the vectors used in \CQS/\CQP.
What our case study proves is that the quality criteria that were originally designed for classical systems are still meaningful in this quantum based setting.
They may, however, not be exhaustive.
Accordingly, in Section~\ref{sec:criteriaQBS} we start the discussion on quality criteria for this new setting of quantum based systems.
The first two candidate criteria that we propose, namely qubit invariance and preservation of quantum register sizes, are relevant, but rather basic.
Since the semantics of quantum based systems is often probabilistic, a variant of operational correspondence that requires the preservation and reflection of probabilities in the respective traces might be meaningful.
The encoding $ \Enc{\cdot} $ presented above does not satisfy probabilistic operational correspondence as presented in \cite{schmittPeters23}.
More difficult and thus also more interesting are criteria that capture the fundamental principles of quantum based systems such as the no-cloning principle.
Hereby, we pose the task of identifying such criteria as research challenge.


\bibliographystyle{alphaurl}
\bibliography{LMCS_encoding_CQP_qCCS_lit}

\newcommand{\etalchar}[1]{$^{#1}$}
\begin{thebibliography}{DRMT{\etalchar{+}}04}

\bibitem[BBC{\etalchar{+}}93]{bennettBassardCrepeauJoszsaPeresWootters93}
Charles~H. Bennett, Gilles Brassard, Claude Cr\'epeau, Richard Jozsa, Asher
  Peres, and William~K. Wootters.
\newblock {Teleporting an unknown quantum state via dual classical and
  Einstein-Podolsky-Rosen channels}.
\newblock {\em Physical Review Letters}, 70:1895--1899, 1993.
\newblock \href {https://doi.org/10.1103/PhysRevLett.70.1895}
  {\path{doi:10.1103/PhysRevLett.70.1895}}.

\bibitem[BNP20]{bispingNestmannPeters20}
Benjamin Bisping, Uwe Nestmann, and Kirstin Peters.
\newblock Coupled similarity: the first 32 years.
\newblock {\em Acta Informatica}, 57(3--5):439--463, 2020.
\newblock \href {https://doi.org/10.1007/s00236-019-00356-4}
  {\path{doi:10.1007/s00236-019-00356-4}}.

\bibitem[DGNP12]{DavidsonGayNagarajanPuthoor12}
Timothy A.~S. Davidson, Simon~J. Gay, Rajagopal Nagarajan, and Ittoop~V.
  Puthoor.
\newblock {Analysis of a Quantum Error Correcting Code using Quantum Process
  Calculus}.
\newblock In {\em Proceedings of QPL}, volume~95 of {\em EPTCS}, pages 67--80,
  2012.
\newblock \href {https://doi.org/10.4204/EPTCS.95.7}
  {\path{doi:10.4204/EPTCS.95.7}}.

\bibitem[DRMT{\etalchar{+}}04]{RiedmattenMarcikicTittelZbindenCollinsGisin04}
Hugues De~Riedmatten, Ivan Marcikic, Wolfgang Tittel, Hugo Zbinden, Daniel
  Collins, and Nicolas Gisin.
\newblock {Long Distance Quantum Teleportation in a Quantum Relay
  Configuration}.
\newblock {\em Physical Review Letters}, 92:047904, 2004.
\newblock \href {https://doi.org/10.1103/PhysRevLett.92.047904}
  {\path{doi:10.1103/PhysRevLett.92.047904}}.

\bibitem[FDJY07]{FengDuanJiYing07}
Yuan Feng, Runyao Duan, Zhengfeng Ji, and Mingsheng Ying.
\newblock Probabilistic bisimulations for quantum processes.
\newblock {\em Information and Computation}, 205(11):1608--1639, 2007.
\newblock \href {https://doi.org/10.1016/j.ic.2007.08.001}
  {\path{doi:10.1016/j.ic.2007.08.001}}.

\bibitem[FDY12]{FengDuanYing12}
Yuan Feng, Runyao Duan, and Mingsheng Ying.
\newblock {Bisimulation for Quantum Processes}.
\newblock {\em ACM Transactions on Programming Languages and Systems}, 34(4),
  2012.
\newblock \href {https://doi.org/10.1145/2400676.2400680}
  {\path{doi:10.1145/2400676.2400680}}.

\bibitem[FGP13]{FrankeArnoldGayPuthoor13}
Sonja Franke{-}Arnold, Simon~J. Gay, and Ittoop~V. Puthoor.
\newblock {Quantum Process Calculus for Linear Optical Quantum Computing}.
\newblock In {\em Proceedings of Reversible Computation}, volume 7948 of {\em
  LNCS}, pages 234--246. Springer, 2013.
\newblock \href {https://doi.org/10.1007/978-3-642-38986-3_19}
  {\path{doi:10.1007/978-3-642-38986-3_19}}.

\bibitem[FGP14]{FrankeArnoldGayPuthoor14}
Sonja Franke{-}Arnold, Simon~J. Gay, and Ittoop~V. Puthoor.
\newblock {Verification of Linear Optical Quantum Computing using Quantum
  Process Calculus}.
\newblock In {\em Preoceedings of EXPRESS/SOS}, volume 160 of {\em EPTCS},
  pages 111--129, 2014.
\newblock \href {https://doi.org/10.4204/EPTCS.160.10}
  {\path{doi:10.4204/EPTCS.160.10}}.

\bibitem[Gay06]{gay06}
Simon~J. Gay.
\newblock {Quantum programming languages: survey and bibliography}.
\newblock {\em Mathematical Structures of Computer Science}, 16(4):581--600,
  2006.
\newblock \href {https://doi.org/10.1017/S0960129506005378}
  {\path{doi:10.1017/S0960129506005378}}.

\bibitem[GN05]{GayNagarajan05}
Simon~J. Gay and Rajagopal Nagarajan.
\newblock Communicating quantum processes.
\newblock In {\em Proceedings of POPL}, pages 145--157, 2005.
\newblock \href {https://doi.org/10.1145/1040305.1040318}
  {\path{doi:10.1145/1040305.1040318}}.

\bibitem[Gor10]{gorla10}
Daniele Gorla.
\newblock {Towards a Unified Approach to Encodability and Separation Results
  for Process Calculi}.
\newblock {\em Information and Computation}, 208(9):1031--1053, 2010.
\newblock \href {https://doi.org/10.1016/j.ic.2010.05.002}
  {\path{doi:10.1016/j.ic.2010.05.002}}.

\bibitem[GP12]{GayPuthoor12}
Simon~J. Gay and Ittoop~V. Puthoor.
\newblock {Application of Quantum Process Calculus to Higher Dimensional
  Quantum Protocols}.
\newblock In {\em Proceedings of Quantum Physics and Logic}, volume 158 of {\em
  EPTCS}, pages 15--28, 2012.
\newblock \href {https://doi.org/10.4204/EPTCS.158.2}
  {\path{doi:10.4204/EPTCS.158.2}}.

\bibitem[Gru09]{gruska09}
Jozef Gruska.
\newblock {Quantum Computing}.
\newblock In {\em Wiley Encyclopedia of Computer Science and Engineering}. John
  Wiley {\&} Sons, Inc., 2009.
\newblock \href {https://doi.org/10.1002/9780470050118.ecse720}
  {\path{doi:10.1002/9780470050118.ecse720}}.

\bibitem[JL04]{JorrandLalire04}
Philippe Jorrand and Marie Lalire.
\newblock {Toward a Quantum Process Algebra}.
\newblock In {\em Proceedings of CF}, pages 111--119, 2004.
\newblock \href {https://doi.org/10.1145/977091.977108}
  {\path{doi:10.1145/977091.977108}}.

\bibitem[KKK{\etalchar{+}}12]{KubotaKakutaniKatoKawanoSakurada12}
Takahiro Kubota, Yoshihiko Kakutani, Go~Kato, Yasuhito Kawano, and Hideki
  Sakurada.
\newblock {Application of a process calculus to security proofs of quantum
  protocols}.
\newblock In {\em Proceedings of FCS}, 2012.

\bibitem[KKK{\etalchar{+}}13]{KubotaKakutaniKatoKawanoSakurada13}
Takahiro Kubota, Yoshihiko Kakutani, Go~Kato, Yasuhito Kawano, and Hideki
  Sakurada.
\newblock {Automated Verification of Equivalence on Quantum Cryptographic
  Protocols}.
\newblock {\em EPiC Series in Computing}, 15:64--69, 2013.

\bibitem[KPY16]{KouzapasPerezYoshida16}
Dimitrios Kouzapas, Jorge~A. P{\'e}rez, and Nobuko Yoshida.
\newblock {On the Relative Expressiveness of Higher-Order Session Processes}.
\newblock In {\em Proceedings of ESOP}, volume 9632 of {\em LNCS}, pages
  446--475, 2016.
\newblock \href {https://doi.org/10.1007/978-3-662-49498-1_18}
  {\path{doi:10.1007/978-3-662-49498-1_18}}.

\bibitem[NC10]{NielsenChuang10}
Michael~A. Nielsen and Isaac~L. Chuang.
\newblock {\em {Quantum Computation and Quantum Information (10th Anniversary
  edition)}}.
\newblock Cambridge University Press, 2010.
\newblock \href {https://doi.org/10.1017/cbo9780511976667.016}
  {\path{doi:10.1017/cbo9780511976667.016}}.

\bibitem[Pet19]{peters19}
Kirstin Peters.
\newblock {Comparing Process Calculi Using Encodings}.
\newblock In {\em Proceedings of EXPRESS/SOS}, volume 300 of {\em EPCTS}, pages
  19--38, 2019.
\newblock \href {https://doi.org/10.4204/EPTCS.300.2}
  {\path{doi:10.4204/EPTCS.300.2}}.

\bibitem[Plo04]{Plotkin04}
Gordon~D. Plotkin.
\newblock {A Structural Approach to Operational Semantics}.
\newblock {\em Journal of Logic and Algebraic Programming}, 60:17--140, 2004.
\newblock [An earlier version of this paper was published as technical report
  at Aarhus University in 1981.].

\bibitem[PNG13]{petersNestmannGoltz13}
Kirstin Peters, Uwe Nestmann, and Ursula Goltz.
\newblock {On Distributability in Process Calculi}.
\newblock In {\em Proceedings of ESOP}, volume 7792 of {\em LNCS}, pages
  310--329, 2013.
\newblock \href {https://doi.org/10.1007/978-3-642-37036-6_18}
  {\path{doi:10.1007/978-3-642-37036-6_18}}.

\bibitem[PvG15]{petersGlabbeek15}
Kirstin Peters and Rob van Glabbeek.
\newblock {Analysing and Comparing Encodability Criteria}.
\newblock In {\em Proceedings of EXPRESS/SOS}, volume 190 of {\em EPTCS}, pages
  46--60, 2015.
\newblock \href {https://doi.org/10.4204/EPTCS.190.4}
  {\path{doi:10.4204/EPTCS.190.4}}.

\bibitem[RP00]{RieffelPolak00}
Eleanor Rieffel and Wolfgang Polak.
\newblock An introduction to quantum computing for non-physicists.
\newblock {\em {ACM Computing Surveys}}, 32(3):300--335, 2000.
\newblock \href {https://doi.org/10.1145/367701.367709}
  {\path{doi:10.1145/367701.367709}}.

\bibitem[SP23]{schmittPeters23}
Anna Schmitt and Kirstin Peters.
\newblock {Probabilistic Operational Correspondence}.
\newblock In {\em Proceedings of CONCUR}, volume 279 of {\em LIPIcs}, pages
  15:1--15:17, 2023.
\newblock \href {https://doi.org/10.4230/LIPIcs.CONCUR.2023.15}
  {\path{doi:10.4230/LIPIcs.CONCUR.2023.15}}.

\bibitem[SPD22a]{schmittPetersDeng22}
Anna Schmitt, Kirstin Peters, and Yuxin Deng.
\newblock {Encodability Criteria for Quantum Based Systems}.
\newblock In {\em Proceedings of Forte}, volume 13273 of {\em LNCS}, pages
  151--169. Springer, 2022.
\newblock \href {https://doi.org/10.1007/978-3-031-08679-3_10}
  {\path{doi:10.1007/978-3-031-08679-3_10}}.

\bibitem[SPD22b]{schmittPetersDengTec22}
Anna Schmitt, Kirstin Peters, and Yuxin Deng.
\newblock {Encodability Criteria for Quantum Based Systems (Technical Report)}.
\newblock Technical report, 2022.
\newblock \href {https://doi.org/10.48550/ARXIV.2204.06068}
  {\path{doi:10.48550/ARXIV.2204.06068}}.

\bibitem[YFDJ09]{YingFengDuanJi09}
Mingsheng Ying, Yuan Feng, Runyao Duan, and Zhengfeng Ji.
\newblock An algebra of quantum processes.
\newblock {\em {ACM Transactions on Computational Logic}}, 10(3):1--36, 2009.
\newblock \href {https://doi.org/10.1145/1507244.1507249}
  {\path{doi:10.1145/1507244.1507249}}.

\bibitem[YKK14]{YasudaKubotaKakutani14}
Kazuya Yasuda, Takahiro Kubota, and Yoshihiko Kakutani.
\newblock {Observational Equivalence Using Schedulers for Quantum Processes}.
\newblock In {\em Proceedings of Quantum Physics and Logic}, volume 172 of {\em
  {EPTCS}}, pages 191--203, 2014.
\newblock \href {https://doi.org/10.4204/EPTCS.172.13}
  {\path{doi:10.4204/EPTCS.172.13}}.

\end{thebibliography}


\begin{appendix}

\section{Type System of Closed Quantum Systems}
\label{sec:typeSystemCQS}

Lemma~\ref{lem:typingFreeQubitsCQS} states that:
\begin{quotation}
	If $ \Sigma \vdash P $ then $ \FreeQubits{P} \subseteq \Sigma $.
\end{quotation}

\begin{proof}[Proof of Lemma~\ref{lem:typingFreeQubitsCQS}]
	Assume $ \Sigma \vdash P $.
	We perform an induction on the structure of $ P $.
	\begin{description}
		\item[$ P = \nilCQS $] Then $ \FreeQubits{P} = \emptyset \subseteq \Sigma $.
		\item[$ P = \success $] Then $ \FreeQubits{P} = \emptyset \subseteq \Sigma $.
		\item[$ P = Q \mid R $] By \ruleTPar, then there are $ \Sigma_1, \Sigma_2 $ such that $ \Sigma_1 \vdash Q $, $ \Sigma_2 \vdash R $, and $ \Sigma = \Sigma_1 \cup \Sigma_2 $.
			By the induction hypothesis, then $ \FreeQubits{Q} \subseteq \Sigma_1 $ and $ \FreeQubits{R} \subseteq \Sigma_2 $.
			Since $ \FreeQubits{P} = \FreeQubits{Q} \cup \FreeQubits{R} $, then $ \FreeQubits{P} \subseteq \Sigma $.
		\item[$ P = \InpCQS{c}{x}{Q} $] By \ruleTIn, then $ c \in \names $, $ x \in \var \setminus \Sigma $, and $ \Sigma \cup \Set{x} \vdash Q $.
			By the induction hypothesis, then $ \FreeQubits{Q} \subseteq \Sigma \cup \Set{x} $.
			Since $ \FreeQubits{P} = \FreeQubits{Q} \setminus \Set{x} $, then $ \FreeQubits{P} \subseteq \Sigma $.
		\item[$ P = \OutCQS{c}{x}{Q} $] By \ruleTOut, then $ c \in \names $, $ x \in \var \cap \Sigma $, and $ \Sigma \setminus \Set{x} \vdash Q $.
			By the induction hypothesis, then $ \FreeQubits{Q} \subseteq \Sigma \setminus \Set{x} $.
			Since $ \FreeQubits{P} = \FreeQubits{Q} \cup \Set{x} $, then $ \FreeQubits{P} \subseteq \Sigma $.
		\item[$ P = \UnitaryCQS{x_1, \ldots, x_n}{U}{Q} $] By \ruleTTrans, then $ x_1, \ldots, x_n \in \var \cap \Sigma $, $ \vdash \At{U}{\OpType{n}} $, and $ \Sigma \vdash Q $.
			By the induction hypothesis, then $ \FreeQubits{Q} \subseteq \Sigma $.
			Since $ \FreeQubits{P} = \FreeQubits{Q} $, then $ \FreeQubits{P} \subseteq \Sigma $.
		\item[$ P = \MeasCQS{v'}{x_1, \ldots, x_n}{Q} $] By \ruleTMsure, then $ v' \in \binaries $, $ x_1, \ldots, x_n \in \var \cap \Sigma $, and $ \Sigma \vdash Q $.
			By the induction hypothesis, then $ \FreeQubits{Q} \subseteq \Sigma $.
			Since $ \FreeQubits{P} = \FreeQubits{Q} $, then $ \FreeQubits{P} \subseteq \Sigma $.
		\item[$ P = \NewCQS{c}{Q} $] By \ruleTNew, then $ c \in \names $ and $ \Sigma \vdash Q $.
			By the induction hypothesis, then $ \FreeQubits{Q} \subseteq \Sigma $.
			Since $ \FreeQubits{P} = \FreeQubits{Q} $, then $ \FreeQubits{P} \subseteq \Sigma $.
		\item[$ P = \QubitCQS{x}{Q} $] By \ruleTQbit, then $ x \in \var \setminus \Sigma $ and $ \Sigma \cup \Set{x} \vdash Q $.
			By the induction hypothesis, then $ \FreeQubits{Q} \subseteq \Sigma \cup \Set{x} $.
			Since $ \FreeQubits{P} = \FreeQubits{Q} \setminus \Set{x} $, then $ \FreeQubits{P} \subseteq \Sigma $.
		\item[$ P = \CondCQS{bv_1}{bv_2}{Q} $] By \ruleTCond, then $ bv_1 \in \binaries $ or $ \vdash \At{bv_1}{\binariesType} $, $ bv_2 \in \binaries $ or $ \vdash \At{bv_2}{\binariesType} $, and $ \Sigma \vdash Q $.
			By the induction hypothesis, then $ \FreeQubits{Q} \subseteq \Sigma $.
			Since $ \FreeQubits{P} = \FreeQubits{Q} $, then $ \FreeQubits{P} \subseteq \Sigma $.
			\qedhere
	\end{description}
\end{proof}

Well-typedness is preserved modulo structural congruence.

\begin{lem}
	If $ \Sigma \vdash P $ and $ P \equiv Q $ then $ \Sigma \vdash Q $.
	\label{lem:typingModuloStructuralCongruence}
\end{lem}

\begin{proof}
	Remember that we assume that there are no name clashes in $ P $ or $ Q $.
	The proof is then by straightforward induction on the rules of structural congruence.
\end{proof}

Well-typedness is also preserved modulo substitutions of variables for binary numbers.

\begin{lem}
	If $ \Sigma \vdash P $, $ v \in \binaries $, and $ bv \in \binaries $ or $ \vdash \At{bv}{\binariesType} $ then $ \Sigma \vdash P\Set{\Subst{bv}{v}} $.
	\label{lem:typingSubstitutionBinaries}
\end{lem}

\begin{proof}
	Assume $ \Sigma \vdash P $, $ v \in \binaries $, and $ bv \in \binaries $ or $ \vdash \At{bv}{\binariesType} $.
	We perform an induction on the structure of $ P $.
	\begin{description}
		\item[$ P = \nilCQS $] Then $ P = P\Set{\Subst{bv}{v}} $ and thus $ \Sigma \vdash P $ implies $ \Sigma \vdash P\Set{\Subst{bv}{v}} $.
		\item[$ P = \success $] Then $ P = P\Set{\Subst{bv}{v}} $ and thus $ \Sigma \vdash P $ implies $ \Sigma \vdash P\Set{\Subst{bv}{v}} $.
		\item[$ P = Q \mid R $] By \ruleTPar, then there are $ \Sigma_1, \Sigma_2 $ such that $ \Sigma_1 \vdash Q $, $ \Sigma_2 \vdash R $, $ \Sigma = \Sigma_1 \cup \Sigma_2 $, and $ \Sigma_1 \cap \Sigma_2 = \emptyset $.
			By the induction hypothesis, then $ \Sigma_1 \vdash Q\Set{\Subst{bv}{b}} $ and $ \Sigma_2 \vdash R\Set{\Subst{bv}{v}} $.
			Since $ P\Set{\Subst{bv}{v}} = Q\Set{\Subst{bv}{b}} \mid R\Set{\Subst{bv}{v}} $ and because of \ruleTPar, then $ \Sigma \vdash P\Set{\Subst{bv}{v}} $.
		\item[$ P = \InpCQS{c}{x}{Q} $] By \ruleTIn, then $ c \in \names $, $ x \in \var \setminus \Sigma $, and $ \Sigma \cup \Set{x} \vdash Q $.
			By the induction hypothesis, then $ \Sigma \cup \Set{x} \vdash Q\Set{\Subst{bv}{v}} $.
			Since $ P\Set{\Subst{b_2}{b_1}} = \InpCQS{c}{x}{\left( Q\Set{\Subst{bv}{v}} \right)} $ and because of \ruleTIn, then $ \Sigma \vdash P\Set{\Subst{bv}{v}} $.
		\item[$ P = \OutCQS{c}{x}{Q} $] By \ruleTOut, then $ c \in \names $, $ x \in \var \cap \Sigma $, and $ \Sigma \setminus \Set{x} \vdash Q $.
			By the induction hypothesis, then $ \Sigma \setminus \Set{x} \vdash Q\Set{\Subst{bv}{v}} $.
			Since $ P\Set{\Subst{bv}{v}} = \OutCQS{c}{x}{\left( Q\Set{\Subst{bv}{v}} \right)} $ and because of \ruleTOut, then $ \Sigma \vdash P\Set{\Subst{bv}{v}} $.
		\item[$ P = \UnitaryCQS{x_1, \ldots, x_n}{U}{Q} $] By \ruleTTrans, then $ x_1, \ldots, x_n \in \var \cap \Sigma $, $ \vdash \At{U}{\OpType{n}} $, and $ \Sigma \vdash Q $.
			By the induction hypothesis, then $ \Sigma \vdash Q\Set{\Subst{bv}{v}} $.
			Since $ P\Set{\Subst{bv}{v}} = \UnitaryCQS{x_1, \ldots, x_n}{U}{\left( Q\Set{\Subst{bv}{v}} \right)} $ and because of \ruleTTrans, then $ \Sigma \vdash P\Set{\Subst{bv}{v}} $.
		\item[$ P = \MeasCQS{v'}{x_1, \ldots, x_n}{Q} $] By \ruleTMsure, then $ v' \in \binaries $, $ x_1, \ldots, x_n \in \var \cap \Sigma $, and $ \Sigma \vdash Q $.
			By the induction hypothesis, then $ \Sigma \vdash Q\Set{\Subst{bv}{v}} $.
			If $ v' = v $ then $ P\Set{\Subst{bv}{v}} = P $, since $ v' $ is bound.
			Then $ \Sigma \vdash P $ implies $ \Sigma \vdash P\Set{\Subst{bv}{v}} $.
			Else if $ v' \neq v $ then $ P\Set{\Subst{bv}{v}} = \MeasCQS{v'}{x_1, \ldots, x_n}{\left( Q\Set{\Subst{bv}{v}} \right)} $.
			By \ruleTMsure, then $ \Sigma \vdash P\Set{\Subst{bv}{v}} $.
		\item[$ P = \NewCQS{c}{Q} $] By \ruleTNew, then $ c \in \names $ and $ \Sigma \vdash Q $.
			By the induction hypothesis, then $ \Sigma \vdash Q\Set{\Subst{bv}{v}} $.
			Since $ P\Set{\Subst{bv}{v}} = \NewCQS{c}{\left( Q\Set{\Subst{bv}{v}} \right)} $ and because of \ruleTNew, then $ \Sigma \vdash P\Set{\Subst{bv}{v}} $.
		\item[$ P = \QubitCQS{x}{Q} $] By \ruleTQbit, then $ x \in \var \setminus \Sigma $ and $ \Sigma \cup \Set{x} \vdash Q $.
			By the induction hypothesis, then $ \Sigma \cup \Set{x} \vdash Q\Set{\Subst{bv}{v}} $.
			Since $ P\Set{\Subst{bv}{v}} = \QubitCQS{x}{\left( Q\Set{\Subst{bv}{v}} \right)} $ and because of \ruleTQbit, then $ \Sigma \vdash P\Set{\Subst{bv}{v}} $.
		\item[$ P = \CondCQS{bv_1}{bv_2}{Q} $] By \ruleTCond, then $ bv_1 \in \binaries $ or $ \vdash \At{bv_1}{\binariesType} $, $ bv_2 \in \binaries $ or $ \vdash \At{bv_2}{\binariesType} $, and $ \Sigma \vdash Q $.
			By the induction hypothesis, then $ \Sigma \vdash Q\Set{\Subst{bv}{v}} $.
			Then $ P\Set{\Subst{bv}{v}} = \CondCQS{bv_1^*}{bv_2^*}{\left( Q\Set{\Subst{bv}{v}} \right)} $, where $ bv_1^* = bv $ if $ bv_1 = v $ and else $ bv_1^* = bv_1 $ and similarly $ bv_2^* \in \Set{ bv_2, bv } $.
			By \ruleTMsure and $ bv \in \binaries $ or $ \vdash \At{bv}{\binariesType} $, then $ \Sigma \vdash P\Set{\Subst{bv}{v}} $.
			\qedhere
	\end{description}
\end{proof}

Let $ \BoundQubits{P} $ denote the set of bound qubit (variables) in $ P $.
Well-typedness is preserved modulo adding qubit names to $ \Sigma $ that are not bound in $ P $.

\begin{lem}
	If $ \Sigma \vdash P $ and $ x \in \var \setminus \BoundQubits{P} $ then $ \Sigma \cup \Set{x} \vdash P $.
	\label{lem:weakening}
\end{lem}

\begin{proof}
	Assume $ \Sigma \vdash P $ and $ x \in \var \setminus \BoundQubits{P} $.
	The proof is by straightforward induction on the rules in Figure~\ref{fig:typingRulesCQS} to derive $ \Sigma \vdash P $.
	The only interesting cases are for \ruleTIn and \ruleTQbit.
	\begin{description}
		\item[\ruleTIn] Then $ P = \InpCQS{c}{y}{Q} $, $ c \in \names $, $ y \in \var \setminus \Sigma $, and $ \Sigma \cup \Set{y} \vdash Q $.
			Since $ x \in \var \setminus \BoundQubits{P} $, $ x \neq y $.
			By the induction hypothesis, then $ \Sigma \cup \Set{x, y} \vdash Q $.
			By \ruleTIn, then $ \Sigma \cup \Set{x} \vdash P $.
	\end{description}
	The case of \ruleTQbit is similar.
	Note that for \ruleTPar it does not matter to which parallel component we give the additional $ x $.
\end{proof}

Well-typedness is also preserved modulo removing qubit names from $ \Sigma $ that are not free in $ P $.

\begin{lem}
	If $ \Sigma \vdash P $ and $ x \in \var \setminus \FreeQubits{P} $ then $ \Sigma \setminus \Set{x} \vdash P $.
	\label{lem:strengthening}
\end{lem}

\begin{proof}
	Assume $ \Sigma \vdash P $ and $ x \in \var \setminus \FreeQubits{P} $.
	The proof is by straightforward induction on the rules in Figure~\ref{fig:typingRulesCQS} to derive $ \Sigma \vdash P $.
	The only interesting case is for \ruleTOut.
	\begin{description}
		\item[\ruleTOut] Then $ P = \OutCQS{c}{y}{Q} $, $ c \in \names $, $ y \in \var \cap \Sigma $, and $ \Sigma \setminus \Set{y} \vdash Q $.
			Since $ x \in \var \setminus \FreeQubits{P} $, $ x \neq y $.
			By the induction hypothesis, then $ \Sigma \setminus \Set{x, y} \vdash Q $.
			By \ruleTOut, then $ \Sigma \setminus \Set{x} \vdash P $.
			\qedhere
	\end{description}
\end{proof}

Well-typedness is preserved modulo substitutions of qubit names.
To prove this property we have to rely on the condition that substitutions on qubit names are not allowed to rename two qubits to the same qubit (see Section~\ref{sec:processCalculi}).
We use $ \mathsf{s} $ to denote substitutions on qubits of the form $ \Set{\Subst{q_1}{x_1}, \ldots, \Subst{q_n}{x_n}} $.
Let $ \Sigma\mathsf{s} $ be the result of applying the substitution $ \mathsf{s} $ simultaneously on all qubit names in the set $ \Sigma $.
Similarly, $ \tilde{x}\mathsf{s} $ is the result of applying the substitution $ \mathsf{s} $ simultaneously on all qubit names in $ \tilde{x} $.
Moreover, let $ \FreeQubits{\mathsf{s}} $ return all qubit names in the substitution $ \mathsf{s} $, \ie $ \FreeQubits{\Set{\Subst{q_1}{x_1}, \ldots, \Subst{q_n}{x_n}}} = \Set{ x_1, q_1, \ldots, x_n, q_n } $.
As usual we require for $ \mathsf{s} = \Set{\Subst{q_1}{x_1}, \ldots, \Subst{q_n}{x_n}} $ that the $ x_1, \ldots, x_n $ are pairwise distinct.
For the next Lemma we additionally explicitly require that also the $ q_1, \ldots, q_n $ are pairwise distinct.

\begin{lem}
	If $ \Sigma \vdash P $, $ \mathsf{s} = \Set{\Subst{q_1}{x_1}, \ldots, \Subst{q_n}{x_n}} $, $ \FreeQubits{\mathsf{s}} \in \var \setminus \BoundQubits{P} $, and $ q_1, \ldots, q_n $ are pairwise distinct, then $ \Sigma\mathsf{s} \vdash P\mathsf{s} $.
	\label{lem:typingSubstitutionQubits}
\end{lem}

\begin{proof}
	Assume $ \Sigma \vdash P $, $ \mathsf{s} = \Set{\Subst{q_1}{x_1}, \ldots, \Subst{q_n}{x_n}} $, $ \FreeQubits{\mathsf{s}} \in \var \setminus \BoundQubits{P} $, and $ q_1, \ldots, q_n $ are pairwise distinct.
	We perform an induction on the structure of $ P $.
	\begin{description}
		\item[$ P = \nilCQS $] Then $ P = P\mathsf{s} $.
			By \ruleTNil, then $ \vdash P\mathsf{s} $.
			By applying Lemma~\ref{lem:weakening} potentially several times, then $ \Sigma\mathsf{s} \vdash P\mathsf{s} $.
		\item[$ P = \success $] Then $ P = P\mathsf{s} $.
			By \ruleTSuc, then $ \vdash P\mathsf{s} $.
			By applying Lemma~\ref{lem:weakening} potentially several times, then $ \Sigma\mathsf{s} \vdash P\mathsf{s} $.
		\item[$ P = Q \mid R $] By \ruleTPar, then there are $ \Sigma_1, \Sigma_2 $ such that $ \Sigma_1 \vdash Q $, $ \Sigma_2 \vdash R $, $ \Sigma = \Sigma_1 \cup \Sigma_2 $, and $ \Sigma_1 \cap \Sigma_2 = \emptyset $.
			By Lemma~\ref{lem:typingFreeQubitsCQS}, then $ \FreeQubits{Q} \subseteq \Sigma_1 $ and $ \FreeQubits{R} \subseteq \Sigma_2 $.
			Then we can split $ \mathsf{s} $ into $ \mathsf{s}_1 = \Set{\Subst{q_{1, 1}}{x_{1, 1}}, \ldots, \Subst{q_{1, n_1}}{x_{1, n_1}}} $ and $ \mathsf{s}_2 = \Set{\Subst{q_{2, 1}}{x_{2, 1}}, \ldots, \Subst{q_{2, n_2}}{x_{2, n_2}}} $, \ie $ \mathsf{s} = \mathsf{s}_1 \cup \mathsf{s}_2 $, and $ x_{1, 1}, \ldots, x_{1, n_1} \notin \FreeQubits{R} $, $ x_{2, 1}, \ldots, x_{2, n_2} \notin \FreeQubits{Q} $, and $ \Set{x_{1, 1}, \ldots, x_{1, n_1}} \cap \Set{x_{2, 1}, \ldots, x_{2, n_2}} = \emptyset $.
			Then $ P\mathsf{s} = Q\mathsf{s}_1 \mid R\mathsf{s}_2 $.
			Since $ \BoundQubits{P} = \BoundQubits{Q} \cup \BoundQubits{R} $, we have $ \FreeQubits{\mathsf{s}_1} \notin \BoundQubits{Q} $ and $ \FreeQubits{\mathsf{s}_2} \notin \BoundQubits{R} $.
			By the induction hypothesis, then $ \Sigma_1\mathsf{s}_1 \vdash Q\mathsf{s}_1 $ and $ \Sigma_2\mathsf{s}_2 \vdash R\mathsf{s}_2 $.
			Because the $ q_1, \ldots, q_n $ are pairwise distinct and $ \Sigma_1 \cap \Sigma_2 = \emptyset $ and since substitutions on qubits cannot rename two qubits to the same qubit, then $ \left( \Sigma_1\mathsf{s}_1 \right) \cap \left( \Sigma_2\mathsf{s}_2 \right) = \emptyset $ and $ \left( \Sigma_1\mathsf{s}_1 \right) \cup \left( \Sigma_2\mathsf{s}_2 \right) = \Sigma\mathsf{s} $.
			By \ruleTPar, then $ \Sigma\mathsf{s} \vdash P\mathsf{s} $.
		\item[$ P = \InpCQS{c}{x}{Q} $] By \ruleTIn, then $ c \in \names $, $ x \in \var \setminus \Sigma $, and $ \Sigma \cup \Set{x} \vdash Q $.
			Note that $ \BoundQubits{P} = \BoundQubits{Q} \cup \Set{x} $.
			By the induction hypothesis, then $ \left( \Sigma \cup \Set{x} \right)\mathsf{s} \vdash Q\mathsf{s} $.
			Since $ \FreeQubits{\mathsf{s}} \notin \BoundQubits{P} $, we have $ x \notin \Set{x_1, \ldots, x_n} $.
			Then $ P\mathsf{s} = \InpCQS{c}{x}{\left( Q\mathsf{s} \right)} $ and $ \left( \Sigma \cup \Set{x} \right)\mathsf{s} = \Sigma\mathsf{s} \cup \Set{x} $.
			By \ruleTIn, then $ \Sigma\mathsf{s} \vdash P\mathsf{s} $.
		\item[$ P = \OutCQS{c}{x}{Q} $] By \ruleTOut, then $ c \in \names $, $ x \in \var \cap \Sigma $, and $ \Sigma \setminus \Set{x} \vdash Q $.
			Note that $ \BoundQubits{P} = \BoundQubits{Q} $.
			If $ x \notin \Set{x_1, \ldots, x_n} $, then $ P\mathsf{s} = \OutCQS{c}{x}{\left( Q\mathsf{s} \right)} $.
			Remember that substitutions on qubits are not allowed to rename two qubits to the same qubit.
			Then either (1)~$ x \notin \Set{q_1, \ldots, q_n} $ or (2)~$ x = q_i \in \Set{q_1, \ldots, q_n} $ but $ x_i \notin \FreeQubits{Q} $.
			\begin{enumerate}[(1)]
				\item By the induction hypothesis, then $ \left( \Sigma \setminus \Set{x} \right) \mathsf{s} \vdash Q\mathsf{s} $ and $ \left( \Sigma \setminus \Set{x} \right) \mathsf{s} = \Sigma\mathsf{s} \setminus \Set{x} $.
					By \ruleTOut, then $ \Sigma\mathsf{s} \vdash P\mathsf{s} $.
				\item In this case, we can ignore the substitution $ \Subst{q_i}{x_i} $, \ie $ \mathsf{s}' = \mathsf{s} \setminus \Set{\Subst{q_i}{x_i}} $ and $ Q\mathsf{s} = Q\mathsf{s}' $ as well as $ P\mathsf{s} = P\mathsf{s}' $.
					By the induction hypothesis, then $ \left( \Sigma \setminus \Set{x} \right) \mathsf{s}' \vdash Q\mathsf{s}' $ and we have that $ \left( \Sigma \setminus \Set{x} \right) \mathsf{s}' = \Sigma\mathsf{s}' \setminus \Set{x} $.
					By \ruleTOut, then $ \Sigma\mathsf{s}' \vdash P\mathsf{s}' $.
					If $ x_i \notin \Sigma $ then also $ \Sigma\mathsf{s} \vdash P\mathsf{s} $.
					Else if $ x_i \in \Sigma $, then $ x_i \in \Sigma\mathsf{s}' $.
					By Lemma~\ref{lem:strengthening} and since $ x_i \notin \FreeQubits{Q} $, then $ \Sigma\mathsf{s}' \setminus \Set{x_i} \vdash P\mathsf{s}' $.
					By Lemma~\ref{lem:weakening} and since $ q_i \notin \BoundQubits{P} $, then $ \left( \Sigma\mathsf{s}' \setminus \Set{x_i} \right) \cup \Set{q_i} \vdash P\mathsf{s}' $.
					If $ x_i \notin \Set{q_1, \ldots, q_{i-1}, q_{i+1}, \ldots, q_n} $ then $ \Sigma\mathsf{s} \vdash P\mathsf{s} $.
					Else we apply once more Lemma~\ref{lem:weakening} to add the respective $ q_j $ and have again $ \Sigma\mathsf{s} \vdash P\mathsf{s} $.
			\end{enumerate}
			Else $ x = x_i \in \Set{x_1, \ldots, x_n} $.
			Then $ P\mathsf{s} = \OutCQS{c}{q_i}{\left( Q\mathsf{s} \right)} $.
			By Lemma~\ref{lem:typingFreeQubitsCQS}, $ \Sigma \setminus \Set{x} \vdash Q $ implies $ x \notin \FreeQubits{Q} $.
			Then we can ignore the substitution $ \Subst{q_i}{x_i} $ for $ Q $, \ie $ \mathsf{s}' = \mathsf{s} \setminus \Set{\Subst{q_i}{x_i}} $ and $ Q\mathsf{s} = Q\mathsf{s}' $.
			By the induction hypothesis, then $ \left( \Sigma \setminus \Set{x} \right) \mathsf{s}' \vdash Q\mathsf{s}' $.
			Since the substitution cannot rename two qubits to the same qubit, then $ \left( \Sigma \setminus \Set{x} \right) \mathsf{s}' = \left( \Sigma\mathsf{s} \right) \setminus \Set{q_i} $.
			By \ruleTOut, then $ \Sigma\mathsf{s} \vdash P\mathsf{s} $.
		\item[$ P = \UnitaryCQS{\tilde{x}}{U}{Q} $] By \ruleTTrans, then $ \tilde{x} \in \var \cap \Sigma $, $ \vdash \At{U}{\OpType{n}} $, and $ \Sigma \vdash Q $.
			By the induction hypothesis, then $ \Sigma\mathsf{s} \vdash Q\mathsf{s} $.
			Since $ P\mathsf{s} = \UnitaryCQS{\tilde{x}\mathsf{s}}{U}{\left( Q\mathsf{s} \right)} $ and because of \ruleTTrans, then $ \Sigma\mathsf{s} \vdash P\mathsf{s} $.
		\item[$ P = \MeasCQS{v'}{\tilde{x}}{Q} $] By \ruleTMsure, then $ v' \in \binaries $, $ \tilde{x} \in \var \cap \Sigma $, and $ \Sigma \vdash Q $.
			By the induction hypothesis, then $ \Sigma\mathsf{s} \vdash Q\mathsf{s} $.
			Since $ P\mathsf{s} = \MeasCQS{v'}{\tilde{x}\mathsf{s}}{\left( Q\mathsf{s} \right)} $ and because of \ruleTMsure, then $ \Sigma\mathsf{s} \vdash P\mathsf{s} $.
		\item[$ P = \NewCQS{c}{Q} $] By \ruleTNew, then $ c \in \names $ and $ \Sigma \vdash Q $.
			By the induction hypothesis, then $ \Sigma\mathsf{s} \vdash Q\mathsf{s} $.
			Since $ P\mathsf{s} = \NewCQS{c}{\left( Q\mathsf{s} \right)} $ and because of \ruleTNew, then $ \Sigma\mathsf{s} \vdash P\mathsf{s} $.
		\item[$ P = \QubitCQS{x}{Q} $] By \ruleTQbit, then $ x \in \var \setminus \Sigma $ and $ \Sigma \cup \Set{x} \vdash Q $.
			By the induction hypothesis, then $ \left( \Sigma \cup \Set{x} \right) \mathsf{s} \vdash Q\mathsf{s} $.
			Since $ \FreeQubits{\mathsf{s}} \notin \BoundQubits{P} $, $ x \notin \FreeQubits{\mathsf{s}} $ and thus $ \left( \Sigma \cup \Set{x} \right) \mathsf{s} = \Sigma\mathsf{s} \cup \Set{x} $.
			Since $ P\mathsf{s} = \QubitCQS{x}{\left( Q\mathsf{s} \right)} $ and because of \ruleTQbit, then $ \Sigma\mathsf{s} \vdash P\mathsf{s} $.
		\item[$ P = \CondCQS{bv_1}{bv_2}{Q} $] By \ruleTCond, then $ bv_1 \in \binaries $ or $ \vdash \At{bv_1}{\binariesType} $, $ bv_2 \in \binaries $ or $ \vdash \At{bv_2}{\binariesType} $, and $ \Sigma \vdash Q $.
			By the induction hypothesis, then $ \Sigma\mathsf{s} \vdash Q\mathsf{s} $.
			Since $ P\mathsf{s} = \CondCQS{bv_1}{bv_2}{\left( Q\mathsf{s} \right)} $ and because of \ruleTMsure, then $ \Sigma\mathsf{s} \vdash P\mathsf{s} $.
			\qedhere
	\end{description}
\end{proof}

Well-typedness is also preserved modulo substitutions of channel names.
Let $ \BoundChan{P} $ return the set of bound names in $ P $.

\begin{lem}
	If $ \Sigma \vdash P $ and $ a, c \in \names \setminus \BoundChan{P} $ then $ \Sigma \vdash P\Set{\Subst{a}{c}} $.
	\label{lem:typingSubstitutionNames}
\end{lem}

\begin{proof}
	Assume $ \Sigma \vdash P $ and $ a, c \in \names \setminus \BoundChan{P} $.
	We perform an induction on the structure of $ P $.
	\begin{description}
		\item[$ P = \nilCQS $] Then $ P = P\Set{\Subst{a}{c}} $ and thus $ \Sigma \vdash P $ implies $ \Sigma \vdash P\Set{\Subst{a}{c}} $.
		\item[$ P = \success $] Then $ P = P\Set{\Subst{a}{c}} $ and thus $ \Sigma \vdash P $ implies $ \Sigma \vdash P\Set{\Subst{a}{c}} $.
		\item[$ P = Q \mid R $] By \ruleTPar, then there are $ \Sigma_1, \Sigma_2 $ such that $ \Sigma_1 \vdash Q $, $ \Sigma_2 \vdash R $, $ \Sigma = \Sigma_1 \cup \Sigma_2 $, and $ \Sigma_1 \cap \Sigma_2 = \emptyset $.
			By the induction hypothesis, then $ \Sigma_1 \vdash Q\Set{\Subst{a}{c}} $ and $ \Sigma_2 \vdash R\Set{\Subst{a}{c}} $.
			Since $ P\Set{\Subst{a}{c}} = Q\Set{\Subst{a}{c}} \mid R\Set{\Subst{a}{c}} $ and because of \ruleTPar, then $ \Sigma \vdash P\Set{\Subst{a}{c}} $.
		\item[$ P = \InpCQS{d}{x}{Q} $] By \ruleTIn, then $ d \in \names $, $ x \in \var \setminus \Sigma $, and $ \Sigma \cup \Set{x} \vdash Q $.
			By the induction hypothesis, then $ \Sigma \cup \Set{x} \vdash Q\Set{\Subst{a}{c}} $.
			Since $ P\Set{\Subst{a}{c}} = \InpCQS{d^*}{x}{\left( Q\Set{\Subst{a}{c}} \right)} $ with $ d^* \in \Set{a, d} $ and because of \ruleTIn, then $ \Sigma \vdash P\Set{\Subst{a}{c}} $.
		\item[$ P = \OutCQS{d}{x}{Q} $] By \ruleTOut, then $ d \in \names $, $ x \in \var \cap \Sigma $, and $ \Sigma \setminus \Set{x} \vdash Q $.
			By the induction hypothesis, then $ \Sigma \setminus \Set{x} \vdash Q\Set{\Subst{a}{c}} $.
			Since $ P\Set{\Subst{a}{c}} = \OutCQS{d^*}{x}{\left( Q\Set{\Subst{a}{c}} \right)} $ with $ d^* \in \Set{ a, d } $ and because of \ruleTOut, then $ \Sigma \vdash P\Set{\Subst{a}{d}} $.
		\item[$ P = \UnitaryCQS{x_1, \ldots, x_n}{U}{Q} $] By \ruleTTrans, then $ x_1, \ldots, x_n \in \var \cap \Sigma $, $ \vdash \At{U}{\OpType{n}} $, and $ \Sigma \vdash Q $.
			By the induction hypothesis, then we have $ \Sigma \vdash Q\Set{\Subst{a}{c}} $.
			Since $ P\Set{\Subst{a}{c}} = \UnitaryCQS{x_1, \ldots, x_n}{U}{\left( Q\Set{\Subst{a}{c}} \right)} $ and because of \ruleTTrans, then $ \Sigma \vdash P\Set{\Subst{a}{c}} $.
		\item[$ P = \MeasCQS{v}{x_1, \ldots, x_n}{Q} $] By \ruleTMsure, then $ v \in \binaries $, $ x_1, \ldots, x_n \in \var \cap \Sigma $, and $ \Sigma \vdash Q $.
			By the induction hypothesis, then we have $ \Sigma \vdash Q\Set{\Subst{a}{c}} $.
			Since $ P\Set{\Subst{a}{c}} = \MeasCQS{v}{x_1, \ldots, x_n}{\left( Q\Set{\Subst{a}{c}} \right)} $ and because of \ruleTMsure, then $ \Sigma \vdash P\Set{\Subst{a}{c}} $.
		\item[$ P = \NewCQS{d}{Q} $] By \ruleTNew, then $ d \in \names $ and $ \Sigma \vdash Q $.
			By the induction hypothesis, then $ \Sigma \vdash Q\Set{\Subst{a}{c}} $.
			Since $ a, c \notin \BoundChan{P} $, $ d \notin \Set{a, c} $.
			Then $ P\Set{\Subst{a}{c}} = \NewCQS{d}{\left( Q\Set{\Subst{a}{c}} \right)} $.
			By \ruleTNew, then $ \Sigma \vdash P\Set{\Subst{a}{c}} $.
		\item[$ P = \QubitCQS{x}{Q} $] By \ruleTQbit, then $ x \in \var \setminus \Sigma $ and $ \Sigma \cup \Set{x} \vdash Q $.
			By the induction hypothesis, then $ \Sigma \cup \Set{x} \vdash Q\Set{\Subst{a}{c}} $.
			Since $ P\Set{\Subst{a}{c}} = \QubitCQS{x}{\left( Q\Set{\Subst{a}{c}} \right)} $ and because of \ruleTQbit, then $ \Sigma \vdash P\Set{\Subst{a}{c}} $.
		\item[$ P = \CondCQS{bv_1}{bv_2}{Q} $] By \ruleTCond, then $ bv_1 \in \binaries $ or $ \vdash \At{bv_1}{\binariesType} $, $ bv_2 \in \binaries $ or $ \vdash \At{bv_2}{\binariesType} $, and $ \Sigma \vdash Q $.
			By the induction hypothesis, then $ \Sigma \vdash Q\Set{\Subst{a}{c}} $.
			Since $ P\Set{\Subst{a}{c}} = \CondCQS{bv_1}{bv_2}{\left( Q\Set{\Subst{a}{c}} \right)} $ and because of \ruleTMsure, then $ \Sigma \vdash P\Set{\Subst{a}{c}} $.
			\qedhere
	\end{description}
\end{proof}

Lemma~\ref{lem:preservationCQS} states:
\begin{quotation}
	If $ \Sigma \vdash P $ and $ \ConfigCQS{\sigma}{\phi}{P} \step \boxplus_{0 \leq i < 2^r} p_i \bullet \ConfigCQS{\sigma_i'}{\phi'}{P_i} $ or if $ \Sigma \vdash P_k $ for all $ 0 \leq k < 2^t $ and $ \boxplus_{0 \leq k < 2^t} p_k' \bullet \ConfigCQS{\sigma}{\phi}{P_k'} \step \boxplus_{0 \leq i < 2^r} p_i \bullet \ConfigCQS{\sigma_i'}{\phi'}{P_i} $ then there is some $ \Sigma' \in \Set{\Sigma, \Sigma \cup \Set{q_n} } $ for some fresh $ q_n $ such that $ \Sigma' \vdash P_i $ for all $ 0 \leq i < 2^{r} $.
\end{quotation}

\begin{proof}[Proof of Lemma~\ref{lem:preservationCQS}]
	Assume $ \Sigma \vdash P $ and $ \ConfigCQS{\sigma}{\phi}{P} \step \boxplus_{0 \leq i < 2^r} p_i \bullet \ConfigCQS{\sigma_i'}{\phi'}{P_i} $ or if $ \Sigma \vdash P_k $ for all $ 0 \leq k < 2^t $ and $ \boxplus_{0 \leq k < 2^t} p_k \bullet \ConfigCQS{\sigma}{\phi}{P_k} \step \boxplus_{0 \leq i < 2^r} p_i \bullet \ConfigCQS{\sigma_i'}{\phi'}{P_i} $.
	We perform an induction on the reduction rules in Figure~\ref{fig:semanticsCQS}.
	\begin{description}
		\item[\ruleRMeasureCQS] Then $ P = \MeasCQS{v}{q_1, \ldots, q_{r - 1}}{Q} $ and all $ P_i = Q\Set{\Subst{\Binary{i}}{v}} $ for all $ 0 \leq i < 2^r $.
			Fix some $ i $ with $ 0 \leq i < 2^r $.
			By \ruleTMsure, then $ v \in \binaries $ and $ \Sigma \vdash Q $.
			By \ruleTBin, $ \vdash \At{\Binary{i}}{\binariesType} $.
			By Lemma~\ref{lem:typingSubstitutionBinaries}, then $ \Sigma \vdash P_i $.
		\item[\ruleRTransCQS] Then $ P = \UnitaryCQS{q_0, \ldots, q_{r' - 1}}{U}{Q} $, $ r = 0 $, there is just one $ i $ such that $ 0 \leq i < 2^r $, and $ P_i = P_0 = Q $.
			By \ruleTTrans, then $ \Sigma \vdash Q $, \ie $ \Sigma \vdash P_i $.
		\item[\ruleRPermCQS] Then $ r = 0 $, there is just one $ i $ such that $ 0 \leq i < 2^r $, and $ P_i = P_0 = P\pi $, where $ \pi $ is a permutation of qubit names that are free, \ie $ \FreeQubits{\pi} \subseteq \FreeQubits{P} $.
			By Lemma~\ref{lem:typingFreeQubitsCQS}, then $ \FreeQubits{\pi} \subseteq \Sigma $.
			Then $ \Sigma\pi = \Sigma $.
			By Lemma~\ref{lem:typingSubstitutionQubits}, then $ \Sigma \vdash P $ implies $ \Sigma \vdash P_i $.
		\item[\ruleRProbCQS] Then $ P_j' = Q\Set{\Subst{\Binary{j}}{v}} $, $ r = 0 $, there is just one $ i $ such that $ 0 \leq i < 2^r $, and $ P_i = P_0 = Q\Set{\Subst{\Binary{j}}{v}} = P_j' $ for some $ 0 \leq j < 2^t $.
			Hence, $ \Sigma \vdash P_j' $ implies $ \Sigma \vdash P_i $.
		\item[\ruleRNewCQS] Then $ P = \NewCQS{c}{Q} $, $ r = 0 $, there is just one $ i $ such that $ 0 \leq i < 2^r $, and $ P_i = P_0 = Q\Set{\Subst{a}{c}} $, where $ a $ is fresh.
			By \ruleTNew, then $ c \in \names $ and $ \Sigma \vdash Q $.
			By Lemma~\ref{lem:typingSubstitutionNames}, then $ \Sigma \vdash P_i $.
		\item[\ruleRQbitCQS] Then $ P = \QubitCQS{x}{Q} $, $ r = 0 $, there is just one $ i $ such that $ 0 \leq i < 2^r $, and $ P_i = P_0 = Q\Set{\Subst{q_n}{x}} $ for some fresh $ q_n $.
			By \ruleTQbit, $ x \in \var \setminus \Sigma $ and $ \Sigma \cup \Set{x} \vdash Q $.
			Because we assume the absence of name clashes and since no qubit variable has a name of the form $ q_j $, $ x \notin \BoundQubits{Q} $.
			Since $ q_n $ is fresh, $ q_n \notin \BoundQubits{Q}  $.
			Note that $ \Sigma \subseteq \left( \Sigma \cup \Set{x}\right) \Set{\Subst{q_n}{x}} $.
			By Lemma~\ref{lem:typingSubstitutionQubits}, then $ \left( \Sigma \cup \Set{x}\right) \Set{\Subst{q_n}{x}} \vdash P_i $.
		\item[\ruleRCommCQS] Then $ P = \OutCQS{c}{q}{Q} \mid \InpCQS{c}{x}{R} $, $ r = 0 $, there is just one $ i $ such that $ 0 \leq i < 2^r $, and $ P_i = P_0 = Q \mid R\Set{\Subst{q}{x}} $.
			By \ruleTPar, then there are $ \Sigma_1, \Sigma_2 $ such that $ \Sigma_1 \vdash \OutCQS{c}{q}{Q} $, $ \Sigma_2 \vdash \InpCQS{c}{x}{R} $, $ \Sigma_1 \cap \Sigma_2 = \emptyset $, and $ \Sigma_1 \cup \Sigma_2 = \Sigma $.
			By \ruleTOut, then $ c \in \names $, $ q \in \var \cap \Sigma_1 $, and $ \Sigma_1 \setminus \Set{q} \vdash Q $.
			By \ruleTIn, then $ x \in \var \setminus \Sigma_2 $ and $ \Sigma_2 \cup \Set{x} \vdash R $.
			Since $ q \in \Sigma_1 $ and $ \Sigma_1 \cap \Sigma_2 = \emptyset $, $ q \notin \Sigma_2 $.
			Because we assume that there are no name clashes for $ P $, $ x, q \notin \BoundQubits{R} $.
			By Lemma~\ref{lem:typingSubstitutionQubits}, then $ \left( \Sigma_2 \cup \Set{x} \right) \Set{\Subst{q}{x}} \vdash R\Set{\Subst{q}{x}} $.
			Since $ x \notin \Sigma_2 $, $ \left( \Sigma_2 \cup \Set{x} \right) \Set{\Subst{q}{x}} = \Sigma_2 \cup \Set{q} $.
			Note that $ \left( \Sigma_1 \setminus \Set{q} \right) \cap \left( \Sigma_2 \cup \Set{q} \right) = \emptyset $ and $ \left( \Sigma_1 \setminus \Set{q} \right) \cup \left( \Sigma_2 \cup \Set{q} \right) = \Sigma $.
			By \ruleTPar, then $ \Sigma \vdash P_i $.
		\item[\ruleRParCQS] Then $ P = Q \mid R $, $ \ConfigCQS{\sigma}{\phi}{Q} \step \boxplus_{0 \leq i < 2^r} p_i \ConfigCQS{\sigma_i'}{\phi}{Q_i} $, and $ P_i = Q_i \mid R $ for all $ 0 \leq i < 2^r $.
			Fix some $ i $ with $ 0 \leq i < 2^r $.
			By \ruleTPar, then there are $ \Sigma_1, \Sigma_2 $ such that $ \Sigma_1 \vdash Q $, $ \Sigma_2 \vdash R $, $ \Sigma_1 \cap \Sigma_2 = \emptyset $, and $ \Sigma_1 \cup \Sigma_2 = \Sigma $.
			By the induction hypothesis, then there is some $ \Sigma_1' \in \Set{\Sigma_1, \Sigma_1' \cup \Set{q}} $ for some fresh $ q $ such that $ \Sigma_1' \vdash Q_i $.
			Since $ q $ is fresh, $ \Sigma_1' \cap \Sigma_2 = \emptyset $.
			By \ruleTPar, then $ \Sigma' \vdash P_i $, where $ \Sigma' \in \Set{\Sigma, \Sigma' \cup \Set{q}} $.
		\item[\ruleRCongCQS] Then $ P \equiv Q $, $ \ConfigCQS{\sigma}{\phi}{Q} \step \boxplus_{0 \leq i < 2^r} p_i \ConfigCQS{\sigma_i'}{\phi}{Q_i'} $, and $ P_i \equiv Q_i' $ for all $ 0 \leq i < 2^r $.
			Fix some $ i $ with $ 0 \leq i < 2^r $.
			By Lemma~\ref{lem:typingModuloStructuralCongruence}, then $ \Sigma \vdash Q $.
			By the induction hypothesis, then there is some $ \Sigma' \in \Set{\Sigma, \Sigma' \cup \Set{q}} $ for some fresh $ q $ such that $ \Sigma' \vdash Q_i' $.
			By Lemma~\ref{lem:typingModuloStructuralCongruence}, then $ \Sigma' \vdash P_i $.
		\item[\ruleRCondCQS] Then $ P = \CondCQS{b}{b'}{Q} $, $ b = b' $, $ r = 0 $, there is just one $ i $ such that $ 0 \leq i < 2^r $, and $ P_i = P_0 = Q $.
			By \ruleTCond, then $ \Sigma \vdash P_i $.
			\qedhere
	\end{description}
\end{proof}

Finally, Lemma~\ref{lem:typingCQSUniqueOwnership} states:
\begin{quotation}
	If $ \Sigma \vdash P \mid Q $ then $ \FreeQubits{P} \cap \FreeQubits{Q} = \emptyset $.
\end{quotation}

\begin{proof}[Proof of Lemma~\ref{lem:typingCQSUniqueOwnership}]
	Assume $ \Sigma \vdash P \mid Q $.
	By \ruleTPar, then there are $ \Sigma_1, \Sigma_2 $ such that $ \Sigma_1 \vdash P $, $ \Sigma_2 \vdash Q $, $ \Sigma_1 \cap \Sigma_2 = \emptyset $, and $ \Sigma_1 \cup \Sigma_2 = \Sigma $.
	By Lemma~\ref{lem:typingFreeQubitsCQS}, then $ \FreeQubits{P} \subseteq \Sigma_1 $ and $ \FreeQubits{Q} \subseteq \Sigma_2 $.
	Since $ \Sigma_1 \cap \Sigma_2 = \emptyset $, then $ \FreeQubits{P} \cap \FreeQubits{Q} = \emptyset $.
\end{proof}

\end{appendix}

\end{document}